%% file: main_file.tex
\documentclass[11pt]{article}
\usepackage{hyperref}
\input{macros.tex}
\usepackage{microtype}
\PassOptionsToPackage{numbers, compress, sort}{natbib}
\usepackage{natbib}
\input{bib_macros}

\title{Reusing Samples in Variance Reduction}
\author{Yujia Jin \\\texttt{yujiajin@stanford.edu}\\ Stanford University \and Ishani Karmarkar \\\texttt{ishanik@stanford.edu} \\ Stanford University \and Aaron Sidford\\\texttt{sidford@stanford.edu}\\Stanford University \and Jiayi Wang \\\texttt{jyw@stanford.edu} \\ Stanford University }
\begin{document}
\maketitle

\begin{abstract}
    We provide a general framework to improve trade-offs between the number of \emph{full batch} and \emph{sample} queries used to solve structured optimization problems. Our results apply to a broad class of randomized optimization algorithms that iteratively solve sub-problems to high accuracy. We show that such algorithms can be modified to \emph{reuse the randomness} used to query the input across sub-problems. 
    Consequently, we improve the trade-off between the number of gradient (full batch) and individual function (sample) queries for finite sum minimization, the number of matrix-vector multiplies (full batch) and random row (sample) queries for top-eigenvector computation, and the number of matrix-vector multiplies with the transition matrix (full batch) and generative model (sample) queries for optimizing Markov Decision Processes. To facilitate our analysis we introduce the notion of \emph{pseudo-independent algorithms}, a generalization of pseudo-deterministic algorithms \citep{Gat2011ProbabilisticSA}, that quantifies how independent the output of a randomized algorithm is from a randomness source.
\end{abstract}

\newpage
\tableofcontents
\newpage

\input{intro}
\input{pseudoindependence}
\input{proof}
\input{finite_sum}
\input{mdp}
\input{minimax}
\input{nouniform_smoothness}
\input{topev}
\input{conclusion}

\section*{Acknowledgements}

We thank anonymous reviewers for their feedback. Yujia Jin and Ishani Karmarkar were funded in part by NSF CAREER Award CCF-1844855, NSF Grant CCF-1955039, and a PayPal research award. Aaron Sidford was funded in part by a Microsoft Research Faculty Fellowship, NSF CAREER Award CCF-1844855, NSF Grant CCF1955039, and a PayPal research award. Yujia Jin’s contributions to the project occurred while she was a graduate student at Stanford.

\bibliographystyle{plainnat}
\bibliography{ref}

\appendix
\input{stable_proof}
\input{appendix_helpers}

\end{document}

%% file: macros.tex
\usepackage{amsmath,amssymb,amsthm}
\usepackage[margin=1in]{geometry}
\usepackage{graphicx}
\usepackage{color}
\usepackage{bbm}
\usepackage{booktabs}
\usepackage{breakcites}
\usepackage{placeins}
\usepackage{float}
\usepackage{bm}
\usepackage{thm-restate}
\usepackage{color-edits}
\usepackage{enumitem}
\usepackage{float}
\usepackage{dsfont}
\usepackage{xspace}
\usepackage[algo2e,linesnumbered,ruled,vlined]{algorithm2e}
\usepackage{algorithmic}
\usepackage{subcaption}
\usepackage{array}
\usepackage[dvipsnames]{xcolor}
\usepackage{cleveref}
\usepackage{makecell}

\theoremstyle{plain}
\newtheorem{theorem}{Theorem}[section]
\newtheorem{lemma}[theorem]{Lemma}

\newtheorem{fact}[theorem]{Fact}

\theoremstyle{definition}
\newtheorem{definition}[theorem]{Definition}

\theoremstyle{remark}

\newcommand{\paren}[1]{\left(#1 \right )}

\newcommand{\Brac}[1]{\left[#1\right]}

\newcommand{\Set}[1]{\left\{#1\right\}}
\newcommand{\set}[1]{\left\{#1\right\}}

\newcommand{\abs}[1]{\left\lvert#1\right\rvert}

\newcommand{\ceil}[1]{\lceil #1 \rceil}
\newcommand{\floor}[1]{\lfloor #1 \rfloor}

\newcommand{\norm}[1]{\left\lVert#1\right\rVert}

\newcommand{\nnz}[1]{\mathrm{nnz}(#1)}

\newcommand{\defeq}{:=}

\newcommand{\normInline}[1]{\lVert#1\rVert}

\newcommand{\R}{\mathbb R}

\newcommand{\Esymb}{\mathbb{E}}
\newcommand{\Psymb}{\mathbb{P}}

\DeclareMathOperator*{\E}{\Esymb}

\DeclareMathOperator*{\ProbOp}{\Psymb}

\newcommand{\prob}[1]{\ProbOp\Set{#1}}

\newcommand{\Prob}{\mathbb{P}}
\newcommand{\probSub}[2]{\mathbb{P}_{#2}\Set{#1}}

\newcommand{\indicator}[1]{\mathds{1}{\set{#1}}}

\newcommand{\cA}{\mathcal A}

\newcommand{\cD}{\mathcal D}

\newcommand{\cF}{\mathcal F}

\newcommand{\cM}{\mathcal M}

\newcommand{\cQ}{\mathcal Q}

\newcommand{\cS}{\mathcal S}
\newcommand{\cT}{\mathcal T}

\newcommand{\cX}{\mathcal X}
\newcommand{\cY}{\mathcal Y}
\newcommand{\cZ}{\mathcal Z}

\DeclareMathOperator*{\argmin}{argmin} 
\DeclareMathOperator*{\argmax}{argmax} 
\newcommand{\poly}{\mathrm{poly}}

\newcommand{\supp}{{\mathsf{supp}}}

\allowdisplaybreaks
\newcommand{\bv}{\bm{v}}
\newcommand{\bz}{\bm{z}}
\newcommand{\by}{\bm{y}}

\newcommand{\bx}{\bm{x}}

\newcommand{\br}{\bm{r}}
\newcommand{\bP}{\bm{P}}
\newcommand{\bu}{\bm{u}}

\newcommand{\bmp}{\bm{p}}
\newcommand{\Atot}{\cA_{\mathrm{tot}}}
\newcommand{\Aset}{\cA}

\newcommand{\tmix}{t_{\mathrm{mix}}}

\newcommand{\otilde}{\tilde{O}}

\newcommand{\code}[1]{\textnormal{\texttt{#1}}}
\SetKwInput{KwInput}{Input}
\SetKwInput{KwParameter}{Parameter}
\SetKwInput{KwSubSolver}{Sub-Solver}
\SetKwInput{KwOutput}{Output}
\SetKwInput{Return}{return}
\usepackage{algorithmic}

\usepackage{booktabs} 
\usepackage{array}

\newcommand{\tv}[2]{d_{TV}\paren{#1, #2}}

\newcommand{\bI}{\bm{I}}

\newcommand{\nOuter}{n_{\mathrm{outer}}}

\newcommand{\nBatchInner}{n_{\mathrm{batch}}}
\newcommand{\nSampleInner}{n_{\mathrm{sample}}}
\newcommand{\UniformDist}{\mathsf{Unif}}

\newcommand{\sr}{\mathrm{sr}}

\newcommand{\gap}{\mathrm{gap}(\bm{A})}

\newcommand{\Solve}{\code{Solve}}

\newcommand{\Dentry}{\cD_{\mathrm{entry}}}
\newcommand{\Drow}{\cD_{\mathrm{row}}}
\newcommand{\Dcol}{\cD_{\mathrm{col}}}

\newcommand{\Traditional}{{\code{StandardLoop}}}
\newcommand{\Noisy}{{\code{NoisyLoop}}}
\newcommand{\Reuse}{{\code{ReuseLoop}}}

\newcommand{\fsub}{f^\mathrm{sub}}
\newcommand{\fsublam}{f^\mathrm{sub}}
\newcommand{\fsubbar}{\bar{f}^\mathrm{sub}}
\newcommand{\traditional}{\code{Standard}}
\newcommand{\noisy}{\code{Noisy}}
\newcommand{\reuse}{\code{Reuse}}
\newcommand{\Outer}{\code{Outer}}
\newcommand{\disteq}{\overset{\cD}{=}}
\newcommand{\sub}{{\mathrm{sub}}}
\newcommand{\ba}{{\bm{u}}}
\newcommand{\project}{{\mathrm{proj}}}

\SetCommentSty{mycommfont}
\SetKwFor{Repeat}{repeat}{}{}

%% file: bib_macros.tex
  \usepackage{nth}
  \usepackage{intcalc}

  \newcommand{\cAAAI}[1]{AAAI\ Conference\ on\ Artificial (AAAI)}

%% file: intro.tex
\section{Introduction}\label{intro:intro}

\emph{Variance reduction} is a powerful technique for designing provably efficient algorithms for solving structured optimization problems. This technique consists of reducing the variance of stochastic or randomized access to a component of the input (what we call \emph{sample queries}) by performing occasional, more expensive queries, which involve the entire input (what we call \emph{full batch queries}). Variance reduction \citep{johnson2013accelerating} has led to improved query complexities for many problems including finite-sum minimization \citep{frostig2015regularizing, johnson2013accelerating, lin2015universal}, top eigenvector computation \citep{garber2016faster}, Markov decision processes (MDPs) \citep{sidford2023variance, sidford2018near, jin2024truncated}, and matrix games \citep{carmon2019variance}.  

\emph{Variance reduction} schemes apply an iterative method---which we refer to as an \emph{outer-solver}---to reduce an original problem to solving a sequence of \emph{sub-problems}. The outer-solver runs $\nOuter$ iterations and, in each iteration, a sub-problem is carefully constructed and solved using what we call a \emph{sub-solver}. This sub-solver is a randomized algorithm that uses $\nBatchInner$ full batch and $\nSampleInner$ \emph{fresh} sample queries to the input to solve a sub-problem. This solves the orignal problem with a total of $\nOuter \nBatchInner$ full batch and $\nOuter\nSampleInner$ sample queries. Often, it is possible to tune the outer loop to trade off how many sub-problems are solved ($\nOuter$), with how challenging each sub-problem is to solve ($\nBatchInner$ and $\nSampleInner$). The central question in this work is: \emph{Can we improve this trade-off between the number of full batch and sample queries in prominent variance reduction settings?}

Concretely, we ask whether randomness in sample queries can be \emph{reused} across \emph{all} $\nOuter$ iterations of the outer-solver to reduce the total number of sample queries from $\nOuter \nSampleInner$ to just $\nSampleInner$, without sacrificing correctness guarantees. We answer this affirmatively by designing a \emph{sample reuse framework} that reuses randomness in variance-reduction settings where (1) the outer-solver is robust to $\ell_\infty$-bounded random noise in sub-problem solutions, and (2) the sub-solver uses randomness \emph{obliviously}, i.e., sampling a random variable whose distribution is independent of the sub-problem.

Applying our sample reuse framework enables us to decrease the total number of sample queries by a factor of $\nOuter$ for finite-sum minimization, top eigenvector computation, and $\ell_2$-$\ell_2$ matrix games. We also propose a {new outer-solver framework for solving \emph{discounted Markov Decision Processes} (DMDPs) that reduces solving a DMDP to solving a sequence of DMDPs of \emph{lower} discount factor. This result---which may be of independent interest---allows us to obtain improved sample query complexities for discounted MDPs and average-reward MDPs as well as faster running times in certain cases. Finally, we show how to apply our framework to obtain sample query complexity improvements for $\ell_2$-$\ell_1$ matrix games, where prior variance-reduction schemes use a mix of oblivious and non-oblivious sampling.

\paragraph{Organization.} This introduction discusses a motivating, illustrative example of finite-sum minimization (\Cref{sec:intro:finite-motivate}), our main results for other structured optimization problems (\Cref{sec:intro:structured}), and preliminaries (Section~\ref{sec:intro:prelim}). Section~\ref{sec:pseudoindependence}, describes our sample-reuse framework and a new notion of \emph{pseudoindependent} algorithms, which generalizes the concept of \emph{pseudodeterminism} introduced in prior work \cite{Gat2011ProbabilisticSA} and enables our analysis of this sample-reuse framework. Details on how our framework can be applied to many other structured optimization problems can be found in Sections~\ref{sec:finite-sum-minimization} through~\ref{sec:topEV}. Section~\ref{sec:conclusion} concludes with a summary and some directions for future work. Omitted proofs are deferred to the Appendix. 

\subsection{Motivating example: convex finite-sum minimization (FSM)}
\label{sec:intro:finite-motivate}

\paragraph{Motivating problem.} As an illustrative, motivating example, consider the prototypical problem of \emph{finite-sum minimization (FSM)} (introduced in greater formality in Section~\ref{sec:finite-sum-minimization}). In the FSM problem, we are given convex, $L$-smooth functions $f_1,\ldots,f_n : \R^{d} \rightarrow \R$. In the standard query model, we have an oracle that, when queried at $\bm{x} \in \R^d$ and $i \in [n]$, outputs $\nabla f_i(\bm{x})$. We then consider the problem of minimizing $F : \R^d \rightarrow \R$ defined as $\smash{F(\bx) \defeq \frac{1}{n} \sum_{i \in [n]} f_i(\bx)}$ under the assumption that $F$ is $\mu$-strongly convex. (The individual $f_i$ need not be strongly convex.)

Near-optimal query complexities can be obtained for this problem by applying an outer-solver such as accelerated proximal point/Catalyst (APP) \citep{frostig2015regularizing, lin2015universal} with stochastic variance-reduced gradient descent (SVRG) as the sub-solver \citep{johnson2013accelerating}. Concretely, APP solves the FSM problem by solving $\smash{\nOuter = \otilde(\sqrt{\alpha/\mu})}$ regularized problems of the form $\smash{\min_{\bx \in \R^{d}} f(\bx) + \frac{\alpha}{2} \norm{\bx - \bm{x}_t}_2^2}$ where for each iteration $t \in [\nOuter]$, $\bm{x}_{t}$ depends on the solution to the $(t-1)$-th sub-problem \citep{frostig2015regularizing}.\footnote{As in prior work, throughout, we may use $\tilde{O}(\cdot)$ to hide poly-logarithmic factors in problem parameters.} Each sub-problem is solved using $\otilde(n + {L}/{\alpha})$ queries, e.g., using SVRG as the sub-solver \citep{johnson2013accelerating}. Taking $\smash{\alpha = \max\{L/n,\mu\}}$ yields a complexity of $\smash{\otilde(n + \sqrt{n L / \mu})}$ oracle queries which is known to be the optimal query complexity, up to logarithmic factors \citep{woodworth2016tight, agarwal2015lower}. 

\paragraph{The batch-sample model.} In light of this prior work, to obtain further improvements, we \emph{must go beyond} this standard query model. We consider a more \emph{fine-grained} \emph{batch-sample query model}, which better captures computational trade-offs that could be present in some settings.

To motivate this model, a closer inspection of the aforementioned algorithm (APP with SVRG) for solving FSM reveals that solving each sub-problem requires $\smash{\otilde(1)}$ computations of the gradient of $F$ and $\smash{\otilde(L / \alpha)}$ computations of the gradient of a component $f_i$, where $i$ is chosen uniformly at random. Since a gradient of $F$ can be computed by querying each of the $n$, $\smash{\nabla f_i}$ once, in the standard query model, each gradient query to $F$ requires $n$ queries. This yields the aforementioned near-optimal query complexity of $\smash{\otilde(n + \sqrt{nL/\mu})}$ in the standard query model.

However, treating all queries to the gradient of $f_i$ as \emph{computationally equivalent} can obscure the structure of the problem. For example, in some practical computational models, computing the gradient of $F$ could be \emph{much cheaper} than computing $\nabla f_i$ for $n$ arbitrary $f_i$---for example, due to caching behavior or memory layout \citep{fischetti2018faster}. In such settings, it could be helpful to optimally trade-off between the number of queries to a gradient of $F$ (batch-queries) and the number of queries to a gradient of a random $f_i$ (sample-queries), depending on their relative costs. 

Moreover, in some classic problems of FSM---namely, empirical risk minimization--- just $O(1)$ queries to the gradient of $f_i$ suffices to \emph{exactly recover} $f_i$ and know $\nabla f_i$ at all points \emph{without any} further queries to $f_i$! This occurs, for example in linear regression when $\smash{f_i(\bx) = \frac{1}{2}(\bm{a}_i^\top \bx - \bm{b}(i))^2}$ for feature vectors $\bm{a}_i \in \R^d$ and (explicitly known) labels $\bm{b} \in \R^n$ and for generalized linear models when $f_i(\bx) = \phi_i(\bm{a}_i^\top \bx - \bm{b}(i))$ for known $\phi_i$ (see Section~\ref{sec:glm-fsm}). Hence, \emph{repeatedly querying the gradient of the same} $f_i$, say $T$ times, can be much cheaper than querying $T$ arbitrary $f_i$. For instance, this can be the case in distributed memory layouts \cite{woodworth2020minibatch, sallinen2016high}. 

To capture these nuances, we consider a more fine-grained analysis of the complexity of FSM where we allow \emph{smoothly trading off} between two types of queries, depending on their relative costs:
\begin{itemize}
    \item \textbf{Full batch query}: for input $\bm{x}$ computes $\nabla F(\bm{x})$, and
    \item \textbf{Sample query}: for input $i \in [n]$ allows $\nabla f_i(\bm{x})$ to be computed for \emph{any} future input $\bm{x}$.
\end{itemize}
For example, recall from above that in the special case of linear regression, each $f_i(x) = \frac{n}{2}(\bm{a}_i^\top x - \bm{b}(i))^2$ where $\smash{\bm{a}_i \in \R^d}$ is row $i$ of a data matrix $\smash{\bm{A} \in \R^{n \times d}}$ and $\bm{b}$ is a label. In this case, a batch oracle query is implementable just with the appropriate \emph{matrix-vector multiply}, $\bm{A}^\top \bm{A} x$, and a sample query is implementable by outputting the appropriate row, $\bm{a}_i$ (see Section~\ref{sec:finite-sum-minimization}). 

This batch-query model captures (1) the fact that full batch queries can be cheaper than $n$ sample queries due to caching layout as well as (2) the fact that re-querying previously cached components is often cheaper than querying fresh components of $F$ due to caching and communication costs between machines. Depending on the computational environment, one can now seek to optimally \emph{trade-off} these two types of oracle queries, depending on their relative costs.

From this perspective, the state-of-the-art FSM algorithms consist of $\smash{\nOuter = \otilde(\sqrt{\alpha/\mu})}$ outer loop iterations of the APP outer-solver. The sub-solver in each iteration (SVRG) is implementable with $\smash{\nBatchInner = \otilde(1)}$ full batch queries and $\smash{\nSampleInner = \otilde(L/\alpha)}$ sample queries. Tuning $\alpha$ allows us to smoothly trade off between these two types of oracle queries.

As mentioned, prior work established that $\smash{\tilde{O}(n + \sqrt{nL/\mu})}$ sample queries is near-optimal for FSM if one \textit{only uses sample queries} \citep{woodworth2016tight}. This lower bound, when $n = 1$, also implies $\tilde{O}(\sqrt{L/\mu})$ full batch queries is optimal when using \textit{only} full batch queries. Thus, it is perhaps natural to expect that solving FSM with $\smash{\otilde(\sqrt{\alpha / \mu})}$ full batch queries and $\smash{\otilde(L/\sqrt{\mu \alpha})}$ sample queries is the optimal query complexity trade-off---after all, it recovers the optimal rate for using \textit{only} full batch queries when $\alpha =L, n=1$ and the optimal sample complexity for using \textit{only} sample queries if $\alpha = \max(L/n, \mu)$.) Still, we ask: \emph{can this {trade-off} be improved?}

\paragraph{Our FSM results.} We show that this trade-off can be improved! By developing and applying techniques for reusing randomness (Section~\ref{sec:pseudoindependence}), we provide a method that uses only $\smash{\otilde(\sqrt{\alpha / \mu})}$ full batch queries and $\smash{\otilde(L/ \alpha)}$ sample queries for any $\alpha \geq \mu$ to solve FSM. That is, we decrease the number of sample queries by $\smash{\otilde(\nOuter) = \otilde(\sqrt{\alpha/\mu})}$. This yields corresponding improvements for regression and generalized linear models (see Section~\ref{sec:glm-fsm}) and shows these problems can be solved \emph{with less information than was known previously.} Although this doesn't directly yield a worst-case asymptotic-runtime improvement that we are aware of, it sheds new light on the information needed to solve FSM and could yield faster algorithms depending on caching and memory layout. 

To obtain this improved trade-off, we observe that in the previously discussed variance-reduction approach, the sub-solver solves each sub-problem to high accuracy by querying the sample oracle for a uniformly random component $i$. Importantly, this distribution for choosing the $i$ \emph{does not depend} on the particular sub-problem, i.e., the samples are \emph{oblivious} of the point at which the gradients are queried. We refer to such queries as \emph{oblivious sample queries}. As mentioned earlier, we show that by solving the regularized sub-problems to high accuracy and adding a small amount of uniform noise to the output, it is possible to \emph{reuse the same $i$} in each sub-problem! To facilitate this proof, we introduce a notion of \emph{pseudo-independence} (Section~\ref{sec:pseudoindependence}), a generalization of pseudo-determinism \citep{Gat2011ProbabilisticSA}, and provide general theorems about pseudo-independence in Appendix~\ref{sec:bound-difference}.

\begin{figure}[t]
    \centering
    \includegraphics[width=\linewidth]{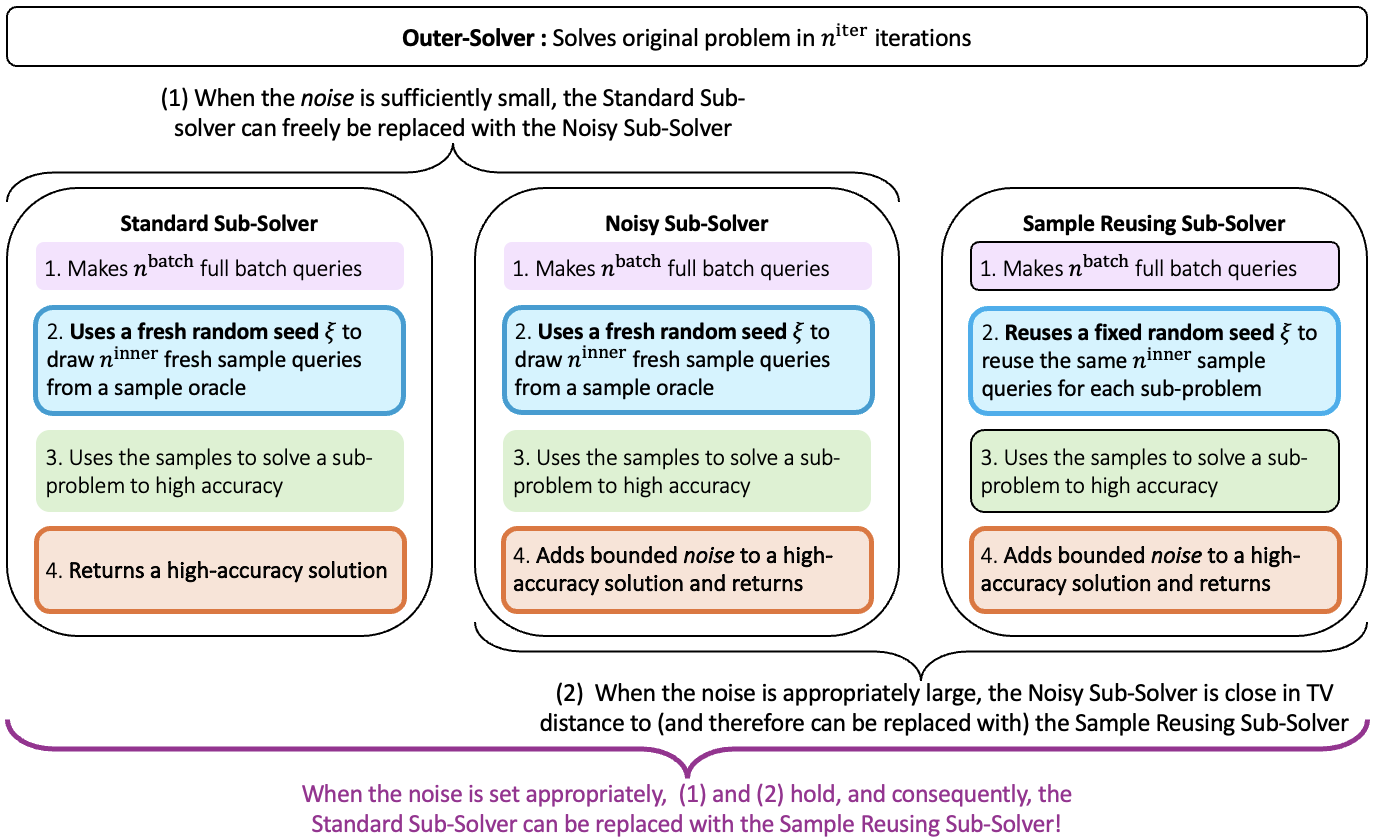}
    \caption{\emph{Sample reuse framework.} The diagram summarizes our approach for replacing Standard Sub-Solvers with a Sample Reusing Sub-Solver, which reuses the same sample queries across all $\nOuter$ iterations of the Outer-Solver. In (2), TV abbreviates total variation distance. 
    }
    \label{fig:diagram}
    \vspace{-1em}
\end{figure}

\subsection{Our results for structured optimization}
\label{sec:intro:structured}

Here we explain our results for prominent variance reduction settings (including FSM). These results follow a similar approach as in \Cref{sec:intro:finite-motivate}, but differ in oracles and sub-problem structure. For each problem we consider using an outer-solver framework (e.g., APP) to iteratively reduce the original optimization problem to a sequence of sub-problems, which are solved to high-accuracy using a sub-solver (e.g., SVRG). (This depected by \emph{Outer-Solver} and \emph{Standard Sub-Solver} in Figure~\ref{fig:diagram}.)

Two observations drive these results. First, in these problems, an outer-solver framework is guaranteed to (probably, approximately) solve the original problem, \emph{even if} at each iteration, the solution to the $t$-th sub-problem were perturbed by some bounded random noise. (This is depected by the \emph{Noisy Sub-Solver} in Figure~\ref{fig:diagram}.) 
Second, when this random noise has a sufficiently large range, the random perturbations \emph{protect} the randomness of the sample queries well enough so that \emph{even if} we were to reuse the \emph{same} sample queries across \emph{all} iterations of the outer-solver (this is depicted by \emph{Sample Reusing Sub-Solver} in Figure~\ref{fig:diagram}), then the distribution over outputs would not change by much---as measured by total variation (TV) distance. To facilitate our proof of this fact, we describe and analyze a new notion of \emph{pseudo-independence} of randomized algorithms in Section~\ref{sec:pseudoindependence}.

Combining these observations, we show that in many problems, when the noise is carefully scaled, the Outer-Solver can use the Sample-Reusing Sub-solver as the sub-problem solver (instead of the Standard Sub-Solver). This reduces the total number of sample queries by a multiplicative $\nOuter$ factor! This is perhaps surprising---it is well-known that randomized algorithms are \emph{not always} amenable to reusing randomness across sequential invocations, see e.g.,\ \citep{cherapanamjeri2020adaptive, cohen2024one}. However, our analysis shows randomness \emph{can} be reused in many applications of variance reduction. 

This sample reuse framework is formalized in Section~\ref{sec:pseudoindependence}. Although this formalism is somewhat abstract and technical---as it is intended to capture a wide range of variance-reduction settings---it enables us to obtain improved query complexity trade-offs for solving a broad range of prominent optimization and machine-learning problems. 

Our improvements are summarized in Table~\ref{table:main-notation}. Note that in certain specialized problems, e.g., TopEV and special cases of FSM such as linear regression, there might be other approaches to obtain our improved trade-offs using preconditioning \citep{cohen2020online, cohen2015uniform}. However, a strength of our sample reuse framework is its versatility--- our approach applies even in problems where there is no clear preconditioning analog, e.g., MDPs or matrix games. In the next paragraphs, we highlight additional implications of our results (beyond improved trade-offs) for MDPs and matrix games.

\paragraph{Faster algorithms for certain DMDPs.}{
Outer-solver frameworks which reduce a problem to a sequence of sub-problems are well-studied for all of the problems in Table~\ref{table:main-notation}, \emph{except} for DMDPs/AMDPs. Thus, to obtain our results for DMDPs (Table~\ref{table:main}) in Section~\ref{sec:dmdp} we derive a new outer-solver algorithm for DMDPs (Theorem~\ref{thm:dmdp-unregularizer}). This algorithm, which may be of independent interest, reduces solving a $\gamma$-discounted DMDP to solving a sequence of $\gamma'$-discounted DMDPs for $\gamma' < \gamma$. This result yields a new algorithm for solving $\gamma$-DMDPs to high-accuracy and improves the \emph{runtime} of prior work \citep{jin2024truncated} under appropriate conditions on the \emph{sparsity} of the underlying transition matrix (Theorem~\ref{thm:corollary-dmdp-unregularizer}.) To extend our results to AMDPs, we leverage a reduction of \citep{jin2021towards}.
}

\paragraph{Improved matrix-vector complexity for $\ell_2$-$\ell_2$ matrix games.}{ For $\ell_2$-$\ell_2$ matrix games, full batch queries and sample queries can \emph{both} be implemented by a (two-sided) matrix-vector oracle that outputs $\smash{(\bm{x}, \bm{y}) \mapsto (\bm{A}\bm{x}, \bm{A}^\top \bm{y})}$ for input $(\bm{x}, \bm{y}) \in \R^{d} \times \R^{n}$. With $\smash{\alpha = \norm{\bm{A}}_F^{2/3}\epsilon^{1/3}}$ we obtain an algorithm that solves $\ell_2$-$\ell_2$ matrix games in only $\smash{\tilde{O}(\norm{\bm{A}}_F^{2/3}\epsilon^{-2/3})}$-matrix-vector oracle queries. This improves the matrix-vector complexity of the problem over the $\smash{\tilde{O}(\norm{\bm{A}}_F \epsilon^{-1})}$-matrix-vector oracle queries required in prior work \citep{carmon2019variance, nemirovski2004prox} and is near-optimal, due to \citep{liu2023accelerated}. \emph{Special} cases of $\ell_2$-$\ell_2$ matrix games reduce to $\ell_2$-regression, in which case this trade-off of $\smash{\tilde{O}(\norm{\bm{A}}_F^{2/3}\epsilon^{-2/3})}$-matrix-vector oracle queries may also follow from preconditioning or Newton method \citep{cohen2020online, liu2023accelerated}. \emph{However}, our work is first to achieve near-optimal matrix-vector oracle query complexity for \emph{general} $\ell_2$-$\ell_2$ matrix games. 
}

\input{table1}
\input{table2}

\paragraph{Improved trade-offs for $\ell_2$-$\ell_1$ matrix games.}{ Beyond the results in Table~\ref{table:main}, in Section~\ref{sec:minimax_unified}, we obtain similar improved full batch versus sample query trade-offs for $\ell_2$-$\ell_1$ matrix games. Notably, in this setting, variance-reduced methods \citep{carmon2019variance} \emph{combine} oblivious sample queries along with non-oblivious sample queries to solve the sub-problems induced by the outer-solver (conceptual proximal point \citep{nemirovski2004prox, carmon2019variance}). Interestingly, our pseudo-independence framework \emph{still} applies to allow us to \emph{reuse} the oblivious (non-adaptive) sample queries across all $\nOuter$ invocations of the sub-solver. We also discuss how our results enable improve 
query complexity trade-offs for two computational geometry problems: the minimum enclosing and maximum inscribed ball problems. }

\subsection{Preliminaries}\label{sec:intro:prelim}

\paragraph{General notation.} Bold lowercase letters are vectors in $\R^d$ where $\bm{u}(i)$ is the $i$-th index of $\bm{u}$. Bold capital letters are matrices. The $\ell_p$ norm of $\bm{u}$ is $\normInline{\bu}_p$. For random variable $A$ over $(\Omega, \cF)$, $p_A$ is its probability measure. For event $E$, $p_{A|E}$ is the measure of $A$ conditioned on $E$, $\neg{E}$ is the complement of $E$, and $\indicator{E}$ is the indicator of $E$. For random variables $A, B$, $A \smash{\overset{{\cD}}{=}} B$ denotes that $A$ and $B$ are equal in distribution. For measures $p$ and $q$, $\smash{\tv{p}{q} := \max_{E \in \cF} |p_A(E) - p_B(E)|}$ is the total variation (TV) distance between $p$ and $q$. The support of distribution $\cD$ or measure $q$ is denoted $\smash{\supp(\cD)}$ or $\supp(q)$. $\mathsf{Ber}$ and $\mathsf{Unif}$ respectively denote Bernoulli and uniform distributions. $\smash{\UniformDist^d(a,b)}$ denotes a $d$-dimensional vector in which the $i$-th entry is an independently and identically distributed $\smash{\UniformDist(a, b)}$ random variable. 

\paragraph{Randomized algorithm notation.} We consider algorithms that use up to two independent sources of randomness. We use $\smash{\cA_{\xi, \chi}}$ to denote a randomized algorithm that takes an input $\ba \in \R^d$ and independent random seeds $\smash{\xi \sim \cD_{\xi}}$ and $\smash{\chi \sim \cD_{\chi}^{\ba}}$ where ${\cD_{\xi}}$ is \emph{independent of, or oblivious with respect to,} $\ba$ and ${\cD^{\ba}_{\chi}}$ might be \emph{dependent on, or be adaptive with respect to} ${\ba}$. We also use the notation ${\mathsf{D}_\chi = \{\cD_\chi^{\ba}\}_{{\ba} \in \R^d}}$ to denote the family of distributions parameterized by ${\ba} \in {\R^d}$. On input ${\ba} \in {\R^d}$, $\smash{\cA_{\xi, \chi}}$ outputs ${\cA_{{\xi}, {\chi}}({\ba})}$ for randomly drawn ${\xi \sim \cD_\xi}$ and ${\chi \sim \cD_\chi^{\ba}}$. At times, we analyze the output of $\smash{\cA_{{\xi}, {\chi}}}$ \emph{conditioned} on the value of one or both random seeds. Specifically, ${\cA_{\xi=s, \chi}({\ba})}$ and ${\cA_{\xi, \chi=c}({\ba})}$ are used to denote the randomized algorithms obtained by fixing the value of ${\xi = s \in \supp(\cD_{\xi})}$ or $\smash{\chi =c \in \supp(\cD_{\chi}^{\ba})}$, respectively. Analogously, ${\cA_{\xi=s, \chi=c}}$ is a \emph{deterministic} algorithm corresponding to conditioning on $\smash{\xi = s \in \supp(\cD_{\xi})}$ and $\smash{\chi = c \in \supp(\cD_{\chi}^{\ba})}.$ We occasionally specify a decomposition of the seed $\chi$ into two sub-seeds $\smash{\chi = (\nu, \iota)}$ where sub-seeds $\nu, \iota$ are drawn independently from $\smash{\nu \sim \cD^x_{\nu}}$ and $\smash{\iota \sim \cD_{\iota}}$. We assume algorithms that output vectors in $\R^d$ have runtime ${\Omega}(d)$. We use the term \emph{high-accuracy} to refer to families of algorithms (parameterized by error and failure probability parameters) with at most \emph{polylogarithmic} dependence on their accuracy and failure probability.

%% file: table1.tex
\begin{table}
\centering
\setlength{\tabcolsep}{2.2pt}
\renewcommand{\arraystretch}{1}
\begin{tabular}{@{}>{\centering\arraybackslash}m{2.5cm} >{\arraybackslash}m{13.5cm} @{}} 
\toprule
\textbf{Problem} & \textbf{Description} \\
\midrule
Finite-sum minimization (FSM)  & 
$F(x) = \frac{1}{n} \sum_{i \in [n]} f_i(x)$ is a $\mu$-strongly convex function where each $f_i$ is $L$-smooth and convex. The goal is to reduce the error (with respect to the minimizer of $F$) of an initial point by a factor of $c > 1$ wp.\ $1-\delta$. (Section~\ref{sec:finite-sum-minimization} and discussed more generally in Section~\ref{apx:nonuniform_smoothnes}.) \\
\midrule
Discounted MDP (DMDP) & The goal is to compute an $\epsilon$-optimal policy for a $\gamma$-discounted infinite-horizon MDP wp.\ $1-\delta$. We use $\Atot$ to denote the total number of state-action pairs in the MDP and assume that rewards and $\Atot$ are bounded. (Section~\ref{sec:mdp}) \\
\midrule
Average-reward MDP (AMDP) & The goal is to compute an $\epsilon$-optimal policy for an average-reward infinite-horizon MDP wp.\ $1-\delta$. We use $\Atot$ to denote the total number of state-action pairs in the MDP and assume that rewards and $\Atot$ are bounded. (Section~\ref{sec:mdp})\\
\midrule
$\ell_2$-$\ell_2$ matrix games ($\ell_2$-$\ell_2$) & The goal is to compute an $\epsilon$-saddle point for an $\ell_2$-$\ell_2$ matrix game in matrix $\bm{A} \in \R^{d \times n}$ wp.\ $1-\delta$. \smash{$\Vert \bm{A}\Vert_F \defeq (\sum_{i} \sum_j \bm{A}_{i,j}^{2})^{1/2}$} is the Frobenius norm.  (Section~\ref{sec:minimax_unified}) \\
\midrule
Top Eigenvector (TopEV) & The goal is to compute an $\epsilon$-approximate top eigenvector of $\bm{A}^\top \bm{A} \in \R^{n \times d}$ with high probability in $d$. Here $\gap$ and $\sr(\bm{A})$ are the relative eigen-gap and stable rank of $\bm{A}$ respectively. (Section~\ref{sec:topEV}) \\
\bottomrule
\end{tabular}
\caption{Summary of structured optimization problems (and abbreviations) studied in this work. Throughout this paper, wp.\ abbreviates ``with probability.''}
\label{table:main-notation}
\end{table}

%% file: table2.tex
\begin{table}
\centering
\setlength{\tabcolsep}{6pt}
\begin{tabular}{@{}c c c p{1cm} c c c@{}}
\toprule
\textbf{Problem} 
  & \multicolumn{3}{c}{\textbf{Prior Work ($\tilde{O}$)}} 
  & \multicolumn{2}{c}{\textbf{Our Trade-off} ($\tilde{O}$)} 
  & \textbf{Range}\\
\cmidrule(lr){2-4} \cmidrule(lr){5-6}
 & FB & \makecell{OS} & \textbf{Paper} & FB & \makecell{OS} & \\
\midrule

FSM 
 & $\sqrt{\alpha/\mu}$ 
 & $L/\sqrt{\alpha \mu}$ 
 & \cite{frostig2015regularizing} 
 & $\sqrt{\alpha/\mu}$ 
 & $L/\alpha$ 
 & $\alpha > \mu$\\

\midrule
DMDP 
 & $1$ 
 & $\Atot(1-\gamma)^{-2}$ 
 & \small{\cite{jin2024truncated}}
 & $\alpha (1-\gamma)^{-1}$ 
 & $\Atot \alpha^{-2}$ 
 & $1 > \alpha \geq 1-\gamma$\\

\midrule
AMDP 
 & $1$ 
 & $\Atot \tmix^2 \epsilon^{-2}$ 
 & \small{\cite{jin2021towards}} 
 & $\alpha \tmix \epsilon^{-1}$ 
 & $\Atot \alpha^{-2}$ 
 & $1 > \alpha \ge \tfrac{\epsilon}{9\tmix}$\\

\midrule
$\ell_2$-$\ell_2$ 
 & $\alpha \epsilon^{-1}$ 
 & $\|\bm{A}\|_{F}^2 (\alpha\epsilon)^{-1}$ 
 & \small{\cite{carmon2019variance}}
 & $\alpha \epsilon^{-1}$ 
 & $\|\bm{A}\|_{F}^2 \alpha^{-2}$ 
 & $\alpha > 0$\\

\midrule
TopEV 
 & $\sqrt{\tfrac{\alpha}{\gap}}$ 
 & $\sr(\bm{A}) (\alpha^{3}\gap)^{-\tfrac{1}{2}}$ 
 & \small{\cite{garber2016faster}}
 & $\sqrt{\tfrac{\alpha}{\gap}}$ 
 & $\sr(\bm{A})\,\alpha^{-2}$ 
 & $\alpha > \Theta(\gap)$\\

\bottomrule
\end{tabular}
\caption{\emph{Main results}. FB and OS denote the required full batch and oblivious (non-adaptive) sample queries, respectively, with $\alpha$ tuning their trade-off. We compare our trade-offs with prior work. Importantly, the ``Our trade-off'' column always improves over the trade-off under ``Prior Work'' under the \emph{same} assumptions and problem definitions made in the prior work.}
\label{table:main}
\end{table}

%% file: pseudoindependence.tex
\section{Sample reuse framework}\label{sec:pseudoindependence}

Here we provide, contextualize, and analyze our sample-reuse framework. First, in Section~\Ref{sec:sample-reuse-overview} we introduce the framework, the main theorem regarding it (Theorem~\ref{thm:AtoC}) and formally define \emph{psuedoindependence} (Definition~\Ref{def:pseudo-independence}) which we use to prove the theorem. For additional context, we then compare pseudo-independence to related definitions in prior work in Section~\ref{sec:comparisons}. We then conclude this section in Section~\ref{sec:proof} by analyzing the sample-reuse framework and proving Theorem~\ref{thm:AtoC}.

\subsection{Sample-reuse framework}
\label{sec:sample-reuse-overview}

Here, we introduce our sample-reuse framework. To describe the framework we provide a general \emph{Meta-Algorithm} for solving an optimization problem. This Meta-Algorithm encompasses three different templates for variance reduction methods. Our framework relates these different templates and applying this relationship to different algorithms yields our improved query-complexity tradeoffs.

More formally, the Meta-Algorithm~\ref{alg:meta-algorithm} iteratively reduces solving a problem instance $X$ to solving a sequence of sub-problems. First, the algorithm initializes $\ba_0 \in {\R^d}$; an oblivious distribution $\cD_\xi$ for sampling a random seed $\xi$; and a \emph{family} of \emph{non-oblivious} distributions $\mathsf{D}_\chi = \{\cD_{\chi}^x\}_{x \in \R^d}$. It iteratively runs one of \emph{three} possible sub-routines (corresponding to the three different templates), which we describe below, and outputs a convex combination of the iterates. Table~\ref{table:main-notation-algos} summarizes the notation. 

\paragraph{The standard sub-routine.}{The first template is $\Outer(X; \traditional)$, i.e., Meta-Algorithm~\ref{alg:meta-algorithm} with sub-routine $\Traditional$ (in the prose, we omit the $\tau$ since it has no effect when $\code{Type} = \traditional$.) This template formalizes the ``Standard Sub-Solver'' panel in Figure~\ref{fig:diagram} and describes the standard variance reduction template that applies to many algorithms for structured optimization problems---including those cited under ``Prior Work'' in Table~\ref{table:main}. In $\Outer(X; \traditional)$, the for loop (Line~\ref{line:call-loop}) iteratively reduces the original optimization problem to solving $\nOuter$ sub-problems, where the $t$-th sub-problem is determined by the $(t-1)$-th iterate, $\ba_{t-1}$. 

The sub-problem solver, or \emph{sub-solver}, denoted  $\smash{\cA^{\mathrm{sub}}_{\xi, \chi}}$, is a randomized algorithm that solves the sub-problem at each iteration by (1) making some deterministic full batch queries to the full batch oracle; (2) using the random seed $\xi$ to make randomized \emph{oblivious} sample queries to a sample oracle; and (3) using the random seed $\chi$ to perform some additional randomization. Depending on the application (Table~\ref{table:main-notation}), (3) might involve adding some noise to the sub-problem or making additional \emph{adaptive} sample queries to a sample oracle. $\smash{\cA^{\mathrm{sub}}_{\xi, \chi}}$ then outputs an intermediate iterate $\ba_{t-1/2} \in \R^p$. 
At the end of each iteration, the routine then applies a ``post-process'' $\zeta$ which performs some \emph{deterministic} post-processing of the pair $(\ba_t, \ba_{t-1/2})$ (e.g., implementing acceleration/momentum.)}

As a illustrative, concrete, example, consider the motivating example of FSM (Section~\ref{sec:intro:finite-motivate}). Meta-Algorithm~\ref{alg:meta-algorithm} captures the previously described combination of APP and SVRG to solve FSM \citep{frostig2015regularizing, lin2015universal}. That is, APP reduces the original problem to solving a sequence of regularized sub-problems, while $\smash{\cA^{\mathrm{sub}}_{\xi, \chi}}$ represents the sub-solver SVRG \citep{johnson2013accelerating}, which is used to solve the regularized sub-problems. In this case, $\zeta$ implements acceleration as in \citep{frostig2015regularizing}, and $\bm{w}$ is set such that the algorithm will simply return the last iterate, $\bm{u}_{\nOuter}$. (See also Section~\ref{sec:finite-sum-minimization}.)

\paragraph{The noisy sub-routine.} To motivate the second template, recall our motivating observation that for many classes of structured optimization problems (Table~\ref{table:main-notation}), there is an instantiation of $\Outer(X; \traditional)$ that provably solves the problem instance $X$ \emph{provided} that: at each iteration $t \in [\nOuter]$, the sub-solver $\smash{\cA^{\mathrm{sub}}_{\xi, \chi}}$ solves the sub-problem induced by $\ba_{t-1}$ to \emph{high-accuracy} in the $\ell_\infty$ norm. We formalize this observation with the following two definitions.

\begin{definition}[Function approximation]\label{def:approx} We say that a randomized algorithm $\cA_{\xi, \chi}$ is an \emph{$(\eta, \delta)$-approximation} of $f: {\R^d} \to {\R^p}$ if $\cA_{\xi, \chi}$ takes an input $\ba \in {\R^d}$ along with two random seeds $\xi \sim \cD_\xi$ and $\chi \sim \cD_{\chi}^{\ba}$ and outputs $\cA_{\xi, \chi}(\ba) \in {\R^p}$ such that $\smash{{\Prob_{\xi \sim \cD_{\xi}, \chi \sim \cD_{\chi}^{\ba}}\paren{\Vert \cA_{\xi, \chi}(\ba) - f(\ba) \Vert_\infty \leq \eta} \geq 1 - \delta}}$.  
\end{definition}
\input{meta_alg}
\vspace{1em}

\begin{definition}[$\ell_\infty$-robust]\label{assumption:robustness} We say $\Outer(X; \traditional)$ is $(\eta, \beta)$-\emph{robust} with respect to $\fsub$ if there exists a deterministic function $\fsub: {\R^d} \to {\R^p}$ such that $\Outer(X; \traditional)$ is guaranteed to output a $\beta$-accurate solution for $X$ \emph{whenever}, for all $\smash{t \in [\nOuter]}$, $\smash{\normInline{\ba_{t-1/2} -  \fsub(\ba_{t-1})}_\infty \leq \eta}$.
\end{definition}

To interpret these definitions, suppose that $\Outer(X; \traditional)$ is $(\eta, \beta)$-\emph{robust}, $\tau \leq \eta/2$, and the sub-solver $\cA^{\mathrm{sub}}_{\xi, \chi}$ is an $(\eta/2, \delta)$-approximation of $\fsub$. Then \emph{even if}, at every iteration $t \in [\nOuter]$, $\ba_{t-1/2}$ were to be randomly perturbed by some $\tau$-bounded random noise, the output would \emph{still} be a $\beta$-accurate solution to $X$ wp.\ $1-\nOuter\delta$! 

Correspondingly our second sub-routine $\Noisy_\tau$ in Meta-Algorithm~\ref{alg:meta-algorithm} (and the corresponding template $\Outer(X; \noisy)$) is \emph{identical} to $\Traditional$ (correspondingly, the template $\Outer(X; \traditional)$) \emph{except} that at each iteration, $\ba_{t-1/2}$ is \emph{perturbed} by a small amount of $\smash{\UniformDist^p(-\tau, \tau)}$ random noise. This formalizes the ``Noisy Sub-Solver'' in Figure~\ref{fig:diagram}. Following the previous intuition, the following lemma shows that when $\tau$ is sufficiently small, we can freely interchange $\smash{\Outer(X; \traditional)}$ with $\smash{\Outer(X; \noisy, \tau)}$, without losing correctness guarantees.

\begin{restatable}{lemma}{AtoB}\label{lemma:AtoB} Suppose $\Outer(X; \traditional)$ is $(\eta, \beta)$-robust with respect to $\fsub$, $\tau \in (0, \eta/2)$, and $\cA^{\mathrm{sub}}_{\xi, \chi}$ is an $(\eta/2, \delta)$-approximation of $\fsub$. Then $\Outer(X; \traditional)$ and 
$\Outer(X; \tau, \noisy)$, wp.\ $1-\nOuter\delta$, output  a $\beta$-accurate solution to $X$. 
\end{restatable}
\begin{proof} The proof is immediate from  Definition~\ref{assumption:robustness} and union bound over all $\nOuter$ iterations. 
\end{proof}

\paragraph{The sample-reuse sub-routine.}
At this point, it is perhaps natural to wonder: \emph{why} is it helpful to work with the ${\Noisy}_\tau$~subroutine as opposed to the standard, \Traditional~sub-routine in Meta-Algorithm~\ref{alg:meta-algorithm}? After all, for any $\tau$, both \Traditional~and $\Noisy_\tau$~require the same number of batch and sample queries, so there is no obvious computational advantage to using $\Noisy_\tau$. 

As discussed in our second observation in Section~\ref{sec:intro:structured}, the key idea is the following: the random {noise} at each iteration ensures that the iterates $\ba_{t}$ in $\Noisy_\tau$ are \emph{almost independent} of the randomness in the sample queries induced by $\xi \sim \cD_{\xi}$ in the following sense: \emph{even if we were to reuse the same realization of $\xi$} in all iterations $t \in [\nOuter]$ of $\Noisy_\tau$, the returned output would be close---in total variation (TV) distance---to the output produced by $\Outer(X; \Noisy, \tau)$. This brings us to our third and final sub-routine.

 The third and final sub-routine $\Reuse_\tau$ in Meta-Algorithm~\ref{alg:meta-algorithm} (and corresponding template $\Outer(X; \reuse)$) is \emph{identical} to the $\Noisy_\tau$~sub-routine (correspondingly, $\Outer(X; \noisy)$) \emph{except that} $\Reuse_\tau$ \emph{reuses} a single realization $s_1 \sim \cD_{\xi}$ in all iterations $t \in [\nOuter]$. This formalizes the ``Sample Reusing Sub-Solver'' in Figure~\ref{fig:diagram}. Reusing this realization across all iterations corresponds to reusing oblivious sample queries to $X$. Consequently, for any $\tau > 0$, $\Outer(X; \reuse, \tau)$ has lower sample query complexity (by an $\nOuter$ multiplicative factor) than the original $\Outer(X; \noisy, \tau)$ or $\Outer(X; \traditional)$! 

\paragraph{Pseudo-independence.}
To obtain improved query-complexity tradeoffs, our goal is to derive conditions when $\Outer(X; \reuse, \tau)$ can be used in place of $\Outer(X; \noisy, \tau)$---without sacrificing correctness. We achieve this goal by introducing a new notion of \emph{pseudo-independence}, which we use to bound the total variation (TV) distance between the outputs of $\Outer(X; \noisy, \tau)$ and $\Outer(X; \reuse, \tau)$. Indeed, if this TV distance is small then as long as $\Outer(X; \noisy, \tau)$ is correct, with high probability, so is $\Outer(X; \reuse, \tau)$.

To define pseudo-independence we first define a \emph{smoothing} of a randomized algorithm as follows

\begin{definition}[$(\epsilon, \delta)$-smoothing]\label{def:epsdelta-smoothing} Let $\cA_{\xi, \chi}$ be a randomized algorithm which takes an input $\ba \in \R^d$ and two random seeds $\xi \sim \cD_{\xi}, \chi \sim  \cD_{\chi}^{\ba}$. We say that algorithm $\bar{\cA}_\chi$ is an \emph{$(\epsilon, \delta)$-smoothing of $\cA_{\xi, \chi}$ with respect to $\xi$} if it takes as an input $\ba \in {\R^d}$ along with one random seed $\chi \sim \cD_{\chi}^\ba$ and for all $\ba \in \R^d$, $\smash{\mathbb{P}_{s \sim \cD_{\xi}} [d_{\mathrm(TV)}(p_{\cA_{\xi=s, \chi}(\ba)}, {p_{\bar{\cA}_{\chi}(\ba)}}) \leq \epsilon]\geq 1-\delta}$.
\end{definition}

We say a randomized algorithm is $(\epsilon, \delta)$-pseudo-independent if we can guarantee \emph{existence} of an $(\epsilon, \delta)$-smoothing. Importantly, this smoothing need not be implementable or even explicitly known. 

\begin{definition}[$(\epsilon, \delta)$-pseudo-independence]\label{def:pseudo-independence} Let $\cA_{\xi, \chi}$ be as in Definition~\ref{def:epsdelta-smoothing}. $\cA_{\xi, \chi}$ is \emph{$(\epsilon, \delta)$-pseudo-independent of $\xi$} if it admits an $(\epsilon, \delta)$-smoothing with respect to $\xi$.
\end{definition}

Intuitively, an algorithm $\smash{\cA_{\xi, \chi}}$ is $\smash{(\epsilon, \delta)}$-pseudo-independent of $\xi$ if, wp.\ $1-\delta$ over the draw of $\smash{s \sim \cD_{\xi}}$, $\smash{\cA_{\xi, \chi}}$ is \emph{almost} independent of the first source of randomness---in the sense that wp.\ $1-\epsilon$ over the draw of $\chi$, $\cA_{\xi=s, \chi}(\ba)$ is equal in distribution to a random variable that is \emph{independent} of $\xi$ (see also Fact~\ref{lemma:tv_decomposition} for further discussion.)

We use this notion of pseudo-independence to bound the TV distance between $\Outer(X; \noisy, \tau)$ and $\Outer(X; \traditional)$. First, in Appendix~\ref{apx:stable_proof} we show that when the noise parameter $\tau$ is set appropriately the sub-solver algorithm $\smash{{\cA'}^{\mathrm{sub}}_{\xi, \chi'}}$ in Meta-Algorithm~\ref{alg:meta-algorithm} is $(\epsilon, \delta)$-pseudo-independent of $\xi$. Intuitively, this means that $\smash{{\cA'}^{\mathrm{sub}}_{\xi, \chi'}}$ is \emph{almost independent} of $\xi$ and consequently, it should be possible to reuse the same realization of the seed $\xi$ in each invocation of $\smash{{\cA'}^{\mathrm{sub}}_{\xi, \chi'}}$. 

To formalize this, we observe that $\Noisy_\tau$ repeatedly composes $\zeta$ with $\smash{{\cA'}^{\mathrm{sub}}_{\xi, \chi'}}$ where $\xi_t \sim \cD_{\xi}$ is drawn fresh in each iteration $t \in [\nOuter]$. Meanwhile, $\Reuse_\tau$ repeatedly composes $\zeta$ with $\smash{{\cA'}^{\mathrm{sub}}_{\xi, \chi'}}$ where the same realization $s_1 \sim \cX_{\xi}$ is \emph{reused} in each iteration $t \in [\nOuter]$. In Appendix~\ref{sec:bound-difference} we provide a general theorems about pseueodindependence, which allow us to leverage the fact that  $\smash{{\cA'}^{\mathrm{sub}}_{\xi, \chi'}}$ is $(\epsilon, \delta)$-pseudoindepent of $\xi$ to bound the TV distance between these repeated compositions as a function of $\nOuter, \delta, \epsilon$. Thus,  when $\tau, \epsilon, \delta$ are set appropriately, the TV distance between the outputs of $\Outer(X; \noisy)$ and $\Outer(X; \reuse)$ is small. This TV distance bound allows us to prove the following theorem, which in turn yields all the results in Table~\ref{table:main}.

\begin{restatable}{theorem}{atoc}\label{thm:AtoC} Suppose $\Outer(X; \traditional)$ is $(\eta, \beta)$-robust with respect to $\fsub$ and $\delta \in (0, 1)$. Let $\eta' \defeq \min(\eta/2, \eta\delta)$. Suppose $\cA^{\mathrm{sub}}_{\xi, \chi}$ is an $(\eta', \delta)$-approximation of $\fsub$. Then wp.\ $1 - 5\nOuter^2\delta$, $\Outer(X; \reuse, \eta'/(2\delta))$ outputs a $\beta$-accurate solution to $X$.
\end{restatable}

In the problems described in Table~\ref{table:main}, the standard variance-reduced algorithms are instantiations of $\Outer(X; \traditional)$ and $(\eta, \beta)$-robust, for $\eta$ scaling polynomially in $\beta$ and the problem parameters. This ensures that the blowups in error and failure probability in Theorem~\ref{thm:AtoC} are at most \emph{polylogarithmic} in all problem parameters. Moreover, for all of our applications, the sub-solvers (e.g., SVRG) are \emph{high-accuracy} algorithms, meaning that they have query- and time-complexity \emph{polylogarithmic} in the accuracy and failure probability parameters. Thus, the polynomial blow-ups in accuracy and failure probability in Theorem~\ref{thm:AtoC} imply at most \emph{polylogarithmic} growth in complexity. 

%% file: meta_alg.tex
\begin{table}[t]
\caption{Parameter table for Meta Algorithm (Algorithm~\ref{alg:meta-algorithm}). } 
\centering
\setlength{\tabcolsep}{2pt}
\renewcommand{\arraystretch}{1} 
\begin{tabular}{@{}>{\centering\arraybackslash}m{2cm} >{\arraybackslash}m{14cm} @{}}
\toprule
\textbf{Symbol} & \textbf{Description } \\
\midrule
$\cD_\xi$ & This is an oblivious distribution governing the random variable $\xi$. \\
\midrule
$\mathsf{D}_\chi$ & This is a family of distributions parameterized by $\bx \in \R^d$. It governs the \emph{adaptive} random variable $\chi$. We use $\cD_\chi^{\bx}$ to denote the distribution in $\mathsf{D}_\chi$ corresponding to an $\bx \in \R^d$. \\
\midrule 
$\cA^{\mathrm{sub}}_{\xi, \chi}$ & $\cA^{\mathrm{sub}}_{\xi, \chi}$ is a sub-problem solver (sub-solver) that takes in $\bu \in \R^d$ and seeds $\xi$ and $\chi$ and outputs $\cA^{\mathrm{sub}}_{\xi, \chi}(\bu) \in \R^p$. $\cA^{\mathrm{sub}}_{\xi, \chi}$ may make full batch queries and sample queries to $X$. Batch queries depend \emph{only} on $\bu$. The seed $\xi$ is used to make \emph{oblivious} sample queries. Meanwhile, the seed $\chi$ may optionally be used to make additional \emph{adaptive} sample queries. \\
\midrule 
$\zeta$ & $\zeta: \R^d \times \R^p \to \R^d$ is a deterministic post-processing function; we call it an \emph{outer-process} since it captures the role of the \emph{outer-solver} discussed in Section~\ref{intro:intro} (Figure~\ref{fig:diagram}). 
\\
\midrule
$\bm{w}$ & The algorithm outputs a convex combination of the $\bu_t$ given by coefficient vector $\bm{w}$. \\
\midrule
$\tau$ & Noise parameter $\tau > 0$ controls the amount of noise to be added.\\
\bottomrule
\end{tabular}
\label{table:main-notation-algos}
\vspace{-1em}
\end{table}

\setcounter{algocf}{0}
\begin{algorithm2e}[H]
\SetKwProg{function}{}{:}{}
\DontPrintSemicolon
\caption{Variance-Reduction Meta-Algorithm $\Outer(X; \code{Type}, \tau)$ }\label{alg:meta-algorithm}
%\SetAlgoLined
\KwParameter{Loop specification:  $\code{Type} \in \{\traditional, \noisy, \reuse\}$ \tcp*{Types defined below}} 
Initialize ${\ba}_0 \in \R^d$\; 
\tcp{If $S = \traditional$, $\tau$ may be omitted, in which case Line~\ref{line:initialize-dist-noise} is skipped and $\tau$ is omitted in Line~\ref{line:call-loop}}
\tcp{The next line builds a noisy version of $\cA^{\mathrm{sub}}_{\xi, \chi}$, which is denoted ${\cA'}^{\mathrm{sub}}_{\xi, \chi'}$}
Let $\cD^{\ba_{t-1}}_{\chi'} \defeq \cD^{\ba_{t-1}}_{\chi} \times \UniformDist^p(-\tau, \tau)$ \label{line:initialize-dist-noise} and ${\cA'}^{\mathrm{sub}}_{\xi=s, \chi'=(c, \bm{e})} (\ba_{t-1}) \defeq {\cA}^{\mathrm{sub}}_{\xi=s, \chi=c}(\ba_{t-1})+\bm{e}$\; 
\textbf{Draw} $s_1, ..., s_{\nOuter} \sim \cD_\xi$ \tcp*{Draw seeds from the oblivious distribution}
\BlankLine
\lFor{each $t \in [\nOuter]$}{
    \textbf{call} $\langle\code{Type}\rangle\code{Loop}_{\tau}$ \label{line:call-loop}
}
\Return{${\sum_{t \in [\nOuter]}} \bm{w}(t) \cdot \bm{u}_{t}$}
\BlankLine
\tcp{Different loops that can be called on Line~\ref{line:call-loop}}
\lfunction{$\Traditional_\tau$}{
    Draw $\smash{\chi_t \sim \cD^{\ba_{t-1}}_{\chi}}$, 
    $\smash{\ba_{t-\frac{1}{2}} \gets \cA^{\mathrm{sub}}_{\xi=s_t, \chi=c_t}(\ba_{t-1})}$, 
    $\smash{\ba_{t} \gets \zeta(\ba_{t-1}, \bm{u}_{t-\frac{1}{2}})}$}
\lfunction{$\Noisy_{\tau}$}{
    Draw $\smash{{(c_t, \bm{e}_t) \sim \cD^{\ba_{t-1}}_{\chi'}}}$, 
    $\smash{\ba_{t-\frac{1}{2}} \gets \textcolor{ForestGreen}{\cA'}^{\mathrm{sub}}_{\xi=s_t, \textcolor{ForestGreen}{\chi'=(c_t, \bm{e}_t)}}(\ba_{t-1})}$, 
    $\smash{\ba_{t} \gets \zeta(\ba_{t-1}, \bm{u}_{t-\frac{1}{2}})}$
}
\lfunction{$\Reuse_\tau$}{
    Draw $\smash{{(c_t, \bm{e}_t) \sim \cD^{\ba_{t-1}}_{\chi'}}}$, 
    $\smash{\ba_{t-\frac{1}{2}} \gets {\cA'}^{\mathrm{sub}}_{\textcolor{ForestGreen}{\xi=s_1}, \chi'=(c_t, \bm{e}_t)}(\ba_{t-1})}$, 
    $\smash{\ba_{t} \gets \zeta(\ba_{t-1}, \bm{u}_{t-\frac{1}{2}})}$
}
\end{algorithm2e}

%% file: proof.tex
\subsection{Comparisons of pseudoindependence to prior work}\label{sec:comparisons}

In this section, we discuss the notion of \emph{pseudo-independence}, which generalizes the well-studied notion of pseudo-determinism (Definition~\ref{def:pseudo-det}) \citep{goldwasser2017pseudo, goldwasser2017bipartite, braverman2023lower, grossman2023tight, dixon2022pseudodeterminism} and is also related to \emph{reproducibility} \citep{impagliazzo2022reproducibility}. 

First, we state the following Fact~\ref{lemma:tv_decomposition}, which will be helpful in our analysis.

\begin{restatable}[Lemma 4.1.13 of \citep{roch_mdp_2024}, restated]{fact}{tvdecomp}\label{lemma:tv_decomposition}
Let $\epsilon > 0$ and $A$ and $B$ be random variables. Then $d_{TV}(p_A, p_B) \leq \epsilon$ if and only if there exist random variables $C, D, F$ and an independent event $E$ such that $\prob{E} = 1-\epsilon$; $A \overset{{\cD}}{=} C \indicator{E} + D \indicator{\neg{E}}$; and $B \overset{\cD}{=} C \indicator{E} + F \indicator{\neg{E}}$.
\end{restatable}

This fact implies the following intuitive interpretation of pseudo-independence. $\smash{\cA_{\xi, \chi}}$ is $\smash{(\epsilon, \delta)}$-pseudo-independent of $\xi$ if wp.\ $1-\delta$ over the draw of $\smash{s \sim \cD_{\xi}}$, $\smash{\cA_{\xi, \chi}}$ is \emph{almost} independent of the first source of randomness---in the sense that wp.\ $1-\epsilon$ over the draw of $\chi$, $\cA_{\xi=s, \chi}(\ba)$ is equal in distribution to a random variable that is \emph{independent} of $\xi.$

Now, we briefly compare pseudo-independence to pseudo-determinism and reproducibility. The term \emph{pseudo-deterministic algorithm}---introduced by \cite{Gat2011ProbabilisticSA}---describes an algorithm that outputs a deterministic value with high probability. This is formalized with the following definition.

\begin{definition}[$\delta$-pseudo-deterministic algorithm]\label{def:pseudo-det} Let $\cA_{\xi}$ be a randomized algorithm which takes an input $\bm{u} \in \R^d$ and a random seed $\xi \sim \cD_{\xi}$. $\cA_{\xi}$ is \emph{$\delta$-pseudo-deterministic} if there exists a function $h$ over $\R^d$ such that $\probSub{\cA_{\xi=s}(\ba) = h(\ba)}{s \sim \cD_{\xi}} \geq 1-\delta.$
 
\end{definition}

Intuitively, in a $\delta$-pseudo-deterministic algorithm $\cA_\xi$, the role of $h$ is similar to that of a smoothing (Definition~\ref{def:epsdelta-smoothing}) in the definition of pseudo-independence (Definition~\ref{def:pseudo-independence}); however, $h$ must be deterministic whereas as a smoothing can be a randomized algorithm. To formalize this intuition we need to introduce some additional notation because pseudo-determinism is defined in terms of a single source of randomness, whereas pseudo-independence (Definition~\ref{def:pseudo-independence}) is defined in terms of two sources of randomness. Thus, to compare pseudo-determinism to pseudo-independence, we define a simple way to \emph{lift} a single-seed randomized algorithm $\cA_{\xi}$ to two sources of randomness. 

\begin{definition} Let $\cA_{\xi}$ be a randomized algorithm which takes an input $\ba \in \R^d$ and a random seed $\xi \sim \cD_{\xi}$. Let $\smash{\cA^{\mathrm{pd}}_{\xi, \chi}(\ba)}$ be the randomized algorithm which takes input $\ba \in \R^d$ and seeds $\xi \sim \cD_{\xi}$ and $\chi \sim \cD_{\chi}^x$ where $\supp(\cD_\chi) = \{0\}$; and maps $\ba \mapsto \cA_{\xi}(\ba)$. 
\end{definition}

That is, $\smash{\cA^{\mathrm{pd}}_{\xi, \chi}}$ is identical to $\cA_{\xi}$; however, it accepts an additional 0-bit ``dummy'' random seed $\chi$. The next lemma explains how pseudo-independence captures pseudo-determinism as a special case. 

\begin{restatable}{lemma}{connectnotions}\label{lemma:connect-notions} $\cA_{\xi}$ is $\delta$-pseudo-deterministic if and only if $\cA^{\mathrm{pd}}_{\xi, \chi}$ is $(0, \delta)$-pseudoindependent of $\xi$. 
\end{restatable}
\begin{proof} First, suppose that $\cA_{\xi}$ is $\delta$-pseudo-deterministic. Then, there exists a function $h$ such that 
\begin{align*}
   1-\delta \geq \probSub{{\cA}_{\xi=s}(\ba) = h(\ba)}{s \sim \cD_{\xi}} = \probSub{\tv{p_{\cA^{\mathrm{pd}}_{\xi=s, {\chi}}(\bu)}}{p_{h(\ba)}} = 0}{s \sim \cD_{\xi}}\,. 
\end{align*}
Hence, if $\cA'_{\chi}$ is the algorithm that deterministically maps $\ba \mapsto h(\ba)$ then $\cA'_{\chi}$ is a $(0, \delta)$-smoothing of $\cA^{\mathrm{pd}}_{\xi, \chi}$ with respect to $\xi$. Thus, $\cA^{\mathrm{pd}}_{\xi, \chi}$ is $(0, \delta)$- pseudo-independent with respect to $\xi$. 

Conversely, if $\cA^{\mathrm{pd}}_{\xi, \chi}$ is $(0, \delta)$-pseudo-independent with respect to $\xi$ then let $\cA'_{\chi}$ be a $(0, \delta)$-smoothing of $\cA^{\mathrm{pd}}_{\xi, \chi}$ with respect to $\xi$. Then, 
\begin{align*}
     1-\delta &\geq \probSub{\tv{p_{\cA^{\mathrm{pd}}_{\xi=s, {\chi}}(\ba)}}{p_{\cA'_{\chi}(\ba)}} = 0}{s \sim \cD_{\xi}} \\
     &= \probSub{{\cA}_{\xi=s}(\ba) = \cA'_{\chi = 0}(\ba)}{s \sim \cD_{\xi}}\,.
\end{align*}
Letting $h$ to be the function mapping $\ba \mapsto \cA'_{\chi = 0}(\ba)$, we see that $\cA_{\xi}$ is $\delta$-pseudo-deterministic. 
\end{proof}

On the other hand, \emph{reproducibility} is a related notion introduced in \citet{impagliazzo2022reproducibility}. 

\begin{definition}[$(\delta, \cD)$-reproducibile algorithm]\label{def:reproduce} Let $\cA_{\xi}$ be a randomized algorithm which takes an input $\ba \in \R^d$ and a random seed $\xi \sim \cD_{\xi}$. Let $\cD$ be a distribution over $\R^d$. $\cA_{\xi}$ is \emph{$(\delta, \cD)$-reproducible} if there is a function $h$ over $\supp(\cD_\xi)$ such that $\probSub{\cA_{\xi=s}(\ba) = h(s)}{\ba \sim \cD, s \sim \cD_{\xi}} \geq 1-\delta.$
\end{definition}

Reproducibility asks that with high probability over the draw of an input sample $\ba \sim \cD$ \emph{and} of the random seed $s \sim \cD_\xi$, the algorithm outputs a deterministic function of the realized seed: $h(s)$. On the other hand, pseudo-determinism asks that for \emph{every} input $\ba$, with high probability over the random draw $s \sim \cD_\xi$, the randomized algorithm outputs a deterministic function of the input $h(\ba)$. Finally, pseudoindependence asks that for \emph{every} input $\ba$, a randomized algorithm should be \emph{almost independent} of one of its random seeds in the sense that with high probability over the draw of $s \sim \xi$, its output is close in total-variation distance to a randomized algorithm which is completely oblivious of $s$.

\subsection{Analysis of the sample-reuse framework}\label{sec:proof}

In this section, we given an outline of our proof of Theorem~\ref{thm:AtoC}. Our discussion here expands on the intuitive explanation in Section~\ref{sec:sample-reuse-overview}. To reduce notational clutter, throughout this section, we omit the superscript $\mathrm{sub}$ from ${\cA'}^{\mathrm{sub}}_{\xi, \chi'}$ and refer to it just as ${\cA'_{\xi, \chi'}}$. 

\paragraph{Inducing pseudoindependence in the noisy solver.}
Recall that in Section~\ref{sec:sample-reuse-overview}, we said that the first step in our proof of Theorem~\ref{thm:AtoC} would be to show that $\cA'_{\xi, \chi'}$ in Meta-Algorithm~\ref{alg:meta-algorithm} is pseudo-independent of $\xi$ when the noise parameter $\tau$ is set appropriately. Concretely, we prove the following theorem in Appendix~\ref{apx:stable_proof}. 
 
\begin{restatable}{theorem}{conversionthm}\label{thm:conversion-thm} Let $\epsilon, \delta \in (0, 1)$, $\eta > 0$, $\eta' \defeq \min(\eta/2, \eta\epsilon)$, and let $\cA_{\xi, \chi}$ be a randomized algorithm that is an $(\eta', \delta)$-approximation of a function $f: {\R^d} \to {\R^p}$. Let $\cA'_{\xi, \chi'}$ be the randomized algorithm defined as follows. $\cA'_{\xi, \chi'}$ takes input $\ba \in \R^d$, random seed $\xi \sim \cD_\xi$, and $\chi'$ where $\chi' = (\chi, \bm{\nu}) \sim \cD_{\chi'}^{\ba}$ is the concatenation of an independently drawn seed $\chi \sim \cD_{\chi}^{\ba}$ and seed $\smash{\bm{\nu} \sim \cD_{\nu} \defeq \UniformDist^p(-\tau, \tau)}$ for $\tau \defeq \eta'/(2\epsilon)$.  For any realization $s, c, \bm{e}$ of $\xi, \chi', \bm{\nu}$ and any $\ba \in {\R^d}$, we define $\cA'_{\xi=s, \chi'=(c, \bm{e})}(\ba) = \cA_{\xi = s, \chi = c}(\ba) + \bm{e}$. Then, ${\cA'_{\xi, \chi'}}$ is an $(\epsilon, \delta)$-pseudo-independent of  $\xi$ and an $(\eta, \delta)$-approximation of $f$ with the \emph{same} runtime and query complexities as $\cA_{\xi, \chi}$ up to an additive $O(p)$ in runtime.
\end{restatable}

\begin{proof} Let $\cA_{\xi, \chi}$ be the randomized algorithm which takes input $\bx \in \R^d$, random seed $\xi \sim \cD_\xi$, and $\chi$ where $\chi = (\chi', \nu) \sim \cD_{\chi}^{\bx}$ is the concatenation of an independently drawn seed $\chi' \sim \cD_{\chi'}^{\bx}$ and seed $\smash{\nu \sim \cD_{\nu} \defeq \UniformDist^p(-t, t)}$, where we use $\UniformDist^p$ to denote the distribution of a $p$-dimensional independent uniform random vector, and $t$ is some parameter to be specified later in this proof.
 
For any realization $s, c, \bm{e}$ of $\xi, \chi', \bm{\nu}$, let $\cA_{\xi=s, \chi=(c, \bm{e})}(x) = \cA_{\xi = s, \chi' = c} + \bm{e}$. That is, on a given input $\bx$, $\cA_{\xi,\chi}(\bx)$ has the same distribution of a random variable which is equal to $\cA_{\xi,\chi'}(\bx)$ plus some independent $\UniformDist(-t, t)$ random noise in each coordinate. 

First, we construct a smoothing for $\cA_{\xi, \chi}$. Let $\bar{\cA}_{\chi}$ be the randomized algorithm which takes input $\bx\in \R^d$ and a random seed $\chi \sim \cD_{\chi}^{\bx}$. For any realization $(c, \bm{e})$ of $\chi = (\chi', \bm{\nu})$, let $\bar{\cA}_{\chi}(\bx) = f(\bx) + \bm{\nu}$. Now, using the fact that 
\begin{align}\label{eq:target-closeness}
    \Prob_{\xi \sim \cD_{\xi}, \chi \sim \cD_{\chi'}^{\bx}}\paren{\Vert \cA'_{\xi, \chi'}(\bx) - f(\bx) \Vert_\infty \leq \eta'} \geq 1 - \delta, 
\end{align}
for $p \geq 1$ and $\eta' < t$,(since the volume of the $p$-timensional hypercube of side length $s$ is $s^p$):
\begin{align*}
    \probSub{\tv{p_{A_{\xi = s, \chi}(x)}}{p_{\bar{A}_{\chi}(\bx)}} \leq {\frac{\eta'}{2t}}}{s \sim \cD_{\chi}} &\geq \probSub{\tv{p_{A_{\xi = s, \chi}(x)}}{p_{\bar{A}_{\chi}(\bx)}} \leq \paren{\frac{\eta'}{2t}}^p}{s \sim \cD_{\chi}} \\
    &\geq 1-\delta.
\end{align*}

Next, we need to show that $\Prob_{\xi \sim \cD_{\xi}}(\Vert \cA_{\xi, \chi}(x) - f(x) \Vert_\infty \geq \eta) \leq \delta.$ Once again, using \eqref{eq:target-closeness}, we see that wp.\ $1-\delta$ over the draw of $\xi$, $\Vert \cA_{\xi, \chi}(x) - f(x) \Vert_\infty \leq \eta' + 2t \leq \eta$ whenever $\eta', 2t \leq \eta/2$. Consequently, when $t = {\eta'}/{2\epsilon}$, $\bar{\cA}_{\chi}$ is an $(\epsilon, \delta)$-smoothing for $\cA_{\xi, \chi}$; further, when $\eta' \leq \min(\eta/2, \eta\epsilon)$, $\Prob_{\xi \sim \cD_{\xi}}(\Vert \cA_{\xi, \chi}(x) - f(x) \Vert_\infty \geq \eta) \leq \delta$ as well. This completes the proof of the first guarantee of $\cA_{\xi, \chi}$. For the second guarantee, note that $\cA_{\xi, \chi}$ has the same runtime and query complexities up to an additive $O(p)$ increase in the runtime due to the cost of adding the $p$-dimensional uniform random noise induced by $\bm{\nu}$.
\end{proof}

We discuss implications to numerical stability further in Appendix~\ref{apx:stable_proof}. 

Theorem~\ref{thm:conversion-thm} show how to convert standard high-accuracy structured optimization algorithms into \emph{pseudo-independent} high-accuracy structured optimization algorithms (Section~\ref{sec:pseudoindependence}) by adding noise to the output. Next, we will prove these {pseudo-independent} versions are amenable to \emph{reusing} samples across multiple iterations of an outer-solver (Figure~\ref{fig:diagram}.) We briefly remark that this technique of adding noise to an algorithm's output also appears in differential privacy \citep{asi2020near, ghazi2022differentially} and in the design and analysis of algorithms that are robust to adaptive adversaries \citep{chen2022maximum, lee2015efficient, beimel2022dynamic, brandmaxflow}.    

\paragraph{Analyzing repeated compositions of pseudoindependent functions}

Once we know that $\cA'_{\xi, \chi'}$ in Meta-Algorithm~\ref{alg:meta-algorithm} is \emph{pseudo-independent} of $\xi$ for appropriately chosen $\tau$, we move on to proving Theorem~\ref{thm:AtoC}. 

Recall that, as outlined in Section~\ref{sec:pseudoindependence}, our goal will be to show that if ${\cA_{\xi, \chi'}}$ in Meta-Algorithm~\ref{alg:meta-algorithm} is ($\epsilon, \delta$)-pseudo-independent of $\xi$, then the distribution over outputs of $\Outer(X; \noisy, \tau)$ and $\Outer(X; \reuse, \tau)$ are close in TV. 
To prove this, first observe that $\Noisy_\tau$ repeatedly composes $\zeta$ with $\cA'_{\xi, \chi'}$. Definition~\ref{def:composition} introduces notation for these compositions. 

In what follows, Eq. \eqref{eq:comp-dist-A} represents an algorithm that repeatedly applies $T$ iterations, where in each iteration $\xi, \chi$ are fresh random variables drawn from their respective distributions (as in $\Noisy_\tau$). Eq. \eqref{eq:comp-partial-A} is the analogous algorithm where in the $i$-th iteration a fresh $\chi'$ is drawn from $\cD_{\chi'}^{\ba}$, {but} $\xi = s$ is \emph{reused} across iterations (as in $\Reuse_\tau$.) Eq. \eqref{eq:comp-dist-Ap} is the analogous expression to \eqref{eq:comp-dist-A} with the randomized algorithm $\bar{\cA}_{\chi'}$, which takes only \emph{one} source of randomness, $\chi'$.  

\begin{definition}[Composition of randomized algorithms]\label{def:composition} Let $\cA'_{\xi, \chi'}$ be a randomized algorithm which takes an input $\ba \in {\R^d}$ and random seeds $\xi \sim \cD_{\xi}, \chi' \sim \cD_{\chi'}^x$ and outputs a point in ${\R^p}$. Let $\bar{\cA}_{\chi'}$ be a randomized algorithm that takes an input $\ba \in {\R^d}$ and one random seed $\chi' \sim \cD_{\chi'}^x$ and outputs a point in ${\R^p}$. Let $\zeta: {\R^d} \times {\R^p} \to {\R^d}$ be a deterministic function. For $T \geq 1$, let: 
\begin{equation}
\begin{split}
\Phi_{\cA'}^1(\ba; \cD_{\xi}, \mathsf{D}_{\chi'}) 
&\;\defeq\; \zeta\bigl(\ba, \cA'_{\xi, \chi'}(\ba)\bigr),\\
\Phi_{\cA'}^{T+1}(\ba; \cD_{\xi}, \mathsf{D}_{\chi'})
&\;\defeq\; \zeta\Bigl(\Phi_{\cA'}^{T}(\ba; \cD_{\xi}, \mathsf{D}_{\chi'}),\;
 \cA'_{\xi, \chi'}\bigl(\Phi_{\cA'}^{T}(\ba; \cD_{\xi}, \mathsf{D}_{\chi'})\bigr)\Bigr)
\end{split}
\label{eq:comp-dist-A}
\end{equation}
\begin{equation}
\begin{split}
\Phi_{\cA'}^1(\ba; s, \mathsf{D}_{\chi'}) 
&\;\defeq\; \zeta\bigl(\ba, \cA'_{\xi=s, \chi'}(\ba)\bigr),\\
\Phi_{\cA'}^{T+1}(\ba; s, \mathsf{D}_{\chi'}) 
&\;\defeq\; \zeta\Bigl(\Phi_{\cA'}^{T}(\ba; s, \mathsf{D}_{\chi'}),\;
  \cA'_{\xi=s, \chi'}\bigl(\Phi_{\cA'}^{T}(\ba; s, \mathsf{D}_{\chi'})\bigr)\Bigr)
\end{split}
\label{eq:comp-partial-A}
\end{equation}
\begin{equation}
\begin{split}
H_{\bar{\cA}}^1(\ba; \mathsf{D}_{\chi'}) 
&\;\defeq\; \zeta\bigl(\bar{\cA}_{\chi'}(\ba)\bigr),\\
H_{\bar{\cA}}^{T+1}(\ba; \mathsf{D}_{\chi'}) 
&\;\defeq\; \zeta\Bigl(\Phi_{\bar{\cA}}^{T}(\ba; \mathsf{D}_{\chi'}),\;
  \bar{\cA}_{\chi'}\bigl(\Phi_{\bar{\cA}}^{T}(\ba; \mathsf{D}_{\chi'})\bigr)\Bigr)
\end{split}
\label{eq:comp-dist-Ap}
\end{equation}
\end{definition}

In Section~\ref{sec:bound-difference}, we show the following general theorem about pseudoindependence and repeated composition.

\begin{restatable}{theorem}{pseudoindependencemain}\label{lem:pseudoindependencemain}Let $\cA'_{\xi, \chi'}$ be randomized algorithm which takes an input $\ba \in {\R^d}$ and two random seeds $\xi \sim \cD_{\xi}, \chi' \sim \cD_{\chi'}^\ba$ and is ($\epsilon, \delta$)-pseudo-independent of  $\xi$. Then, 
\[
d_{TV}\big({p_{\Phi^T_{\cA'}(\ba; s, \mathsf{D}_{\chi'})}}, {p_{\Phi^T_{\cA'}(\ba; \cD_\xi, \mathsf{D}_{\chi'})}}\big) \leq 2T(\delta + \epsilon)\,.
\]
\end{restatable}

Finally, to prove Theorem~\ref{thm:AtoC}, we can combine the previous results.

\atoc*
\begin{proof} For notational convenience, set $\epsilon = \delta$ and $\tau = \eta'/(2\epsilon)$. Notice that 
\begin{align}
    \Outer(X; \noisy, \tau) &\disteq \sum_{t \in [\nOuter]} w(t) \cdot \Phi^{t}_{{\cA'}}(\ba_0; \cD_\xi, \mathsf{D}_{\chi'}), \label{eq:first-version} \\
    \Outer(X; \reuse, \tau) &\disteq \sum_{t \in [\nOuter]} w(t) \cdot \Phi^{t}_{{\cA'}}(\ba_0; s_1, \mathsf{D}_{\chi'}), \label{eq:second-version}
\end{align}
where $s_1 \sim \cD_\xi$. Since $\tau \leq \eta / 2$ and $\eta' \leq \eta / 2$, by Lemma~\ref{lemma:AtoB}, \eqref{eq:first-version} is a $\beta$-accurate solution to $X$ wp.\ $1 - \nOuter\delta$. 

Also, by Theorem~\ref{thm:conversion-thm}, ${\cA'}^\sub_{\xi, \chi'}$ is $(\epsilon, \delta)$-pseudo-independent of $\xi$. Thus, by Theorem~\ref{lem:pseudoindependencemain}, the TV distance between the $t$-th summand of \eqref{eq:first-version} and the $t$-th summand of \eqref{eq:second-version} is bounded by $2t(\delta + \epsilon)$.

As \eqref{eq:first-version} is $\beta$-accurate wp.\ $1 - \nOuter\delta$, we can apply Fact~\ref{lemma:tv_decomposition} and take union bound over all $t \in [\nOuter]$ to conclude that wp.\ $1 - 2\nOuter^2 (\delta + \epsilon) - \nOuter \delta \geq 1 - 5 \nOuter^2\delta$, \eqref{eq:second-version} is also $\beta$-accurate.
\end{proof}

%% file: finite_sum.tex
\section{Application: Finite-sum minimization}\label{sec:finite-sum-minimization}

In this section, we apply our sample reuse framework to obtain improved full batch versus sample query trade-offs for finite sum minimization (FSM). As a special case, we discuss implications for generalized linear models. 

We focus on this setting for FSM due to its simplicity as it is perhaps the most illustrative and foundational setting. We consider a more general variant of FSM in Appendix~\ref{apx:nonuniform_smoothnes}, where we obtain rates dependent on the (potentially non-uniform) smoothness $L_i$ of each $f_i$.

FSM arises in many machine learning settings, such as \emph{empirical risk minimization}, where each $f_i$ might be a loss function for a sampled data point $i \in [n]$ and the goal is to minimize the average loss across all the data. In this context, a query $\bm{x}$ to a gradient oracle for $F$ corresponds to making a pass over the \emph{full batch} of data $i \in [n]$ to compute the gradient at $\bm{x}$. Meanwhile, querying a component oracle with $i$ only requires accessing information pertaining to the $i$-th data point. 

Below, we formally define the problem and the batch-sample oracle model and provide a brief sketch of how we obtain our result in Table~\ref{table:main} for FSM (Section~\ref{sec:finite-sum-minimization}).

\begin{definition}[{\color{ForestGreen}{FSM problem}}]\label{def:finite-sum-minimization} In the \emph{FSM problem} we are given $c > 1$ and $\bm{x}_0 \in \R^d$ and must output $\hat{\bm{x}} \in \R^d$ such that $F(\hat{\bm{x}}) - \min_{\bm{x} \in \R^d} F(\bm{x}) \leq 1/c \cdot (F(\bm{\bm{x}}_0) - \min_{\bm{z} \in \R^d} F(\bm{z}))$ where $\smash{F(\bm{x}) \defeq \frac{1}{n}\sum_{i\in[n]} f_i(\bm{x})}$, $F$ is $\mu$ strongly-convex, and each $f_i: \R^d \to \R$ is $L$-smooth and convex.
\end{definition}

\begin{definition}[Gradient oracle - {\color{ForestGreen}{FSM batch oracle}}]\label{def:gradient-oracle} When queried with $\bm{x} \in \R^d$, a \emph{gradient oracle} for differentiable $F : \R^d \to \R$ returns $\nabla F({\bx}) \in \R^d$.
\end{definition}

\begin{definition}[Component oracle - {\color{ForestGreen}{FSM sample oracle}}]\label{def:component-oracle} When queried with $i \in [n]$, a \emph{component oracle} for $F(\bm{x}) = \frac{1}{n} \sum_{i \in [n]} f_i(\bm{x})$ for differentiable $f_i: \R^d \to \R$ returns a \emph{gradient oracle} for $f_i$.
\end{definition}

\paragraph{The standard outer-solver and sub-solver.}

As discussed in Section~\ref{sec:intro:finite-motivate}, state-of-the-art query complexities for FSM can be achieved by with accelerated proximal point/Catalyst (APP) as an outer-solver \citep{frostig2015regularizing, lin2015universal}, $\ell_2$-regularized FSM for the sub-problems, and stochastic variance-reduced gradient descent (SVRG) as a sub-solver \citep{johnson2013accelerating}. Formally, each iteration of APP approximately solves a $\lambda$-regularized \emph{sub-problem}, as follows. 

\begin{definition}[{\color{ForestGreen} FSM sub-problem}]\label{def:fsm-subproblem} 
Let $F$ be as in Definition~\ref{def:finite-sum-minimization} and $\lambda \geq \mu$.  For any ${\ba} = (\bm{x}, \bm{v}) \in \R^d \times \R^d$, let $\rho = (\mu + 2\lambda)/\mu$ and $\bm{y}_{{\ba}} \defeq 1/(1+\rho^{-1/2}) {\bx} + \rho^{-1/2}/(1 + \rho^{-1/2}) \bv$. We define
\begin{align*}
    \fsublam({\ba}) &= \argmin_{\tilde{{\bx}} \in \R^d}  F(\tilde{{\bx}}) + \frac{\lambda}{2} \norm{\tilde{{\bx}} - {\by}_{{\ba}}}_2^2.
\end{align*}
Moreover, we say that $\bm{a'}=(\bx', \bv')$ is a solution to the $({\ba}, \lambda, c)$-sub-problem if 
\begin{align*}
 F({\bx}') - \min_{\tilde{{\bx}} \in \R^d}  F(\tilde{{\bx}}) + \frac{\lambda}{2} \norm{\tilde{{\bx}} - {\by}_{{\ba}}}_2^2 \leq \frac{1}{c} \paren{F({\by}_{{\ba}}) - \min_{\tilde{{\bx}} \in \R^d}  F(\tilde{{\bx}}) + \frac{\lambda}{2} \norm{\tilde{{\bx}} - {\by}_{{\ba}}}_2^2}\,.
\end{align*}
\end{definition}
APP uses the following post-process (recall Meta-Algorithm~\ref{alg:meta-algorithm}). This is essentially intended to implement \emph{acceleration} (sometimes called \emph{momentum}) on the iterates. 

\begin{definition}[{\color{ForestGreen} FSM post-process}]\label{def:fsm-post-process} 
Let $F$ be as in Definition~\ref{def:finite-sum-minimization} and $\lambda > 0$.  For any ${\ba} = ({\bx}, \bv) \in \R^d \times \R^d$ and ${\ba}' = ({\bx}', \bv') \in \R^d \times \R^d$, let $\rho = (\mu + 2\lambda)/\mu$ and $\iota = 2/\mu + 1/\lambda$, and define 
\begin{align*}
    \zeta({\ba}, {\ba}') \defeq (1 - \rho^{-1/2})\bv + \rho^{-1/2} \paren{ \bm{y}_{{\bx}, \bv} - \iota \lambda (\bm{y}_{{\bx}, \bv} - {\bx}')}. 
\end{align*}
\end{definition}

Now, APP is known to solve the original problem, given approximate solutions to these sub-problems, in the following sense. 
 
\begin{theorem}[APP, Theorem 1.1 of \cite{frostig2015regularizing}, restated - {\color{ForestGreen} FSM outer-solver}]\label{thm:approx-prox} Let $F$ be a $\mu$-strongly-convex function and $\lambda \geq \mu$. There is a $c' = \poly(\lambda, \mu, c, d)$ such that the following holds. Suppose that in each iteration $t \in [\nOuter]$ of Algorithm~\ref{alg:fsm}, $\bm{x}_{t-1/2}$ is a solution to the $({\ba}_{t-1}, \lambda, c')$-sub-problem; then $\bm{u}_{\nOuter}$ is a solution to the FSM problem.
\end{theorem}

\begin{algorithm2e}[ht!]
\DontPrintSemicolon
\SetKwInput{KwParameters}{Parameters}
\caption{$\code{APP-SVRG}_{\lambda, \delta}(\bm{x}_0, F, c)$ Pseudocode}\label{alg:fsm}
\KwInput{Initial point $\bm{x}_0 \in \R^d$, gradient oracle and component oracle for $F$, error factor $c$}
\KwParameters{Failure probability $\delta$, and $\lambda \geq \mu$.}
Initialize $\bv_0 \gets \bm{0} \in\R^d$ \; 
Initialize $\ba_0 \gets (\bm{x}_0, \bm{v}_0)\in \R^d \times \R^d$ \; 
Set sufficiently large $\nOuter = \tilde{O}(\sqrt{\lambda/\mu})$\; 
Set sufficiently large $c' = \poly(\lambda, \mu, c, d)$\;
\For{each $t \in [\nOuter]$}{
    \tcp{We draw realizations of the random seed $\xi$ according to the following distribution.}
    Draw $s_t \defeq \{i_1, ..., i_T\} \sim \cD_\xi$ where each $i_j \sim \UniformDist{[n]}$ for sufficiently large $T = \tilde{O}(L/\lambda)$\; 
    \tcp{We include $\chi$, a zero-bit random bit for technical consistency with \textnormal{Section~\ref{sec:pseudoindependence}}}
    Draw $c_t \sim \cD_\chi \defeq \UniformDist{[0]}$\; 
    ${\ba}_{t-1/2} = ({\bx}_{t-1/2}, \bv_{t-1/2})\gets \cA^{\mathrm{SVRG} | \lambda, c', \delta/\nOuter}_{\xi=s_t, \chi=c_t}({\ba}_{t-1})$.\; 
    \tcp{The post-process implements momentum, as described in \textnormal{Definition~\ref{def:fsm-post-process}} }
    ${\ba}_{t} = ({\bx}_{t}, \bv_{t}) \gets \zeta({\ba}_t, {\ba}_{t-1/2})$\; 
}
\Return{$\bm{x}_{\nOuter}$}
\end{algorithm2e}

In particular, SVRG is a commonly used \emph{sub-solver} to solve the sub-problems required in APP (Theorem~\ref{thm:approx-prox}) efficiently. We restate this result below. 

\begin{theorem}[SVRG, Theorem 2.2 of \cite{frostig2015regularizing}\label{def:sub-problem}, restated - {\color{ForestGreen} FSM sub-problem solver}]\label{thm:svrg-convex} Let $F$ be as in Definition~\ref{def:finite-sum-minimization} and $\lambda \geq \mu$. There is a randomized algorithm $\cA^{\mathrm{SVRG} | \lambda, c, \delta}_{\xi, \chi}$ such that the following holds. 
\begin{itemize}[leftmargin=*, nosep]
    \item The algorithm takes in the random seeds $\xi, \chi$ distributed as follows. The first random seed $\xi \sim \{i_1, ..., i_T\}$ where each $i_j \sim \UniformDist{[n]}$ for some $T = \tilde{O}(L/\lambda)$. The second random seed $\chi \sim \UniformDist{[0]}$ is a 0-bit random seed. 
    \item If $\bm{x}' = \cA^{\mathrm{SVRG} | \lambda, c, \delta}_{\xi, \chi}({\ba})$, then wp.\ $1-\delta$ over the draw of the random seeds $\xi$ and $\chi$, ${\bx}'$ is a solution to the $({\ba}, \lambda, c)$-sub-problem.
    \item $\cA^{\mathrm{SVRG} | \lambda, c, \delta}_{\xi, \chi}$ makes $\tilde{O}(1)$ batch queries and makes sample queries \emph{only} on the indices contained within $\xi$.
\end{itemize}
\end{theorem}

By combining Theorem~\ref{thm:approx-prox} (APP) with Theorem~\ref{thm:svrg-convex} (SVRG) and taking a union bound over all $\nOuter$ iterations in Algorithm~\ref{alg:fsm}, we obtain an algorithm for solving FSM problems. This algorithm  obtains a trade-off of $\tilde{O}(\sqrt{\lambda/\mu})$ full batch (gradient oracle) queries and $\tilde{O}(L/\sqrt{\lambda \mu})$ sample (component oracle) queries for any $\lambda \geq \mu$, as reported in the ``Prior Work'' column of Table~\ref{table:main}. 

However, using our sample reuse framework from Section~\ref{sec:pseudoindependence}, we can improve this trade-off!

\paragraph{The sample reusing sub-solver} 

Observe that Algorithm~\ref{alg:fsm} \emph{exactly} fits into the Meta-Algorithm~\ref{alg:meta-algorithm} framework from Section~\ref{sec:pseudoindependence}. In particular, Algorithm~\ref{alg:fsm} implements $\Outer(X, \traditional)$ where $X = (\bm{x}_0, F, c)$ is an FSM problem instance as per Definition~\ref{def:finite-sum-minimization}. To show that we can replace the standard sub-solver with the sample-reusing sub-solver (i.e., $\Outer(X, \reuse)$), we need to invoke Theorem~\ref{thm:AtoC}. 
To do so, we need to show (1) that APP satisfies robustness with respect to $\fsublam$ in the sense of Definition~\ref{assumption:robustness} and (2) that SVRG can compute an approximation in the $\ell_\infty$ norm, as per Definition~\ref{def:approx}. We prove this below 

\begin{lemma}[{\color{ForestGreen} FSM outer-solver robustness}]\label{lemma:fsm-robust} Let $F$ be as in Definition~\ref{def:finite-sum-minimization} and $\lambda > 0$. There is a $c' = \poly(\lambda, \mu, c, d)$ such that the following holds. Suppose that in each iteration $t \in [\nOuter]$ of the algorithm, 
\begin{align*}
    \norm{{\bu}_{t-1/2} - \fsublam(\bm{y}_{\ba_{t-1}})}_\infty^2 \leq \frac{1}{c'} \paren{F({\by}_{{\ba}_{t-1}}) - \min_{\tilde{{\bx}} \in \R^d}  F(\tilde{{\bx}}) + \frac{\lambda}{2} \norm{\tilde{{\bx}} - {\by}_{{\ba}_{t-1}}}_2^2}, 
\end{align*}
Then $\bm{u}_{\nOuter}$ is a solution to the FSM problem.
\end{lemma}
\begin{proof} We prove the lemma by invoking Theorem~\ref{thm:approx-prox}. That is, it suffices to show that for any $c, \lambda$, there exists a $c' = \poly(\lambda, \mu, c, d)$ such that whenever ${\ba}'$ satisfies 
\begin{align}\label{eq:bound-fsm}
    \norm{{\bu}' - \fsublam(\bm{y}_{\ba})}_\infty^2 \leq \frac{1}{c'} \paren{F({\by}_{{\ba}}) - \min_{\tilde{{\bx}} \in \R^d}  F(\tilde{{\bx}}) + \frac{\lambda}{2} \norm{\tilde{{\bx}} - {\by}_{{\ba}}}_2^2}, 
\end{align}
we have that 
\begin{align}\label{eq:bound-fsm-2}
 F({\bu}') - \min_{\tilde{{\bx}} \in \R^d}  F(\tilde{{\bx}}) + \frac{\lambda}{2} \norm{\tilde{{\bx}} - {\by}_{{\ba}}}_2^2 \leq \frac{1}{c} \cdot \paren{F({\by}_{{\ba}}) - \min_{\tilde{{\bx}} \in \R^d}  F(\tilde{{\bx}}) + \frac{\lambda}{2} \norm{\tilde{{\bx}} - {\by}_{{\ba}}}_2^2}. 
\end{align}

To prove this we will use the smoothness of the function $F_{\lambda, {\ba}}$, which is defined as follows: 
\begin{align*}
    F_{\lambda, {\ba}}(\tilde{{\bx}}) \defeq F(\tilde{{\bx}}) + \frac{\lambda}{2} \norm{\tilde{{\bx}} - {\by}_{\bu}}_2^2. 
\end{align*}
First, we have that \eqref{eq:bound-fsm} implies 
\begin{align*}
    \norm{{\bu}' - \fsublam(\bm{y}_{\ba})}_2^2 \leq d \norm{{\bu}' - \fsublam(\bm{y}_{\ba})}_\infty^2 \leq \frac{2d}{c'} \paren{F({\by}_{{\ba}}) - \min_{\tilde{{\bx}} \in \R^d}  F(\tilde{{\bx}}) + \frac{\lambda}{2} \norm{\tilde{{\bx}} - {\by}_{{\ba}}}_2^2}. 
\end{align*}
Thus, by $L$-smoothness, whenever $c' > 2d/cL$, \eqref{eq:bound-fsm-2} holds. The result follows by Theorem~\ref{thm:approx-prox}.
\end{proof}

\begin{lemma}[{\color{ForestGreen} FSM sub-problem solver high-precision}]\label{lemma:svrg-convex-hp} Let $F$ be as in Definition~\ref{def:finite-sum-minimization}. There is a randomized algorithm $\cA^{\mathrm{SVRG-HP} | \lambda, c, \delta}_{\xi, \chi}$ such that the following holds. 
\begin{itemize}[leftmargin=*]
    \item The algorithm takes in the random seeds $\xi, \chi$ distributed as follows. The first random seed $\xi \sim \{i_1, ..., i_T\}$ where each $i_j \sim \UniformDist{[n]}$ and $T = \tilde{O}(L/\lambda)$. The second random seed $\chi \sim \UniformDist{[0]}$ is a 0-bit random seed. 
    \item If $\bm{x}'= \cA^{\mathrm{SVRG-HP} | \lambda, c, \delta}_{\xi, \chi}({\ba})$, then wp.\ $1-\delta$ over the draw of the random seeds $\xi$ and $\chi$,
    \begin{align*}
    \norm{{\bx}' - \fsublam(\bm{y}_{\ba})}_\infty^2 \leq \frac{1}{c} \cdot \paren{F({\by}_{{\ba}}) - \min_{\tilde{{\bx}} \in \R^d}  F(\tilde{{\bx}}) + \frac{\lambda}{2} \norm{\tilde{{\bx}} - {\by}_{{\ba}}}_2^2}. 
    \end{align*}
    \item $\cA^{\mathrm{SVRG-HP} | \lambda, c, \delta}_{\xi, \chi}$ makes $\tilde{O}(1)$ batch queries and makes sample queries \emph{only} on indices in $\xi$.
\end{itemize}
\end{lemma}
\begin{proof} We prove the lemma by invoking Theorem~\ref{thm:svrg-convex}. It is enough to show that there is a $c' = \poly(\lambda, \mu, c, d)$ such that whenever ${\bx}'$ satisfies 
\begin{align*}
     F({\bx}') - \min_{\tilde{{\bx}} \in \R^d}  F(\tilde{{\bx}}) + \frac{\lambda}{2} \norm{\tilde{{\bx}} - {\by}_{{\ba}}}_2^2 \leq \frac{1}{c'} \cdot \paren{F({\by}_{{\ba}}) - \min_{\tilde{{\bx}} \in \R^d}  F(\tilde{{\bx}}) + \frac{\lambda}{2} \norm{\tilde{{\bx}} - {\by}_{{\ba}}}_2^2}, 
\end{align*}
we also have that 
\begin{align*}
    \norm{{\bx}' - \fsublam(\bm{y}_{\ba})}_\infty^2 \leq \frac{1}{c} \cdot \paren{F({\by}_{{\ba}}) - \min_{\tilde{{\bx}} \in \R^d}  F(\tilde{{\bx}}) + \frac{\lambda}{2} \norm{\tilde{{\bx}} - {\by}_{{\ba}}}_2^2}, 
\end{align*}
We will use the strong convexity of the function $F_{\lambda, {\ba}}$, which is defined as follows:
\begin{align*}
    F_{\lambda, \bm{x}}(\tilde{{\bx}}) \defeq F(\tilde{{\bx}}) + \frac{\lambda}{2} \norm{\tilde{{\bx}} - {\bx}}_2^2. 
\end{align*}
Now, by $\mu$-strong convexity, 
\begin{align*}
    \norm{{\bx}' - \fsublam(\bm{y}_{\ba})}_\infty^2 \leq \norm{{\bx}' - \fsublam(\bm{y}_{\ba})}_2^2 \leq \frac{\mu}{2c} \cdot \paren{F({\by}_{{\ba}}) - \min_{\tilde{{\bx}} \in \R^d}  F(\tilde{{\bx}}) + \frac{\lambda}{2} \norm{\tilde{{\bx}} - {\by}_{{\ba}}}_2^2}, 
\end{align*}
thus it suffices to set $c' = c \mu/2$. 
\end{proof}

\begin{algorithm2e}[ht!]
\DontPrintSemicolon
\SetKwInput{KwParameters}{Parameters}
\caption{$\code{APP-SVRG-Reuse}_{\lambda, \delta}(\bm{z}_0, F, c)$ Pseudocode}\label{alg:fsm-reuse}
\KwInput{Initial point $\bm{z}_0 \in \R^d$, gradient oracle and component oracle for $F$, error factor $c$}
\KwParameters{Failure probability $\delta$, and $\lambda \geq \mu$.}
Initialize ${\ba}_0 \gets (\bm{x}_0, \bm{0})\in \R^d \times \R^d$ \; 
Set sufficiently large $\nOuter = \tilde{O}(\sqrt{\lambda/\mu})$\; 
\tcp{We draw a realization of the random seed $\xi$ according to the following distribution.} 
Draw $s_1 \defeq \{i_1, ..., i_T\} \sim \cD_\xi$ where each $i_j \sim \UniformDist{[n]}$ for sufficiently large $T = \tilde{O}(L/\lambda)$\; 
\For{each $t \in [\nOuter]$}{ 
    \tcp{Draw $c'_t$ from the noisy distribution $\cD_{\chi'}$ as in Meta-Algorithm~\ref{alg:meta-algorithm}.} 
    Draw $c'_t \sim \cD_{\chi'}$\; 

    \tcp{Implement the noisy analog of ${\cA}^{\mathrm{SVRG-HP}| \lambda, \poly(\lambda, \mu, c), \delta/(5\nOuter^2)}$ as in Meta-Algorithm~\ref{alg:meta-algorithm}.}
    ${\ba}_{t-1/2} = ({\bx}_{t-1/2}, \bv_{t-1/2})\gets {\cA'}^{\mathrm{SVRG-HP}| \lambda, \poly(\lambda, \mu, c), \delta/(5\nOuter^2)}_{\xi=s_1, \chi' = c_t}({\ba}_{t-1})$.\; 
    \tcp{The post-process implements momentum, as described in \textnormal{Definition~\ref{def:fsm-post-process}} } 
    ${\ba}_{t} = ({\bx}_{t}, \bv_{t}) \gets \zeta({\ba}_t, {\ba}_{t-1/2})$\; 
}
\Return{$\bm{x}_{\nOuter}$}
\end{algorithm2e}

By combining Lemma~\ref{lemma:svrg-convex-hp} with Lemma~\ref{lemma:fsm-robust} and applying Theorem~\ref{thm:AtoC}, we immediately obtain our improved trade-off. 

\begin{restatable}[{\color{ForestGreen} FSM improvement}]{theorem}{fsmfinal}\label{thm:improved} There is an algorithm (Algorithm~\ref{alg:fsm-reuse}) that for $F$ as in Definition~\ref{def:finite-sum-minimization}, $\lambda \geq \mu$, and $\delta \in (0, 1)$, makes $\tilde{O}(\sqrt{\lambda/\mu})$-batch queries and $\tilde{O}(L/\lambda)$-sample queries and solves the FSM problem wp.\ $1-\delta$. 
\end{restatable}
\begin{proof} We apply Theorem~\ref{thm:AtoC} using $\cA^{\mathrm{SVRG-HP} | c', \lambda, \delta/(5\nOuter^2)}_{\xi, \chi}$ as the sub-solver to solve the $(\bu_{t-1}, \lambda, c)$ sub-problem in each iteration, where $c'$ is as required by Lemma~\ref{lemma:fsm-robust}. The lower failure probability $\delta/(5\nOuter^2)$ is used in order to counter the blowup in failure probability in Theorem~\ref{thm:AtoC}. 
\end{proof}

\subsection{Applications to regression with generalized linear models (GLM) }\label{sec:glm-fsm} Regression with generalized linear models is a special case of FSM where we have $n$ data vectors $\{{\bm{a}}_i\}_{i=1}^n \in \R^d$ with $n$ corresponding, \emph{explicitly known}, labels $\{b_i\}_{i=1}^n \in \R$, and $f_i = \phi_i({{\bm{a}}}_i^\top \bm{x})$ for some explicit function $\phi_i(z)$ (e.g., for least-squares regression, $\phi_i(z) = 1/2 (z - b_i)^2$ and for logistic regression, $\phi_i(z) = \log(1+\exp(-zb_i))$. 

Because the gradient mapping $g:z \mapsto \nabla \phi_i(z)$ is known explicitly, component-oracles for $f_i$ are easily implementable: we can simply query $i \sim \bm{p}$, lookup ${{\bm{a}}}_i$, and return the function $\nabla \phi_i({{\bm{a}}}_i^\top \bm{x}) = {{\bm{a}}}_i \cdot g({{\bm{a}}}_i^\top \bm{x})$ explicitly for \emph{any} future point $\bm{x}$. 

So, in these settings, our improved sample query complexity corresponds to an improved trade-off between the number of full passes over $\{{{\bm{a}}}_i\}_{i=1}^m$ and the number of random samples ${{\bm{a}}}_{i \sim \UniformDist[n]}$ required to train a GLM. The implications of this result are two-fold: 
\begin{itemize}[leftmargin=*]
    \item First, this improvement means that our results prove that problems such as fitting a generalized linear model can be solved with \emph{less information} than was known previously. 
    \item Second, our results show that one can \emph{reuse} the same samples $\bm{a}_{i \sim \UniformDist[n]}$ across all iterations of the outer solver. In distributed memory settings or in settings where caching significantly affects computational costs, this ability to reuse the same cached samples might be beneficial in reducing memory retrieval or communication costs.  
\end{itemize}

\subsection{Extension to non-uniform smoothness}\label{sec:non-uniform-discussion} Our method can also be extended to non-uniformly smooth $f_i$, where we relax the assumption that each $f_i$ is $L$-smooth in Definition~\ref{def:finite-sum-minimization} to assume that each $f_i$ is $L_i$-smooth. In this setting, can develop improved trade-offs that depend on the distribution of the $L_i$'s rather than the worst-case smoothness $\max_{i} L_i$ by instantiating APP~with the primal-dual FSM sub-problem solver of \cite{jin2022sharper}. We defer a discussion of this setting to Appendix~\ref{apx:nonuniform_smoothnes}.

\paragraph{A note on pseudocode in the remainder of this paper.} We include the pseudocode Algorithm~\ref{alg:fsm} and Algorithm~\ref{alg:fsm-reuse} in this section in order to provide a clear example of a concrete instantiation of Meta-Algorithm~\ref{alg:meta-algorithm} for FSM. 

However, for the sake of brevity, in later appendix sections, we will not rewrite the instantiation of Meta-Algorithm~\ref{alg:meta-algorithm} for each application. Instead, we will simply define the initializations; setting for $\nOuter$; post-processing functions; sub-solvers, distributions, and target functions $\fsublam$---as this fully characterizes the instantiations of Meta-Algorithm~\ref{alg:meta-algorithm} for each application. Additionally, we include thorough references to related work which contains application-specific pseudocodes for the outer-solvers and sub-solevrs related to each application. 

An exception is in the case of discounted MDPs (Section~\ref{sec:dmdp}), where we \emph{will} include pseudocode. This is because our outer-solver for DMDPs is a new contribution of our work, so we feel it is helpful to include the full pseudocode for completeness.  

%% file: mdp.tex
\section{Application: Infinite-horizon Markov Decision Processes (MDPs)}\label{sec:mdp}

In this section, we discuss how our framework can be applied to infinite-horizon Markov Decision Processes (MDPs). We primarily focus on the discounted case (DMDPs) in Section~\ref{sec:dmdp} and then in Section~\ref{sec:amdp} extend our results to the average-reward case (AMDPs) using a standard reduction due to \cite{jin2021towards}. We begin with preliminaries in Section~\ref{ssubsecc:prelim}. 

\subsection{Preliminaries}\label{ssubsecc:prelim}

We denote an MDP by $\cM = (\cS, \cA, \bP)$ where $\cS$ denotes a finite state-action space, $\cA$ denotes the total set of all state-action pairs. Concretely, we use $\cA_s$ to denote the set of available actions in state $s \in \cS$ and define ${\cA = \{(s, a) : s \in \cS, a \in \cA_s\}}$ as well as $\Atot \defeq |\cA|$. The state-action-state transition matrix is given by $\bP \in \R^{\cA \times \cS}$. We use $\bmp(s, a) \in \Delta^{\cS}$ to denote the $(s, a)$-th row of $\bP$. A policy $\pi$ is a mapping from states to actions, i.e., $\pi : s \mapsto \pi(s) \in \cA_s$. We use $\Pi$ to denote the set of all possible policies. 

Throughout this section, for vectors $\bm{a}, \bm{b}$, we use $\bm{a} \geq \bm{b}$ to mean entrywise inequality (and use $ \leq, >, <$ analogously.) For $\Atot$-dimensional vectors, say $\br \in \R^\Atot$, we use the notation $\br(s, a)$ to denote the $(s,a)$-th entry of $\br$. We assume rewards are known a-priori, as in prior works in this setting (see, e.g., \cite{jin2024truncated} and references therein.)

For any policy $\pi$, we use $\br^\pi \in \R_{\geq 0}^\cS$ to denote the $|\cS|$-dimensional vector with $s$-th entry given by $\br^\pi(s, \pi(s))$. Likewise, we use $\bP^\pi \in \R^{\cS \times \cS}$ to denote the sub-matrix of $\bP$ where the $s$-th row of $\bP^\pi$ is $\bmp(s, \pi(s)).$ 

An agent interacts with the MDP in \emph{timestep} as follows. At each timestep $t \geq 0$, an agent begins in state $s_t$, chooses an available action $a_t \in \cA_{s_t}$, and collects a (finite) reward $\br(s_t, a_t)$. The agent then (stochastically) transitions to the next state $s_{t+1}$ where $s_{t+1} \sim \bmp(s, a)$. The agent's goal is to compute a \emph{policy} for selecting actions in each state that maximizes the agent's infinite-horizon expected utility over the set of all policies, denoted $\Pi$. 

This objective is often formalized in two settings: the discounted setting (DMDP) and the average-reward setting (AMDP), which we will discuss in Section~\ref{sec:dmdp} and Section~\ref{sec:amdp} respectively. However, we will first define the full batch and sample oracle access for $\cM$. 

\begin{definition}[Matrix-vector oracle - {\color{ForestGreen}{DMDP/AMDP batch oracle}}] When queried with $\bm{x} \in \R^\cS$, a \emph{matrix-vector oracle} for $\bP$ returns $\bP \bm{x}$.
\end{definition}

\begin{definition}[Simulator oracle - {\color{ForestGreen}{DMDP/AMDP sample oracle}}]\label{def:generative-model} When queried with $(s, a) \in \cA$, a \emph{simulator oracle} returns a sample $s' \sim \bmp(s, a)$. 
\end{definition}

The oracle described in Definition~\ref{def:generative-model} is often called a \textit{generative model} in the reinforcement learning theory literature (see, for example, \cite{kearns1998finite}).

\subsection{Discounted MDP}\label{sec:dmdp}

We begin by describing MDPs in the \emph{discounted setting}. 

\paragraph{Discounted MDP (DMDP).} In a DMDP, the objective is to compute an $\epsilon$-optimal policy for maximizing the infinite-horizon $\gamma$-discounted reward for some known constant $\gamma \in (0, 1)$. 

\begin{definition}[Discounted MDP (DMDP)] Let $\cM = (\cS, \cA, \bP)$ be an MDP, $\gamma \in (0, 1)$, and $\br \in [0, 1]^{\cA}$. We denote the associated $\gamma$-discounted-MDP by $\cM_{\gamma, \br} = (\cM; \gamma, \br)$.
\end{definition}

In a DMDP, the value of a policy $\pi$ is defined as follows. 

\begin{definition}[DMDP policy value] Let $\cM_{\gamma, \br}$ be a DMDP and $\pi$ be a policy. The value of policy $\pi$, denoted $\bv_{\gamma, \br}^\pi$, is defined as 
\begin{align*}
    \bv_{\gamma, \br}^\pi(s) \defeq \E\Brac{\sum_{t \geq 0} \gamma^t \br(s_t, \pi(s_t)) | s_0 = s}. 
\end{align*}
\end{definition}

We can alternatively define the value of a policy in terms of the Bellman operator. 
\begin{definition}[Bellman operator]\label{def:bellman} Let $\cM_{\gamma, \br}$ be a DMDP. Given a policy $\pi$, and vector $\bv \in \R^{\cS}$ we define the \emph{Bellman operator} associated with $\pi$ as $\cT^\pi_{\gamma, \br}[\bv] \in \R^{\cS}$ as 
\begin{align*}
    \cT^\pi_{\gamma, \br}[\bv](s) \defeq \br(s, \pi(s)) + \gamma \bmp(s, \pi(s))^\top(s,  \pi(s)). 
\end{align*} The value of policy $\pi$ is the unique vector such that $\bv^\pi_{\gamma, \br} = \cT^{\pi}_{\gamma, \br}[\bv_{\gamma, \br}^\pi]$. We also define the Bellman operator $\cT_{\gamma, \br}$ to be the operator that maps $\bv \mapsto \cT_{\gamma, \br}[\bv] \in \R^{\cS}$ where
\begin{align*}
    \cT_{\gamma, \br}[\bv](s) \defeq \max_{a \in \cA_s} \br(s,a) + \gamma \bmp(s,  a)^\top \bv.
\end{align*} 
\end{definition}

\begin{fact}[Bellman optimality conditions] The value of the optimal policy $\pi_{\gamma, \br}^\star$ for $\cM_{\gamma, \br}$ satisfies 
\begin{align*}
    \bv_{\gamma, \br}^{\pi_{\gamma, \br}^\star}(s) = \max_{\pi \in \Pi} \E\Brac{\sum_{t \geq 0} \gamma^t \br(s_t, \pi(s_t)) | s_0 = s}, 
\end{align*} 
and $\bv_{\gamma, \br}^{\pi_{\gamma, \br}^\star}$ is the unique vector such that $\bv_{\gamma, \br}^{\pi_{\gamma, \br}^\star} = \cT_{\gamma, \br}[\bv_{\gamma, \br}^{\pi_{\gamma, \br}^\star}]$. We use the shorthand $\bv_{\gamma, \br}^{\star} = \bv_{\gamma, \br}^{\pi_{\gamma, \br}^\star}$ and refer to it as the \emph{optimal value vector} for $\cM_{\gamma, r}$. 
\end{fact}

Equipped with these definitions of values in DMDPs, we define the DMDP problem as follows. 

\begin{definition}[{\color{ForestGreen}{DMDP problem}}]\label{def:dmdp-problem} Let $\cM_{\gamma, \br}$ be a $\gamma$-discounted MDP and $\epsilon \in (0, (1-\gamma)^{-1}]$ be given. We must compute a value $\bv$ and a policy $\pi$ such that
\begin{align*}
    \bm{0} \leq \bv_{\gamma, \br}^\star - \bv \leq \epsilon \bm{1}, \quad \text{ and } \quad \bv_{\gamma, \br}^\star- \epsilon \bm{1} \leq \bv_{\gamma, \br}^\pi \leq \bv_{\gamma, \br}^\star.
\end{align*}
Such a policy $\pi$ and value $\bv$ is called an $\epsilon$-optimal policy and $\epsilon$-optimal value for $\cM_{\gamma, \br}$. 
\end{definition} 

\paragraph{The standard outer-solver and sub-solver.}

Now, we will show how to reduce the DMDP problem for a discount factor $\gamma > 0$ to solving a sequence of sub-problems. In this case, the sub-problem will essentially require solving a DMDP with a \textit{smaller} discount factor $\gamma'$. Note that in the sub-problems, we relax the requirement from $\br \in [0, 1]^{\cA}$ to simply $\br \in \R_{\geq 0}^{\cA}$. 

\begin{definition}[{\color{ForestGreen} DMDP sub-problem}]\label{def:dmdp-sub-problem} Let $\cM = (\cS, \cA, \bP)$ be an MDP and $\br \in \R_{\geq 0}^\cA$ be a reward vector. In the $(\cM, \br, \gamma', \bv, \epsilon)$-sub-problem, we are given $\gamma' > 0$, $\bv \in \R^\cS$ and $\epsilon \in (0, 1/(1-\gamma')]$. We define $\br' \defeq \br - (\gamma' - \gamma)\bP \bv$,  and must compute a value $\bv$ such that $\normInline{\bv_{\gamma', \br'}^\star - \bv}_\infty \leq \epsilon/2.$
\end{definition}

To enable computing approximately optimal policies in addition to approximately optimal values, we also introduce the following \emph{policy}-sub-problem as follows. 

\begin{definition}[{\color{ForestGreen}{DMDP policy-sub-problem}}]\label{def:dmdp-policy-sub-problem} Let $\cM = (\cS, \cA, \bP)$ be an MDP and $\br \in \R_{\geq 0}^\cS$ be a reward vector. In the $(\cM, \br, \gamma', \bv, \epsilon)$-policy-sub-problem, we are given $\gamma' > 0$, $\bv \in \R^\cS, \epsilon \in (0, 1/(1-\gamma')]$ 
\begin{align*}
    \br' \defeq \br - (\gamma' - \gamma)\bP \bv, 
\end{align*}
and we must compute a value $\bv$ and a policy $\pi$ such that $\bm{0} \leq \bv_{\gamma', \br'}^\star - \bv \leq \epsilon \bm{1}$ and $\bv \leq \bv_{\gamma', \br'}^{\pi}.$ That is, we must find an $\epsilon$-optimal value and policy for $\cM_{\gamma', \br'}$. 
\end{definition}

In comparison to the DMDP sub-problem (Definition~\ref{def:dmdp-sub-problem}), in the DMDP policy-sub-problem, we must not only output an $\epsilon$-optimal value for $\cM_{\gamma', \br'}$ but must \emph{also} output a policy $\pi$ that attains at least that value. For the DMDP problem, the outer process is defined (simply) as follows. 

\begin{definition}[{\color{ForestGreen} DMDP post-process}]\label{def:mdp-post-process} Fix $\epsilon \in (0, 1/(1-\gamma')]$. For any $\bv, \bv' \in \R^\cS$, we define $\zeta_\epsilon(\bv, \bv') \defeq \bv' - \epsilon \bm{1}.$
\end{definition}

The outer process is intented to adjust (shift down the entries of) $\bv'$ to ensure that it is \emph{an underestimate} of the optimal value. 

We prove the following outer-solver for DMDPs, which we call the Proximal Reward Method (PRM), since it is in spirit similar to proximal point methods in convex optimization.   

\begin{restatable}[PRM - {\color{ForestGreen} DMDP outer-solver}]{theorem}{prm}\label{thm:dmdp-unregularizer}
Let $\cM_{\gamma, \br}$ be a DMDP, $\gamma' < \gamma$, and $\epsilon \in (0, 1/(1-\gamma)]$. Let $\epsilon' = \epsilon/4 \cdot (1-\gamma)/(1-\gamma')$. Suppose that in each iteration $t \in [\nOuter]$ of Algorithm~\ref{alg:dmdp}, $\bv_{t-1/2}$ is a solution to the $(\cM, \br, \gamma', \bv_{t-1}, \epsilon')$-sub-problem. Suppose that $\bv_{\nOuter+1}, \pi_{\nOuter+1}$ is a solution to the $(\cM, \br, \gamma', \bv_{\nOuter}, \epsilon')$-policy sub-problem. Then, $\bv_{\nOuter+1}, \pi_{\nOuter+1}$ solves the DMDP problem.
\end{restatable}

We prove Theorem~\ref{thm:dmdp-unregularizer} in Section~\ref{sec:thm}, and for now, focus on its application. To do so, we need to discuss high-precision algorithms for solving the DMDP sub-problem and policy-sub-problem. 

\paragraph{Subproblem solvers.} There are many high-accuracy algorithms for solving the DMDP subproblem as well as the DMDP-policy-sub-problem. These methods can broadly be divided into interior-point methods (e.g., \citep{SWWY18, lee2014path, brand2021, cohen2020solving, jiang2021faster}), classical value iteration (VI) \citep{littman2013complexity, tseng1990solving}, variance-reduced variants variants of (VI), as well as policy-based methods such as policy iteration or policy gradient and variants (see e.g., \cite{grondman2012survey, bertsekas2011approximate} for a survey.) These variants \citep{SWWY18, jin2024truncated} implement an approximate version of VI by trading-off between full batch queries (matrix-vector products in $\bP$) and sample queries (simulator oracle queries) to obtain faster runtimes. Since we are interested in trade-offs between full batch and sample queries, we consider the Truncated Variance-Reduced Value Iteration (TVRVI) outer-solver from \cite{jin2024truncated} since it achieves the best trade-off between full batch and sample queries in this setting.

\begin{theorem}[TVRVI, Theorem 1.2 of \cite{jin2024truncated}, restated - {\color{ForestGreen} DMDP sub-problem solver}]\label{thm:dmdp-sub-problem-solver} Let $\cM = (\cS, \cA, \bP)$ be an MDP, $\br \in \R_{\geq 0}^\cA$ and $0 < \gamma' < \gamma< 1$. There are randomized algorithms $\cA^{\mathrm{SVRG} | \gamma', \epsilon}_{\xi, \chi}(\bm{x})$ and $\cA^{\mathrm{SVRG-Policy} | \gamma', \epsilon}_{\xi, \chi}(\bm{x})$
such that the following hold true. 
\begin{itemize}[leftmargin=*]
    \item The algorithms take in the random seeds $\xi, \chi$ distributed as follows. The first random seed $\xi \sim \{i_1, ..., i_T\}$ where each $i_j \in \cS^\cA$ and each $i_j(s, a) \sim \bmp(s, a)$ and $T = \tilde{O}((1-\gamma')^{-2})$. The second random seed $\chi \sim \UniformDist{[0]}$ is a 0-bit random seed. 
    \item If $\bm{v}' = \cA^{\mathrm{TVRVI} | \gamma', \epsilon, \delta}_{\xi, \chi}({\bv})$, then wp.\ $1-\delta$ over the draw of the random seeds $\xi$ and $\chi$, ${\bv}'$ is a solution to the $(\cM, \br, \gamma', \bv, \epsilon)$-sub-problem.
    \item If $\bm{v}', \pi = \cA^{\mathrm{TVRVI-Policy} | \gamma', \epsilon, \delta}_{\xi, \chi}({\bv})$, then wp.\ $1-\delta$ over the draw of the random seeds $\xi$ and $\chi$, ${\bv}', \pi$ are a solution to the $(\cM, \br, \gamma', \bv, \epsilon)$-policy sub-problem.
    \item The algorithms make only $\tilde{O}(1)$ batch queries and make only the sample queries that are encoded in $\xi$. 
\end{itemize}
\end{theorem}

By instantiating the DMDP outer-solver PRM (Theorem~\ref{thm:dmdp-unregularizer}) with TVRVI (Theorem~\ref{thm:dmdp-sub-problem-solver}) as the sub-solver, for $\gamma' \leq \gamma$, we obtain a full batch versus sample query trade-off of $\tilde{O}((1-\gamma')/(1-\gamma))$ full batch and $\tilde{O}((1-\gamma')^{-1}(1-\gamma)^{-1})$ sample queries per state-action pair. The pseudo-code is shown in Algorithm~\ref{alg:dmdp}. 

\begin{algorithm2e}[ht!]
\DontPrintSemicolon
\SetKwInput{KwParameters}{Parameters}
\caption{$\code{PRM-TVRVI}_{\gamma', \delta}(\cM, \br, \epsilon, \gamma)$ Pseudocode}\label{alg:dmdp}
\KwInput{Matrix-vector oracle for $\bP$, simulator oracle for $\bP$, reward vector $\br \in [0, 1]^\cS$, error tolerance $\epsilon \in (0, 1/(1-\gamma)]$, discount factor $\gamma \in (0, 1)$}
\KwParameters{Failure probability $\delta$, and $0 < \gamma' < \gamma< 1$.}
Initialize $\bv_0 \gets \bm{0} \in \R^\cS$ \; 
Set sufficiently large $\nOuter = \tilde{O}((1-\gamma')/(1-\gamma))$\; 
Set $\epsilon' = \epsilon/4 \cdot (1-\gamma)/(1-\gamma')$\;
\For{each $t \in [\nOuter]$}{
    \tcp{We draw realizations of the random seed $\xi$ according to the following distribution.} 
    Draw $s_t \defeq \{i_1, ..., i_T\} \sim \cD_\xi$ where each $i_j \in \cS^\cA$ and each $i_j(s, a) \sim \bmp(s, a)$ for sufficiently large $T = \tilde{O}((1-\gamma')^{-2})$\; 
    \tcp{We include a $\chi$, a zero-bit random bit for technical consistency with \textnormal{Section~\ref{sec:pseudoindependence}}} 
    Draw $c_t \sim \cD_\chi \defeq \UniformDist{[0]}$\; 
    ${\bv}_{t-1/2} \gets \cA^{\mathrm{TVRVI} | \gamma', \epsilon', \delta/\nOuter}_{\xi=s_t, \chi=c_t}({\bv_{t-1}})$.\; 
    \tcp{The post-process is as described in \textnormal{Definition~\ref{def:mdp-post-process}} } 
    ${\bv}_{t} \gets \zeta_{\epsilon'}(\bv_{t-1}, {\bv}_{t-1/2})$\; 
}
\tcp{For the final iteration, we again draw a realization of the random seed $\xi$ according to the following distribution.} 
Draw $s \defeq \{i_1, ..., i_T\} \sim \cD_\xi$ where each $i_j \in \cS^\cA$ and each $i_j(s, a) \sim \bmp(s, a)$ for sufficiently large $T = \tilde{O}((1-\gamma')^{-2})$\; 
\tcp{Again, we include a $\chi$, a zero-bit random bit for technical consistency with \textnormal{Section~\ref{sec:pseudoindependence}}} 
Draw $c_t \sim \cD_\chi \defeq \UniformDist{[0]}$\; 
${\bv}_{\nOuter+1}, \pi_{\nOuter+1} \gets \cA^{\mathrm{TVRVI-Policy} | \gamma', \epsilon', \delta/\nOuter}_{\xi=s, \chi=c}({\bv}_{\nOuter})$.\; 
\Return{$\bm{v}_{\nOuter+1}, \pi_{\nOuter+1}$}
\end{algorithm2e}

\paragraph{The sample reusing sub-solver} 

Next, we observe that Algorithm~\ref{alg:dmdp} \emph{exactly} fits into the Meta-Algorithm~\ref{alg:meta-algorithm} framework from Section~\ref{sec:pseudoindependence}. In particular, we observe that Algorithm~\ref{alg:fsm} implements $\Outer(X, \traditional)$ where $X = (\cM, \br, \gamma)$ is an DMDP problem instance as per Definition~\ref{def:dmdp-problem}. To show that we can replace the standard sub-solver with the sample-reusing sub-solver (i.e., $\Outer(X, \reuse)$), we need to invoke Theorem~\ref{thm:AtoC}.

However, we can also observe that the DMDP outer-solver PRM in Theorem~\ref{thm:dmdp-sub-problem-solver} is \emph{by definition} robust with respect to the target function $\fsub$ (defined below), in the sense of Definition~\ref{assumption:robustness}. Further, the sub-solver TVRVI \emph{already} guarantees $\ell_\infty$ accuracy guarantees in the sense of Definition~\ref{def:approx} for the target function
\begin{align*}
    \fsub(\bv) \defeq \br - \bv^\star_{\gamma', (\gamma' - \gamma) \bP \bv}. 
\end{align*}

Consequently, we can \emph{directly} apply our sample reuse framework to reuse samples across all $\nOuter$ iterations of the for loop in Algorithm~\ref{alg:dmdp} to obtain the following improved trade-off. In particular, by invoking Theorem~\ref{thm:AtoC}, we obtain the following result. 

\begin{theorem}[{\color{ForestGreen} DMDP trade-off improvement}]\label{thm:dmdp-main} Let $\cM_{\gamma, \br}$ be a DMDP and $\gamma' < \gamma$. There is an algorithm (Algorithm~\ref{alg:dmdp-reuse}) which makes $\tilde{O}((1-\gamma')/(1-\gamma)$-batch queries and $\tilde{O}(\Atot(1-\gamma')^{-2})$-sample queries and solves the DMDP problem wp.\ $1-\delta$. 
\end{theorem}

\begin{algorithm2e}[ht!]
\DontPrintSemicolon
\SetKwInput{KwParameters}{Parameters}
\caption{$\code{PRM-TVRVI-Reuse}_{\gamma', \delta}(\cM, \br, \epsilon, \gamma)$ Pseudocode}\label{alg:dmdp-reuse}
\KwInput{Matrix-vector oracle for $\bP$, simulator oracle for $\bP$, reward vector $\br \in [0, 1]^\cS$, error tolerance $\epsilon \in (0, 1/(1-\gamma)]$, discount factor $\gamma \in (0, 1)$}
\KwParameters{Failure probability $\delta$, and $0 < \gamma' < \gamma< 1$.}
Initialize $\bv_0 \gets \bm{0} \in \R^\cS$ \; 
Set sufficiently large $\nOuter = \tilde{O}((1-\gamma')/(1-\gamma))$\; 
Set $\epsilon' = \epsilon/4 \cdot (1-\gamma)/(1-\gamma')$\;
\tcp{We draw realizations of the random seed $\xi$ according to the following distribution.} 
Draw $s_1 \defeq \{i_1, ..., i_T\} \sim \cD_\xi$ where each $i_j \in \cS^\cA$ and each $i_j(s, a) \sim \bmp(s, a)$ for sufficiently large $T = \tilde{O}((1-\gamma')^{-2})$\; 
\For{each $t \in [\nOuter-1]$}{
    \tcp{Draw $c'_t$ from the noisy distribution $\cD_{\chi'}$ as in Meta-Algorithm~\ref{alg:meta-algorithm}.} 
    Draw $c'_t \sim \cD_{\chi'}$\; 

    \tcp{Implement the noisy analog of ${\cA}^{\mathrm{TVRVI}| \gamma', \epsilon', \delta/(5\nOuter^2)}$ as in Meta-Algorithm~\ref{alg:meta-algorithm}.}
    ${\bv}_{t-1/2} \gets {\cA'}^{\mathrm{TVRVI} | \gamma', \epsilon', \delta/(5\nOuter^2)}_{\xi=s_t, \chi'=c'_t}({\bv_{t-1}})$.\; 
    \tcp{The post-process is as described in \textnormal{Definition~\ref{def:mdp-post-process}} } 
    ${\bv}_{t} \gets \zeta_{\epsilon'}(\bv_{t-1}, {\bv}_{t-1/2})$\; 
}
\tcp{For the final iteration, we again draw a realization of the random seed $\xi$ according to the following distribution.} 
Draw $s \defeq \{i_1, ..., i_T\} \sim \cD_\xi$ where each $i_j \in \cS^\cA$ and each $i_j(s, a) \sim \bmp(s, a)$ for sufficiently large $T = \tilde{O}((1-\gamma')^{-2})$\; 
\tcp{Again, we include a $\chi$, a zero-bit random bit for technical consistency with \textnormal{Section~\ref{sec:pseudoindependence}}}
Draw $c_t \sim \cD_\chi \defeq \UniformDist{[0]}$\; 
${\bv}_{\nOuter+1}, \pi_{\nOuter+1} \gets \cA^{\mathrm{TVRVI-Policy} | \gamma', \epsilon', \delta/(5\nOuter^2)}_{\xi=s, \chi=c}({\bv}_{\nOuter})$.\; 
\Return{$\bm{v}_{\nOuter+1}, \pi_{\nOuter+1}$}
\end{algorithm2e}

\paragraph{Faster algorithm for certain DMDPs} {Theorem~\ref{thm:dmdp-unregularizer} also yields another interesting implication regarding the runtime for solving a DMDP. In particular, Truncated Variance-Reduced Value Iteration \cite{jin2024truncated} is known to solve the $(\cM, \gamma', \br', \epsilon, \delta)$-sub-problem in $\tilde{O}(\nnz \bP + \Atot (1-\gamma')^{-2})$-time (Theorem 1.1 of \cite{jin2024truncated}) for $\br \in [0, 1]^\cS$.\footnote{We use $\nnz \bP$ to denote the number of nonzero entries in $\bP$.} Consequently, Theorem~\ref{thm:dmdp-sub-problem-solver} solves the DMDP problem in 
\begin{align*}
    \tilde{O}\paren{ \frac{(1-\gamma')}{(1-\gamma)} \cdot (\nnz \bP + \Atot (1-\gamma')^{-2})} = \tilde{O}\paren{ \frac{\nnz \bP (1-\gamma')}{(1-\gamma)} + \frac{\Atot}{(1-\gamma)(1-\gamma')}}
\end{align*}
time for $\gamma' \leq \gamma$. So, by selecting $1-\gamma' = \max\paren{\sqrt{\Atot / \nnz \bP}, 1-\gamma}$ to minimize \emph{runtime}, we obtain the following immediate corollary of Theorem~\ref{thm:dmdp-unregularizer}.}

\begin{theorem}[Faster runtime for solving certain MDPs]\label{thm:corollary-dmdp-unregularizer} Suppose $\nnz\bP \leq \Atot (1-\gamma)^{-2}$. Then, there is an algorithm that solves the DMDP problem (Definition~\ref{def:dmdp-problem}) in $\tilde{O}(\nnz \bP + \sqrt{\nnz \bP \Atot}(1-\gamma)^{-1})$-time. 
\end{theorem}
\begin{proof} If $\nnz \bP \leq \Atot (1-\gamma)^{-2}$, then taking 
\begin{align*}
    1-\gamma' = \max(\sqrt{\Atot/\nnz\bP}, (1-\gamma)) = \sqrt{\Atot/\nnz\bP}, 
\end{align*}
in which case the runtime becomes 
\begin{align*}
    \tilde{O}\paren{\nnz \bP + \sqrt{\nnz \bP\Atot} (1-\gamma)^{-1}}. 
\end{align*}
\end{proof}
Theorem~\ref{thm:corollary-dmdp-unregularizer} improves on \cite{jin2024truncated} in the regime where $\nnz \bP = o((1-\gamma)^{-2} \Atot)$ and improves upon vanilla value iteration (which runs in $\tilde{O}(\nnz \bP (1-\gamma)^{-1})$ in the regime where $\nnz \bP \geq \Atot$. 

\subsection{Infinite-horizon Average-reward MDPs}\label{sec:amdp}

In this section, we consider the average-reward setting, in which there is no discount factor $\gamma$.

\begin{definition}[Average-reward MDP (AMDP)] Let $\cM = (\cS, \cA, \bP)$ be an MDP and $\br \in [0, 1]^{\cA}$. We denote the associated AMDP by $\cM_{\br} = (\cM; \br)$.
\end{definition}

In the average-reward case, we define the value of a policy as follows. 

\begin{definition}[AMDP policy value] Let $\cM_{\br}$ be an AMDP and $\pi$ be a policy. The value of policy $\pi$, denoted $\bv_{\br}^\pi$, is defined as 
\begin{align*}
    \bv_{\br}^\pi(s) = \lim_{T \rightarrow \infty}\frac{1}{T}\E\Brac{\sum_{t \geq 0} \br(s_t, \pi(s_t)) | s_0 = s}. 
\end{align*}
\end{definition}

As in the DMDP case, the goal is to find an approximately optimal policy. 

\begin{definition}[{\color{ForestGreen}{AMDP problem}}]\label{def:amdp-problem} Let $\cM_{\br}$ be an AMDP and $\epsilon \in (0, 1)$ be given. We must compute, wp.\ $1-\delta$, a policy $\pi$ such that $\Vert \bv_{\br}^\pi - \bv_{ \br}^\star\Vert_\infty \leq \epsilon$. Such a policy is called $\epsilon$-optimal for $\cM_{\br}$. 
\end{definition} 

In order to extend our improved trade-offs for DMDPs to AMDPs, we leverage a result of \citet{jin2021towards}, which showed that the AMDP problem can be reduced to solving a DMDP with a sufficiently high discount factor. In the following, $\Delta$ denotes the probability simplex. 

\begin{definition}[Mixing time] Let $\cM_{\br}$ be an AMDP. $\cM_{\br}$ is said to be mixing if there exists a stationary distribution $\bm{\nu} \in \Delta^{\cS}$ such that
\begin{align*}
    \tmix \defeq \max_{\pi \in \Pi} \argmin_{t \geq 1} \sum_{s \in \cS} \max_{\bm{q} \in \Delta^{\cS}} \bmp(s, \pi(s))^\top \bm{q} - \bm{\nu} < \infty. 
\end{align*}
The quantity $\tmix$ is called the \emph{mixing time} of $\cM_{\br}$. 
\end{definition}

\begin{lemma}[Lemma 3 of \cite{jin2021towards}]\label{lemma:amdp-reduction} Let $\cM_{\br}$ be an AMDP. Suppose the mixing time of the $\cM_{\br}$ is $\tmix < \infty$. Then, for any $\epsilon > 0$ and $\gamma \in (0, 1 - \epsilon/(9 \tmix))$, an $\epsilon/(3(1-\gamma))$-optimal policy for the DMDP $\cM_{\gamma, \br}$ is also an $\epsilon$-optimal policy for the AMDP $\cM_{\br}$.
\end{lemma}

By combining Lemma~\ref{lemma:amdp-reduction} with Theorem~\ref{thm:dmdp-main}, we immediately obtain the following full batch and sample query trade-off for AMDPs. 

\begin{restatable}[{\color{ForestGreen} AMDP trade-off improvement}]{theorem}{amdpmain}\label{thm:amdp-main} Let $\cM_{\br}$ be an AMDP. Suppose the mixing time of $\cM_{\br}$ is $\tmix < \infty$. Then, for $\gamma' \leq 1 - \epsilon/(9 \tmix)$, there is an algorithm that solves the AMDP problem using $\tilde{O}((1-\gamma')\tmix/\epsilon)$ full batch queries and only $\otilde(\Atot (1-\gamma')^{-2})$ sample queries.
\end{restatable}

\subsection{Proximal Reward Method for DMDPs: Proof of Theorem~\ref{thm:dmdp-unregularizer}}\label{sec:thm} 

In this section we present the proof of Theorem~\ref{thm:dmdp-unregularizer}. First, we prove a stability result regarding the optimal value of a DMDP under a reward perturbation. 

\begin{lemma}\label{lemma:bellman-stability} Let $\br, \br' \in \R^{\cA}$ such that $\br' \leq \br$, and let $\gamma > 0$. Then, for any $\bv \in \R^{\cS}$, we have that for all $s \in \cS$,
\begin{align*}
    \bm{0} \leq \bv_{\gamma, \br}^\star - \bv_{\gamma, \br'}^\star \leq \frac{1}{1-\gamma} \cdot \paren{\max_{(s, a) \in \cA}{\br(s, a) - \br'(s, a)}} \cdot \bm{1}. 
\end{align*}
\end{lemma}
\begin{proof} By the Bellman optimality conditions (Definition~\ref{def:bellman}), for each $s \in \cS$ we have 
\begin{align*}
     \bv_{\gamma, \br}^\star(s) &= \max_{\pi \in \Pi} \paren{(\bm{I} - \gamma \bm{P}^{\pi})^{-1} \br^\pi}(s), \\
     \bv_{\gamma, \br'}^\star(s) &=\max_{\pi \in \Pi} \paren{(\bm{I} - \gamma \bm{P}^{\pi})^{-1} {\br'}^\pi}(s). 
\end{align*}
Since $(\bm{I} - \gamma \bP^\pi)^{-1}$ is a positive matrix, we have that for each $s \in \cS$,
\begin{align*}
    \bm{0} \leq \bv_{\gamma, \br}^\star(s) - \bv_{\gamma, \br'}^\star(s) \leq \max_{\pi \in \Pi} \paren{(\bm{I} - \gamma \bm{P}^{\pi})^{-1} (\br^\pi - \br'^\pi) }(s) \leq \frac{1}{1-\gamma} \cdot \max_{(s, a) \in \cA} \br(s, a) - \br'(s, a). 
\end{align*}
\end{proof}

Next, we show how to iteratively solve a sequence of $\gamma'$-discounted MDPs to solve a $\gamma$-discounted MDP. First, we prove the following lemma, which bounds the convergence rate of the scheme which solves $\cM_{\gamma, \br}$ by iteratively solving problems in $\cM_{\gamma', \br'}$ for a sequence of rewards $\br^{(1)}, ..., \br^{(T)}$. 

\begin{lemma}\label{lemma:main-unregularizer-dmdp} Let $\cM$ be an MDP and $\br \in \R_{\geq 0}^{\cA}$ be a reward vector. Let $0 < \gamma' \leq \gamma < 1$,  $\eta, \epsilon > 0$, and $T \geq 1$. Let $\br^{(0)} = \br$ and $\bv^{(0)} = \bm{0}$. For each $t \geq 1$, define
\begin{align}\label{eq:reward-iterate}
    \br^{(t)} \defeq \br - (\gamma' - \gamma)\bP \bv^{(t-1)}. 
\end{align}
Suppose that for each $t \geq 1$, $\bv^{(t)}$ satisfies
\begin{align*}
\bm{0} \leq \bv_{\gamma', \br^{(t)}}^\star - \eta \bm{1} \leq \bv^{(t)} \leq \bv_{\gamma', \br^{(t)}}^\star. 
\end{align*}
Then, for any $T \geq 1$, 
\begin{align*}
    \bm{0} \leq \bv_{\gamma, \br}^\star - \paren{ \paren{\frac{(\gamma-\gamma')}{(1-\gamma')}}^{T} \max_{s \in \cS}\paren{\bv_{\gamma, \br}^{\star} - \bv^{(0)}}(s) + \frac{(1-\gamma')}{(1-\gamma)} \eta} \cdot \bm{1} \leq \bv^{(T)} \leq \bv_{\gamma, \br}^\star. 
\end{align*}
Consequently, if $\eta \leq  \frac{(1-\gamma)}{2(1-\gamma')} \epsilon$ and $T = \tilde{\Omega}\paren{\frac{(1-\gamma')}{(1-\gamma)}}$ then $\bm{0} \leq \bv_{\gamma, \br}^\star -\epsilon \bm{1}\leq \bv^{(T)} \leq \bv_{\gamma, \br}^\star.$
\end{lemma}
\begin{proof} First, we will induct on $t$ to show that for $t\geq 0$,
\begin{align}\label{eq:induction}
    \bm{0} \leq \bv_{\gamma, \br}^\star - \paren{ \paren{\frac{(\gamma-\gamma')}{(1-\gamma')}}^{t} \max_{s \in \cS}\paren{\bv_{\gamma, \br}^{\star} - \bv^{(0)}}(s) + \sum_{j = 1}^{t} \paren{\frac{(\gamma-\gamma')}{(1-\gamma')}}^{(t-j)} \eta} \cdot \bm{1} \leq \bv^{(t)} \leq \bv_{\gamma, \br}^\star. 
\end{align}

\paragraph{Inductive proof of \eqref{eq:induction}.}{In the base case, when $t = 0$ the claim reduces to $\bm{0} \leq \bv^{(0)} \leq \bv_{\gamma, \br}^\star,$ which is trivially satisfied because $\bv^{(0)} = \bm{0}$ and $\br \in \R^\cA_{\geq 0}$. For the inductive step, we have
\begin{align*}
    \bm{0} \leq \bv_{\gamma, \br}^\star - \paren{ \paren{\frac{(\gamma-\gamma')}{(1-\gamma')}}^{t-1} \max_{s \in \cS} \paren{\bv_{\gamma, \br}^{\star} - \bv^{(0)}} + \sum_{j = 1}^{t-1} \paren{\frac{(\gamma-\gamma')}{(1-\gamma')}}^{(t-1-j)} \eta } \bm{1} \leq \bv^{(t-1)} \leq \bv_{\gamma, \br}^\star. 
\end{align*}

Next, observe that 
\begin{align}\label{eq:induction-1}
     \bv_{\gamma, \br}^\star - \bv^{(t)} &= \bv^\star_{\gamma, \br} - \bv^\star_{\gamma', \br^{(t)}} + \bv^\star_{\gamma', \br^{(t)}} - \bv^{(t)}. 
\end{align}
By the assumption on $\bv^{(t)}$,
\begin{align}\label{eq:decompose-gamma}
   \bm{0} \leq \bv^\star_{\gamma', \br^{(t)}} - \bv^{(t)} \leq \eta \bm{1}. 
\end{align}

By the Bellman optimality conditions, for each $s \in \cS$, 
\begin{align*}
    \bv_{\gamma, \br}^\star(s) &= \max_{a \in \cA_s} \br(s,a)+ \gamma \bmp(s, a)^\top \bv_{\gamma, \br}^{\star} = \max_{a \in \cA_s} \br(s,a) - (\gamma'-\gamma) \bmp(s, a)^\top \bv_{\gamma, \br}^\star + \gamma' \bmp(s, a)^\top \bv_{\gamma, \br}^{\star}. 
\end{align*}
Consequently, $\bv_{\gamma, \br}^\star = \bv_{\gamma', \br - (\gamma' - \gamma) \bP \bv_{\gamma, \br}^\star}^\star$. By substituting into \eqref{eq:induction-1} and applying the definition of $\br^{(t)}$, Lemma~\ref{lemma:bellman-stability} implies that
\begin{align*}
    \bm{0} \leq \bv_{\gamma, \br}^\star - \bv^{(t)} 
    &\leq \frac{(\gamma - \gamma')}{(1-\gamma')} \cdot \max_{(s, a) \in \Aset} (\bP (\bv^{(t-1)} - \bv^{\star}_{\br, \gamma}))(s, a) \cdot \bm{1} + \eta \bm{1} \\
    &\leq \frac{(\gamma - \gamma')}{(1-\gamma')} \cdot \max_{s \in \cS } (\bv^{(t-1)}(s) - \bv^{\star}_{\br, \gamma}(s)) \cdot \bm{1} + \eta \bm{1},
\end{align*}
where in the last step we used that $\norm{\bP}_\infty = 1$ and $\bP$ is a positive matrix. Consequently, 
\begin{align*}
    \bm{0} \leq \bv_{\gamma, \br}^\star - \paren{ \paren{\frac{(\gamma-\gamma')}{(1-\gamma')}}^{t} \max_{s \in \cS}(\bv_{\gamma, \br}^{\star} - \bv^{(0)})(s) + \sum_{j = 1}^t \paren{\frac{(\gamma-\gamma')}{(1-\gamma')}}^{t-j} \eta } \cdot \bm{1} \leq \bv^{(t)} \leq \bv_{\gamma, \br}^\star. 
\end{align*}
This completes the inductive argument. }

Now, we bound the geometric series as follows 
\begin{align*}
    \sum_{j = 1}^t \paren{\frac{(\gamma-\gamma')}{(1-\gamma')}}^{t-j} \eta = \sum_{j = 1}^t \paren{1 - \frac{(1-\gamma)}{(1-\gamma')}}^{t-j} \eta \leq \frac{(1-\gamma')}{(1-\gamma)} \eta. 
\end{align*}
Finally, note that when $T = \tilde{\Omega}\paren{\frac{(1-\gamma')}{(1-\gamma)}}$ and $\eta \leq \frac{(1-\gamma)}{2(1-\gamma')} \epsilon$, 
\begin{align*}
    \bv_{\gamma, \br}^\star - \epsilon \leq \bv^{(T)} \leq \bv_{\gamma, \br}^\star. 
\end{align*}
\end{proof}

Finally, we can leverage Lemma~\ref{lemma:main-unregularizer-dmdp} to complete the proof of Theorem~\ref{thm:dmdp-unregularizer}. 

\prm* 
\begin{proof} Note that the outer process in Algorithm~\ref{alg:dmdp} (and Algorithm~\ref{alg:dmdp-reuse}) ensures the following. Suppose each $\bv_{t-1/2}$ is a solution to the $(\cM, \br, \gamma', \bv_{t-1}, \epsilon/4 \cdot (1-\gamma)/(1-\gamma'))$-sub-problem. Then due to the post-process, each $\bv_{t}$ meets the conditions on $\bv^{(t)}$ ins Lemma~\ref{lemma:main-unregularizer-dmdp} for $\eta = \epsilon/2 \cdot (1-\gamma)/(1-\gamma')$.

Consequently, by Lemma~\ref{lemma:main-unregularizer-dmdp},  we have that for sufficiently large $\nOuter = \tilde{O}((1-\gamma')/(1-\gamma))$,
\begin{align*}
    \bm{0} &\leq \bv_{\gamma, \br}^\star - \frac{\epsilon}{2} \bm{1} \leq \bv_{\nOuter} \leq \bv_{\gamma, \br}^\star, \\
    \bm{0} &\leq \bv_{\gamma, \br}^\star - \frac{\epsilon}{2} \bm{1} \leq \bv_{\nOuter+1} \leq \bv_{\gamma, \br}^\star, 
\end{align*}
and by the definition of the policy-sub-problem, we can further conclude
\begin{align}\label{eq:uselater}
    \bm{0} &\leq \bv_{\gamma', \br^{\nOuter+1}}^{\star} - \frac{\epsilon}{2} \bm{1} \leq \bv_{\nOuter+1} \leq \bv_{\gamma', {\br}^{\nOuter+1}}^{{\pi_{\nOuter+1}}} \leq \bv^\star_{\gamma, \br},
\end{align}
where 
\begin{align*}
    {\br}^{\nOuter+1} \defeq \br - (\gamma' - \gamma) \bP \bv_{\nOuter}. 
\end{align*}
Now, note that for all $s \in \cS$, 
\begin{align*}
    \bv_{\gamma, \br}^{\pi_{\nOuter+1}}(s) &= \br(s,{\pi_{\nOuter+1}}(s))+ \gamma \bmp(s, {\pi_{\nOuter+1}}(s))^\top \bv_{\gamma, \br}^{{\pi_{\nOuter+1}}} \\
    &= \br(s,{\pi_{\nOuter+1}}(s)) - (\gamma'-\gamma) \bmp(s, {\pi_{\nOuter+1}}(s))^\top \bv_{\gamma, \br}^{\pi_{\nOuter}} + \gamma' \bmp(s, {\pi_{\nOuter+1}}(s))^\top \bv_{\gamma, \br}^{{\pi_{\nOuter+1}}}. 
\end{align*}
Above, in the first line we used the Bellman formulation for values of policies (Definition~\ref{def:bellman}). 

Thus, $\bv_{\gamma, \br}^{\pi_{\nOuter+1}} = \bv_{\gamma', \br - (\gamma'-\gamma)\bP \bv^{\pi_{\nOuter+1}}_{\gamma, \br}}^{\pi_{\nOuter+1}}.$
Since $(\bm{I} - \gamma' \bP^{\pi_{\nOuter+1}})^{-1}$ and $\bP^{\pi_{\nOuter+1}}$ are positive matrices and $(\gamma-\gamma')/(1-\gamma') \leq 1$ we have that 
\begin{align*}
    \bv_{\gamma', \br - (\gamma'-\gamma)\bP \bv^{\pi_{\nOuter+1}}_{\gamma, \br}}^{\pi_{\nOuter+1}} - \bv_{\gamma', {\br}^{\nOuter+1}}^{{\pi_{\nOuter+1}}} &= (\bI - \gamma' \bP^{\pi_{\nOuter+1}})^{-1} ((\gamma - \gamma') \bP^{\pi_{\nOuter}} (\bv_{\gamma, \br}^{\pi_{\nOuter}} - \bv_{\nOuter} )) \\
    &\leq (\bI - \gamma' \bP^{\pi_{\nOuter+1}})^{-1} ((\gamma - \gamma') \bP^{\pi_{\nOuter}} (\bv_{\gamma, \br}^{\star} - \bv^{(\nOuter)} )) \\
    &\leq \max_{(s, a) \in \Aset} (\bv_{\gamma, \br}^{\star}(s, a) - \bv_{\nOuter}(s, a) ) \\
    &\leq (\gamma - \gamma')/(1-\gamma') \epsilon/2 \\
    &\leq \frac{\epsilon}{2}, 
\end{align*}
where the second to last step used that $\normInline{\bI - \gamma'\bP^{\pi}}_\infty \leq 1/(1-\gamma')$ for any policy $\pi$. Consequently, combining with \eqref{eq:uselater}, we conclude that $\bv_{\gamma, {\br}}^\star - \epsilon \leq \bv_{\gamma,\br}^{{\pi_{\nOuter}}} \leq \bv_{\gamma, \br}^{\star}$, which completes the proof.
\end{proof} 

%% file: minimax.tex
\section{Application: Matrix games and minimax problems}\label{sec:minimax_unified}

In this section, we consider minimax problems, including $\ell_2$-$\ell_1$ and $\ell_2$-$\ell_2$ matrix-games and finite-sum minimax problems. In Section~\ref{sec:preliminaries} we discuss general preliminaries. In Section~\ref{sec:matrix-games}, we discuss $\ell_2$-$\ell_1$ matrix games in Section~\ref{subsec:l2l1} and $\ell_2$-$\ell_2$ matrix games in Section~\ref{subsec:l2l2}. 
Section~\ref{subsec:applications} discusses applications of our $\ell_2$-$\ell_1$ matrix games improvements for two computational geometry problems: maximum inscribed ball and minimum enclosing ball. 

\subsection{Preliminaries}\label{sec:preliminaries}
We first outline preliminaries of the minimax problems we consider. The notation in this subsection is consistent with that of \cite{carmon2019variance}. 

\paragraph{Problem setup.}{A \emph{setup} is the triplet $(\cZ = \cX \times \cY, \norm{\cdot}, r)$ where we use $\cZ \defeq \cX \times \cY$, where (1) $\cX$ is a compact and convex subset of $\R^n$ and $\cY$ is a compact and convex subset of $\R^m$; (2) $\norm{\cdot}$ is a norm on $\cZ$, and (3) $r$ is a 1-strongly convex function with respect to $\cZ$ and $\norm{\cdot}$. We can $r$ a \emph{distance generating function} and denote the associated Bregman divergence as 
\begin{align*}
    V_{\bz}(\bz') \defeq r(\bz') - r(\bz) - \langle \nabla r(\bz), \bz'-\bz \rangle \geq \frac{1}{2} \norm{\bz' - \bz}^2. 
\end{align*}
We also denote $\Theta \defeq \max_{\bz'} r(\bz') - \min_{\bz'} r(\bz)$ and assume it is finite. We use $\norm{\cdot}_*$ to denote the dual norm of $\norm{\cdot}$. For $\bz \in \cZ$, we often write $\bz^\cX \in \cX$, $\bz^\cY \in \cY$ to denote the first $n$ and last $m$ coordinates of $\bz$, respectively. We use $d = m + n$ throughout this section. 

We assume that $\cZ$ is sufficiently simple such that given any $\bz' \in \R^d$, one can compute the projection of $\bz'$ onto $\cZ$ with respect to $\normInline{\cdot}$, denoted $\project_\cZ(\bz')$, in $\tilde{O}(d)$-time. (This is true, for example, for the Euclidean unit ball, or the probability simplex, which are the relevant cases for $\ell_2$-$\ell_2$ matrix-games and $\ell_2$-$\ell_1$ matrix games.) Finally, we assume that $c \normInline{\cdot}_\infty \leq \normInline{\cdot} \leq C \normInline{\cdot}_\infty$ over $\cZ$, for some $c, C = \poly(d)$. (This is true, for example, for the $\ell_1$ and $\ell_2$ norms, which are the relevant cases for $\ell_2$-$\ell_2$ matrix-games and $\ell_2$-$\ell_1$ matrix games.)
} 

\paragraph{Minimax problems.} {We consider minimax (saddle-point) problems of the form 
\begin{align*}
    \min_{\bm{x} \in \cX} \max_{\bm{y} \in \cY} f(x, y), 
\end{align*}
for some setup $(\cZ = (\cX, \cY), \norm{\cdot}, r)$. We use $g(\bm{z}) \defeq (\nabla_x f(\bm{z}), -\nabla_y f(\bm{z})) \in \R^d$ to denote the \emph{gradient mapping} of $f$. We use $L$ to denote the Lipschitz constant of $g$, $D \defeq \max_{\bz, \bz' \in \cZ} \norm{\bz - \bz'}$, and $G \defeq \max_{\bz \in \cZ} \normInline{g(z)}$. We assume that $D, L, G$ are finite and hide polylogarithmic dependencies in these parameters inside of $\tilde{O}(\cdot)$ notation. 

}

Next, we define the minimax problem we consider in this Section.

\begin{definition}[{\color{ForestGreen} Minimax problem}]\label{def:minimax-def} In the minimax problem, we are given $f : \cZ = (\cX \times \cY) \to \R$, $\epsilon > 0$, and $\delta \in (0, 1).$ We must compute $\bm{x}, \bm{y}$ such that
\begin{align*}
    \max_{\bm{y}'\in\cY} f(\bm{x}, \bm{y}') - \min_{\bm{x}' \in \cX} f(\bm{x}', \bm{y}) \leq \epsilon. 
\end{align*}
\end{definition}

We will later discuss matrix-games and finite-sum minimax problems as special cases of Definition~\ref{def:minimax-def} and discuss the relevant batch and sample query models therein. However, in this section we discuss the outer-solver and sub-problem structure for general minimax problems following the conceptual proximal point framework of \cite{carmon2019variance, nemirovski2004prox}.  

\paragraph{Conceptual proximal point.}

\citet{carmon2019variance, nemirovski2004prox} showed how to solve minimax problems of the form of Definition~\ref{def:minimax-def} by iteratively solving a series of $\alpha$-regularized sub-problems. This method is inspired by Nemirovski's ``conceptual prox point method'' \citep{nemirovski2004prox}. Correspondingly, we define the minimax sub-problem as follows.

\begin{definition}[{\color{ForestGreen} Minimax sub-problem}]\label{def:minimax-subproblem} Let $f$ be as in Definition~\ref{def:minimax-def}. Let $\bm{z}_0 \in \cZ, \alpha > 0, \epsilon> 0$. Let $\bz^\alpha$ be the solution to the problem $\min_{\bx' \in \cX} \max_{\by' \in \cY} f(\bz) + \alpha V_{\bz}(\bx', \by')$. In the $(\bm{z}_0, \alpha, \epsilon)$-sub-problem we define $\fsub(\bz) \defeq \bz_\alpha$ to be the unique point in $\cZ$ such that 
\begin{align*}
    \langle g(\bz_\alpha) + \alpha \nabla V_{\bz}(\bz_\alpha), \bz_\alpha - \bu \rangle, ~~\text{ for all } \bu \in \cZ. 
\end{align*}
We must output a $\bz' \in \R^{d}$ such that $\norm{\bz' - \fsub(\bz)}_\infty \leq \epsilon$. 
\end{definition}

\citet{carmon2019variance, nemirovski2004prox} showed how to solve the minimax problem (Definition~\ref{def:minimax-def}) by solving a sequence of sub-problems of the form of Definition~\ref{def:minimax-subproblem} using Conceptual Proximal Point (CPP) as the outer-solver. Here, the outer process is a projection step (to ensure feasibility) followed by an extragradient step. 

\begin{definition}[{\color{ForestGreen} Minimax post-process}]\label{def:minimax-post-process} 
Let $f$ be as in Definition~\ref{def:minimax-subproblem}, and $\alpha >0$ For any ${\bz, \bz'} \in \R^d \times \R^d$, we define $\bz'' \defeq \project_{\cZ}(z')$
\begin{align*}
    \zeta({\bz}, {\bz}') \defeq \argmin_{\tilde{\bz} \in \cZ} \paren{\langle g(\bz''), \tilde{\bz} \rangle + \alpha V_{\bz}(\tilde{\bz})}.
\end{align*}
\end{definition}

We now specify the outer-solver framework using CPP as follows. We slightly restate a variant of Proposition 4 of \cite{carmon2019variance}. 

\begin{theorem}[CPP, Proposition 4 of \cite{carmon2019variance}, adapted - {\color{ForestGreen} Minimax outer-solver}]\label{thm:conceptual-prox} Let $f$ be as in Definition~\ref{def:minimax-def} and $\alpha > 0$. Consider Meta-Algorithm~\ref{alg:meta-algorithm} with the following instantiation of parameters. For fixed $\alpha > 0$, 
\begin{itemize}
    \item Initialize $\bm{u}_0$ with $\argmin_{\bz \in \cZ} r(\bz)$.
    \item Define $\zeta$ as in Definition~\ref{def:minimax-post-process}.
    \item Set $\bm{w}(t) \defeq 1/\nOuter$ for each $t \in [\nOuter]$. 
\end{itemize}
Then, there is an $\epsilon' = \poly(G, L, D, \Theta, \epsilon, d)$ such that the following holds. Suppose that in each iteration $t \in [\nOuter]$ of Algorithm~\ref{alg:meta-algorithm}, $\bm{u}_{t-1/2}$ is a solution to the $({\ba}_{t-1}, \alpha, \epsilon')$-sub-problem; then $\bu_{\nOuter}$ is a solution to the minimax problem.
\end{theorem}
\begin{proof} To prove this, we appeal to Proposition 4 of \cite{carmon2019variance}. By \citet{carmon2019variance}'s Proposition 4, it is sufficient to show that $\bu_{t-1/2}$ satisfies
\begin{align}\label{eq:sufficient-condition-mg}
    \langle g(\project_\cZ(\bu_{t-1/2})), \project_\cZ(\bu_{t-1/2}) - \bu \rangle - \alpha V_{\bu_{t-1}}(\bu) \leq \epsilon ~~\text{ for all } \bu \in \cZ.
\end{align}

Consequently, to prove the theorem, it suffices to show that whenever $\ba_{t-1/2}$ is a solution to the $({\ba}_{t-1}, \alpha, \epsilon')$-sub-problem, \eqref{eq:sufficient-condition-mg} holds. To this end, let $\bz_\alpha$ be the unique point in $\cZ$ such that 
\begin{align*}
    \langle g(\bz_\alpha) + \alpha \nabla V_{\bu_{t-1}}(\bz_\alpha), \bz_\alpha - \bu \rangle \leq 0, ~~\text{ for all } \bu \in \cZ. 
\end{align*}
By the three-point-equality for Bregman divergences, we have that 
\begin{align*}
    \langle g(\bz_\alpha), \bz_\alpha - \bu \rangle - \alpha V_{\bu_{t-1}}(\bu) \leq - \alpha V_{\bz_\alpha}(\bu) - \alpha V_{\bu_{t-1}}(\bz_\alpha), ~~\text{ for all } \bu \in \cZ. 
\end{align*}
Now, since $\normInline{\bu_{t-1/2} - \bz_{\alpha}}_\infty \leq \epsilon$ and $\bz_\alpha \in \cZ$, we have that
\begin{align*}
    \normInline{\project_\cZ(\bu_{t-1/2}), \bz_\alpha} \leq \normInline{\bu_{t-1/2}, - \bz_\alpha}
\end{align*}
Consequently, by equivalence of norms, 
\begin{align*}
    c \normInline{\project_\cZ(\bu_{t-1/2}), \bz_\alpha}_\infty \leq \normInline{\project_\cZ(\bu_{t-1/2}), \bz_\alpha} \leq \normInline{\bu_{t-1/2}, - \bz_\alpha} \leq C \normInline{\bu_{t-1/2}, - \bz_\alpha}_\infty. 
\end{align*}
Thus, 
\begin{align*}
    \normInline{\project_\cZ(\bu_{t-1/2}), \bz_\alpha}_\infty \leq \normInline{\project_\cZ(\bu_{t-1/2}), \bz_\alpha} \leq \normInline{\bu_{t-1/2}, - \bz_\alpha} \leq C/c \cdot \normInline{\bu_{t-1/2}, - \bz_\alpha}_\infty. 
\end{align*}

Thus, for any $\bu \in \cZ$ we can write 
\begin{align*}
     &\langle g(\project_\cZ(\bu_{t-1/2})), \project_\cZ(\bu_{t-1/2}) - \bu \rangle \\
     &= \langle g({\bz}_{\alpha}), \project_\cZ(\bu_{t-1/2}) - \bu \rangle + \langle g(\project_\cZ(\bu_{t-1/2})) - g({\bz}_{\alpha}), \project_\cZ(\bu_{t-1/2} - \bu \rangle \\
     &= \langle g(\project_\cZ(\bu_{t-1/2})), \bz_\alpha - \bu \rangle + \langle g({\bz}_{\alpha}), \project_\cZ(\bu_{t-1/2}) - \bz_\alpha \rangle  \\
     &~~-\langle g(\project_\cZ(\bu_{t-1/2})) - g({\bz}_{\alpha}), \project_\cZ(\bu_{t-1/2}) - \bu \rangle \\
     &\leq \langle g({\bz}_{\alpha}), \bz_\alpha - \bu \rangle + \norm{g(\bz_\alpha)}_1 C/c \cdot \epsilon' + \normInline{g(\bz_\alpha) - g(\project_\cZ(\bu_{t-1/2}))}_* D
\end{align*}
where in the last line we used Holder's inequality. Now, using the Lipschiztness of $g$ and equivalence of norms, 
\begin{align*}
    \normInline{g(\bz_\alpha) - g(\project_\cZ(\bu_{t-1/2}))}_* \leq L \normInline{\bz_\alpha - \bu_{t-1/2}} \leq L C\normInline{\bz_\alpha - \bu_{t-1/2}}_\infty \leq L C \epsilon'. 
\end{align*}
And again, using 
Consequently, taking $\epsilon'$ to be a sufficiently small polynomial in $L, D, \epsilon, d$ is enough to ensure that \eqref{eq:sufficient-condition-mg} holds, i.e., that 
\begin{align*}
    \langle g(\project_\cZ(\bu_{t-1/2})), \project_\cZ(\bu_{t-1/2}) - \bu \rangle - \alpha V_{\bu_{t-1}}(\bu) \leq GC/c \cot \epsilon' + DLC\epsilon' \leq \epsilon. 
\end{align*}
\end{proof}

In our applications, we often leverage the following fact about Bregman divergences. (Recall that $c$ is defined in Section~\ref{sec:preliminaries}).
\begin{fact}\label{fact:minimax} Let $r : \cZ \to \R$ be a 1-strongly convex function with respect to $\cZ$ and $\normInline{\cdot}$. Then, 
\begin{align*}
        V_{\bz}(\bz') \geq \frac{1}{2} \norm{\bz' - \bz}^2 \geq \frac{c^2}{2} \norm{\bz' - \bz}_\infty^2. 
\end{align*}
\end{fact}

The specific sub-problem solvers will vary between $\ell_2$-$\ell_2$ matrix-games, $\ell_2$-$\ell_1$ matrix-games, and finite-sum minimax problems. However, observe that Theorem~\ref{thm:conceptual-prox} already establishes that CPP~ is robust to bounded $\ell_\infty$ error in the sub-problem solutions, in the sense of Definition~\ref{assumption:robustness}. Consequently, it is amenable to using our sample-reuse framework developed in Section~\ref{sec:pseudoindependence}. We discuss this in the following sections. 

\subsection{Matrix games}\label{sec:matrix-games}

We consider a matrix $\bm{A} \in \R^{m \times n}$ and use $\bm{a}_i$, $\bm{a}^j$ to denote the $i$-th row and $j$-th column of $\bm{A}$, respectively. We use $A_{i,j}$ to denote the $(i,j)$-th entry of $A$. We use $\smash{\norm{\bm{A}}_{2 \rightarrow \infty} \defeq \max_{i \in [m]} \norm{\bm{a}_i}_2}$ and $\norm{\bm{A}}_F$ to denote the Frobenius norm. 

To discuss the application of our pseudo-independence results for improved oracle complexity trade-offs on matrix games, we first restrict to minimax problems on functions $f$ corresponding to composite matrix games. 

\begin{definition}[{\color{ForestGreen} MG problem}]\label{def:mg-def} In the matrix-game problem, we are given a setup $(\cZ = \cX \times \cY, \norm{\cdot}, r)$; a matrix $\bm{A}\in \R^{m \times n}$; convex, differentiable functions $\phi: \cX \to \R, \psi: \cY \to \R$; $\epsilon > 0$; and $\delta \in (0, 1)$. We must solve the minimax problem (as in Definition~\ref{def:minimax-def}) for the function
\begin{align*}
    f(\bx, \by) = \by^\top \bm{A} \bx + \phi(\bx) - \psi(\by). 
\end{align*}
\end{definition}

Next, we define the relevant full batch and sample query oracles for the matrix games problems we consider. As in \cite{carmon2019variance} we assume that $\phi$ and $\psi$ are explicit and that $\bm{A}$ can be accessed via oracle queries. Concretely, we define one batch oracle and two types of sample oracles for accessing $\bm{A}$. 

\begin{definition}[Matrix-vector oracle - {\color{ForestGreen} MG batch oracle}]\label{def:batch-oracle-mg} When queried with $(\bx, \by) \in \cZ$, a matrix-vector oracle for $\bm{A}$ returns $(\bm{A}x, \bm{A}^\top \bm{y})$. 
\end{definition}

\begin{definition}[Row/column oracle - {\color{ForestGreen} MG sample oracle, Type I}]\label{def:sample-oracle-type2} When queried with $(i,)$ for $i \in [m]$, the oracle returns $\bm{a}_{i}$, the $i$-th row of $\bm{A}$. When queried with $(,j)$ for $j \in [m]$, the oracle returns $\bm{a}^{j}$, the $j$-th column of $\bm{A}$.
\end{definition}

\begin{definition}[Entry oracle - {\color{ForestGreen} MG sample oracle Type II}]\label{def:sample-oracle-type1} When queried with $(i,j) \in [m] \times [n]$, the entry oracle returns $A_{i,j}.$
\end{definition}

\subsubsection{$\ell_2$-$\ell_1$ matrix games}\label{subsec:l2l1}

In this section, we discuss the $\ell_2$-$\ell_1$ matrix games setting. Throughout this section, we use the setup $(\cZ = (\cX, \cY), \norm{\cdot}, r)$ where 
\begin{itemize}[leftmargin=*]
    \item $\cX = \mathbb{B}^{n} \defeq \{\bx \in \R^n : \norm{\bx}_2 \leq 1\}$ is the Euclidean ball of radius 1 centered at the origin; and $\cY = \Delta^m$ is the $m$-dimensional simplex. 
    \item $\norm{\cdot}: \cZ \to \R$ is given by $\norm{\bm{z}} = \sqrt{\norm{\bz^{\cX}}_2^2 + \norm{\by^{\cY}}_1^2}$. 
    \item $r: \cZ \to \R$ is given by $r(\bm{z}) = \frac{1}{2} \norm{\bz^{\cX}}_2^2 + \sum_{i \in [m]} \bz^\cY(i) \log(\bz^\cY(i))$. 
\end{itemize}

In this case, the relevant constants defined in Section~\ref{sec:preliminaries} are trivially given by 
\begin{itemize}[leftmargin=*]
    \item $L \leq \norm{\bm{A}}_{2 \rightarrow \infty}$, $G \leq \max_{i,j} |A_{i,j}|$
    \item $D = 1$
    \item $c = 1$, $C = d^2$
\end{itemize}

The function $r$ is known to be 1-strongly convex with respect to $\cZ$, $\norm{\cdot}$ (see, e.g., Section 4.2 of \cite{carmon2019variance}). To begin, we define the distributions which we will sample from in order to make sample queries. 

\begin{definition} Let $\bm{A} \in \R^{m \times n}$. For each $i \in [m]$, define the distribution $\Dentry(i)$ as follows:
\begin{align*}
    \probSub{b=j}{b \sim \Dentry(i)} = \frac{A_{i,j}^2}{\norm{\bm{a}_i}_2^2}.
\end{align*}
\end{definition}

The following theorem states the guarantees of the variance-reduced mirror descent (Algorithm 4 of \cite{carmon2019variance}, VRMD1) solves a minimax-subproblem in the $\ell_2$-$\ell_1$ matrix-games setting using $\tilde{O}(1)$ full batch queries, a number of non-oblivious, i.e., adaptive sample queries of Type I (see Definition~\ref{def:sample-oracle-type2}), and a number of \emph{oblivious} sample queries of Type II (see Definition~\ref{def:sample-oracle-type1}). 

\begin{theorem}[VRMD1, Theorem 2 of \cite{carmon2019variance}, adapted - {\color{ForestGreen} $\ell_2$-$\ell_1$ MG sub-solver}]\label{thm:sub-problem-solver-l2l1} Let $f$ be as defined in Definition~\ref{def:mg-def} and $\bm{z} \in \cZ$. There is a randomized algorithm $\cA^{\mathrm{VRMD1-HP} | \alpha, \epsilon, \delta}_{\xi, \chi}$ such that the following holds. 
\begin{itemize}[leftmargin=*]
    \item The algorithm takes in the random seeds $\xi, \chi$ distributed as follows, for some sufficiently large $T = \tilde{O}(\norm{\bm{A}}_{2 \rightarrow \infty}/\alpha^2)$. The first random seed $\xi \sim \{\{({i_1}_t, {j_1}_t) ..., ({i_m}_t, {j_m}_t)\}_{t=1}^T\} \subset [m]\times[n]$ where for each $q \in [m], t \in [T]$, ${i_q}_t = q$ and ${j_q}_t \sim \Dentry(q)$. The second random seed $\chi$ is sampled \emph{adaptively} based on $\bz$ and is used to make some additional \emph{adaptive} sample queries of Type I (row/column oracle queries). 
    \item If $\bm{z}' = \cA^{\mathrm{VRMD1-HP} | \alpha, \epsilon, \delta}_{\xi, \chi}({\ba})$, then wp.\ $1-\delta$ over the draw of the random seeds $\xi$ and $\chi$, $\bm{z}'$ solves the $(\bz, \alpha, \epsilon)$-minimax-subproblem. 
    \item $\cA^{\mathrm{VRMD1-HP} | \alpha, \epsilon, \delta}_{\xi, \chi}$ makes only $\tilde{O}(1)$ batch queries, makes sample queries of Type II only the indices encoded in the seed $\xi$, and makes sample queries of Type I \emph{only} on the indices encoded in the seed $\chi$.
\end{itemize} 
\end{theorem}
\begin{proof} The proof follows directly from Theorem 2 of \cite{carmon2019variance} and Fact~\ref{fact:minimax}. 
\end{proof} 

By instantiating CPP (Theorem~\ref{thm:conceptual-prox}) with VRMD1 (Theorem~\ref{thm:sub-problem-solver-l2l1} as the sub-problem solver in the $\ell_2$-$\ell_1$ setup, we see that we can obtain the following query complexity trade-of of $\tilde{O}(\alpha/\epsilon)$ full batch queries; along with $\tilde{O}(\norm{A}_{2 \rightarrow \infty}^2 \alpha^{-1}\epsilon^{-1})$ non-oblivious sample queries of Type I; and $\tilde{O}(m\norm{A}_{2 \rightarrow \infty}^2 \alpha^{-2})$ oblivious sample queries of Type II. This, corresponds to instantiating Meta-Algorithm~\ref{alg:meta-algorithm} with $S = \traditional$ and:
\begin{itemize}[leftmargin=*]
    \item $\bu_0 = \argmin_{z \in \cZ} r(z)$
    \item $\nOuter = \tilde{O}(\alpha/\epsilon)$
    \item $\zeta$ as defined in Definition~\ref{def:minimax-post-process}
    \item $\cA_{\xi, \chi}^{\mathrm{VRMD1-HP}}$ as the sub-solver
    \item $\cD_\xi$ and $\cD_\chi$ as in Theorem~\ref{thm:sub-problem-solver-l2l1}
    \item $\fsub(\bz)$ as defined in Definition~\ref{def:minimax-subproblem}
\end{itemize}

Moreover, observe that Theorem~\ref{thm:conceptual-prox} ensures that CPP is $\ell_\infty$ robust with respect to $\fsub$ (in the sense of Definition~\ref{assumption:robustness}), and VRMD1-HP solves sub-problems to high-precision in the $\ell_\infty$ norm (in the sense of Definition~\ref{def:approx}). Thus, using our sample-reuse framework, Theorem~\ref{thm:AtoC} implies that we can \emph{reuse} the oblivious sample queries of Type II across all $\nOuter$ iterations of CPP. We obtain the following improved trade-off. 

\begin{theorem}[{\color{ForestGreen} $\ell_2$-$\ell_1$ MG trade-off improvement}]\label{thm:mg-trade-off-l2l1-improvement} For $\alpha > 0$, there is an algorithm that wp.\ $1 - \delta$ solves the MG problem in the $\ell_2$-$\ell_1$ setup using $\tilde{O}(\alpha/\epsilon)$ full batch queries; $\tilde{O}(\norm{\bm{A}}_{2\rightarrow \infty}^2 \alpha^{-1}\epsilon^{-1})$ sample queries of Type I; 
and $\otilde(m \norm{\bm{A}}_{2\rightarrow \infty}^2 \alpha^{-2})$ sample queries of Type II.
\end{theorem}

\subsubsection{$\ell_2$-$\ell_2$ matrix games}\label{subsec:l2l2}

In this section, we discuss the $\ell_2$-$\ell_2$ matrix games setting. Throughout this section, we use the setup $(\cZ = (\cX, \cY), \norm{\cdot}, r)$ where 
\begin{itemize}[leftmargin=*]
    \item $\cX = \mathbb{B}^{n} \defeq \{\bx \in \R^n : \norm{\bx}_2 \leq 1\}$ and $\cY = \mathbb{B}^m \defeq \{\by \in \R^m : \norm{\by}_2 \leq 1\}$ are the Euclidean balls of radius 1 centered at the origin. 
    \item $\norm{\cdot}: \cZ \to \R$ is given by $\norm{\bm{z}} =\norm{\bz}_2$. 
    \item $r: \cZ \to \R$ is given by $r(\bm{z}) = \frac{1}{2} \norm{\bz}_2^2$. 
\end{itemize}

In this case, the relevant constants defined in Section~\ref{sec:preliminaries} are trivially given by 
\begin{itemize}[leftmargin=*]
    \item $L, G \leq \norm{\bm{A}}_{2}$
    \item $D = 1$
    \item $c = 1$, $C = d$
\end{itemize}

The function $r$ is known to be 1-strongly convex with respect to $\cZ$, $\norm{\cdot}$ (see, e.g., Section 4.3 of \cite{carmon2019variance}). To begin, we define the distributions which we will sample from in order to make sample queries. 

\begin{definition} Let $\bm{A} \in \R^{m \times n}$. For each $i \in [m]$, define the distributions $\Drow$ and $\Dcol$ as follows:
\begin{align*}
    \probSub{a=i}{a \sim \Drow} = \frac{\norm{\bm{a}_i}_2^2}{\norm{\bm{A}}_F^2}, \quad \text{ and } \quad \probSub{a=j}{a \sim \Dcol} = \frac{\norm{\bm{a}^j}_2^2}{\norm{\bm{A}}_F^2}. 
\end{align*}
\end{definition}

The following theorem states the guarantees of the variance-reduced mirror descent. VRMD-2 solves a subproblem usin $\tilde{O}(1)$ full batch queries and \emph{oblivious} sample queries of Type I (see Definition~\ref{def:sample-oracle-type2}). 

\begin{theorem}[VRMD2, Lemma 5 of \cite{carmon2019variance}, adapted - {\color{ForestGreen} $\ell_2$-$\ell_1$ MG sub-solver}]\label{thm:sub-problem-solver-l2l2} Let $f$ be as defined in Definition~\ref{def:mg-def} and $\bm{z} \in \cZ$. There is a randomized algorithm $\cA^{\mathrm{VRMD2-HP} | \alpha, \epsilon, \delta}_{\xi, \chi}$ such that the following holds. 
\begin{itemize}[leftmargin=*]
    \item The algorithm takes in the random seeds $\xi = (\xi_1, \xi_2)$ distributed as follows, for some sufficiently large $T = \tilde{O}(\norm{\bm{A}}_{F}^2/\alpha^2)$.  $\xi_1 \sim \{i_1, ..., i_T\} \subset [m]$ and $\xi_2 \sim \{j_1, ..., j_T\} \subset [n]$, where for each $t \in [T]$, $i_t \sim \Drow$ and $j_t \sim \Dcol$. The second random seed $\chi \sim \UniformDist{[0]}$ is a zero-bit random seed. 
    \item If $\bm{z}' \cA^{\mathrm{VRMD2-HP} | \alpha, \epsilon, \delta}_{\xi, \chi}({\ba})$, then wp.\ $1-\delta$ over the draw of the random seeds $\xi$ and $\chi$, $\bm{z}'$ solves the $(\bz, \alpha, \epsilon)$-minimax-subproblem. 
    \item $\cA^{\mathrm{VRMD2-HP} | \alpha, \epsilon, \delta}_{\xi, \chi}$ makes only $\tilde{O}(1)$ batch queries, and makes row/column queries (sample queries of Type I) only the indices encoded in the seed $\xi$. The algorithm makes \emph{no} entry oracle queries (sample queries of Type II). 
\end{itemize} 
\end{theorem}
\begin{proof} The proof follows directly from Lemma 5 of \cite{carmon2019variance} and Fact~\ref{fact:minimax}. 
\end{proof} 

By instantiating CPP (Theorem~\ref{thm:conceptual-prox}) with VRMD2 (Theorem~\ref{thm:sub-problem-solver-l2l2}) as the sub-problem solver in the $\ell_2$-$\ell_2$ setting, we see that we can obtain the following query complexity trade-off of $\tilde{O}(\alpha/\epsilon)$ full batch queries along with $\tilde{O}(m\norm{A}_{F}^2 \alpha^{-1}\epsilon^{-1})$ oblivious sample queries of Type II. This, corresponds to instantiating Meta-Algorithm~\ref{alg:meta-algorithm} with $S = \traditional$ and:
\begin{itemize}[leftmargin=*]
    \item $\bu_0 = \argmin_{z \in \cZ} r(z)$
    \item $\nOuter = \tilde{O}(\alpha/\epsilon)$
    \item $\zeta$ as defined in Definition~\ref{def:minimax-post-process}
    \item $\cA_{\xi, \chi}^{\mathrm{VRMD2-HP}}$ as the sub-solver
    \item $\cD_\xi$ and $\cD_\chi$ as in Theorem~\ref{thm:sub-problem-solver-l2l2}
    \item $\fsub(\bz) \defeq \bz_\alpha$ as defined in Definition~\ref{def:minimax-subproblem}
\end{itemize}

Moreover, observe that Theorem~\ref{thm:conceptual-prox} ensures that CPP is $\ell_\infty$ robust with respect to $\fsub$ (in the sense of Definition~\ref{assumption:robustness}), and VRMD2-HP solves sub-problems to high-precision in the $\ell_\infty$ norm (in the sense of Definition~\ref{def:approx}). Thus, using our sample-reuse framework, Theorem~\ref{thm:AtoC} implies that we can \emph{reuse} the oblivious sample queries of Type I across all $\nOuter$ iterations of CPP. We obtain the following improved trade-off. 

\begin{theorem}[{\color{ForestGreen} $\ell_2$-$\ell_2$ MG trade-off improvement}]\label{thm:mg-trade-off-l2l2-improvement} For $\alpha > 0$, there is an algorithm that wp.\ $1 - \delta$ solves the MG problem in the $\ell_2$-$\ell_2$ setup using $\tilde{O}(\alpha/\epsilon)$ full batch queries 
and $\otilde(\norm{\bm{A}}_{F}^2 \alpha^{-2})$ sample queries of Type II.
\end{theorem}

\paragraph{Optimal Frobenius-norm-dependent query complexities for $\ell_2$-$\ell_2$ matrix games.}

In the case of $\ell_2$-$\ell_2$ games, note that the row/column oracle queries (sample queries of Type I) are \emph{strictly} less powerful than the batch queries (matrix-vector oracle queries) in the sense that one can always use a matrix-vector oracle to implement a row/column oracle. Consequently, setting $\alpha = \norm{\bm{A}}_F^{2/3} \epsilon^{1/3}$ in Theorem~\ref{thm:mg-trade-off-l2l2-improvement}, we obtain an overall matrix-vector oracle query complexity of $\tilde{O}(\norm{\bm{A}}_F^{-2/3} \alpha^{-2/3})$. 

Note that lower bounds of \cite{liu2023accelerated} indicate that $\tilde{\Omega}(\norm{\bm{A}}_F^{-2/3} \alpha^{-2/3})$-matrix-vector oracle queries is information-theoretically necessary. To our knowledge, Theorem~\ref{thm:mg-trade-off-l2l2-improvement} is the first to get this information-theoretically near-optimal rate for general $\ell_2$-$\ell_2$ matrix-games. 

\subsection{Applications of $\ell_2$-$\ell_1$ matrix games}\label{subsec:applications}
In this section, we discuss the implications of our results for two problems in computational geometry. \citet{allen2014optimization, carmon2019variance} showed how to reduce the minimum enclosing ball problem to $\ell_2$-$\ell_1$ matrix games. Much of the notation and presentation in this section is inspired by \cite{allen2014optimization} and \cite{carmon2019variance}. 

\paragraph{Maximum inscribed ball.}{In the maximum inscribed ball problem, we are given a polyhedron specified by a matrix $\bm{A} \in \R^{m \times n}$ and vector $\bm{b} \in \R^m$ so that $P = \{\bx \in \R^n : \bm{A} \bx + \bm{b} \geq \bm{0}\}$. We make the following assumptions, as in \cite{carmon2019variance, allen2014optimization}: 
\begin{itemize}[leftmargin=*]
    \item We assume that the polytope $P$ is bounded and hence $m \geq n$. 
    \item We assume that $\norm{\bm{a}_i}_2 = 1$ for all $i \in [m]$ so that $\norm{\bm{A}_{2 \rightarrow \infty}} = 1$. 
    \item The origin is inside the polytope. This is without loss of generality, as we may always shift the polytope to satisfy this requirement.
\end{itemize}

In the maximum inscribed ball problem, we must (approximately) compute $\bx^\star \in P$ such that 
\begin{align}\label{eq:maxIB}
    \bx^\star \in \argmax_{\bx \in P} \min_{i \in [n]} \frac{\langle \bm{a}_i, \bx \rangle + b_i}{\norm{\bm{a}_i}_2}. 
\end{align}
We define
\begin{align*}
    r^\star &\defeq \max_{\bx \in P} \min_{i \in [n]} \frac{\langle \bm{a}_i, \bx \rangle + b_i}{\norm{\bm{a}_i}_2}, \text{ and } \\
    R &\defeq \min\{r > 0: P \subset \{\bx \in \R^n : \norm{\bx}_2 \leq r\}\}. 
\end{align*}
We use $\rho \defeq R/r^\star$ to denote the aspect ratio of the problem. \citet{allen2014optimization} showed that solving \eqref{eq:maxIB} is equivalent to solving the following minimax problem: 
\begin{align*}
    \max_{\bx \in \R^n}\min_{\by \in \Delta^m} \by^\top \bm{A} \bx + \by^\top \bm{b}. 
\end{align*}
and formulated the maximum inscribed ball approximation problem as follows. 

\begin{definition}[Maximum inscribed ball (Max-IB)] In the \emph{maximum inscribed ball problem}, we must compute $\hat{\bx} \in \R^n$ such that wp.\ $1-\delta$, 
\begin{align*}
    \min_{\by \in \Delta^m} \by^\top \bm{A} \hat{\bx} + \by^\top \bm{b} \geq (1-\epsilon) \max_{\bx \in \R^n}\min_{\by \in \Delta^m} \by^\top \bm{A} \bx + \by^\top \bm{b}. 
\end{align*}
\end{definition}

\citet{carmon2019variance} further showed that the Max-IB problem can be solved using a two stage approach. In the first stage, we solve $\ell_2$-$\ell_1$ matrix games of the following form for a sequence of $\mu$'s: 
\begin{align}\label{eq:stage-1}
    \max_{\bx \in \mathbb{B}^n}\min_{\by \in \Delta^m} \by^\top \bm{A} \bx + \by^\top \bm{b} + \mu \sum_{i \in [m]} \by(i) \log(\by(i)) - \frac{\mu}{2} \norm{\bx}_2^2 
\end{align}
to obtain a constant-multiplicative approximation $\hat{r}$ to $r^\star$ (Lemma 10 of \cite{carmon2019variance}). In the second stage, we solve an $\ell_2$-$\ell_1$ matrix games of the following form to $O(\epsilon\hat{r})$-accuracy (Theorem 3 of \cite{carmon2019variance}): 
\begin{align}\label{eq:stage-2}
      \max_{\bx \in \mathbb{B}^n}\min_{\by \in \Delta^m} \by^\top \tilde{\bm{A}} \bx + \by^\top \bm{b}
\end{align}
where $\tilde{\bm{A}} = 2R \cdot \bm{A}$. As our MG Problem definition in the $\ell_2$-$\ell_1$ setup (Definition~\ref{def:mg-def}) captures both \eqref{eq:stage-1} and \eqref{eq:stage-2}, our methods can be used to obtain improved full batch versus sample query complexity trade-offs for solving the Max-IB problem. 
}

\paragraph{Minimum enclosing ball.} 
In the \emph{minimum enclosing ball problem}, we are given a data matrix $\bm{A} \in \R^{m \times n}$ such that $\bm{a}_1 = \bm{0}$, and $\max_{i \in [m]} \norm{\bm{a}_i} = 1$ so that $\norm{\bm{A}}_{2 \rightarrow \infty} = 1$. We must (approximately) find $R^\star$ such that there exists a point $\bx$ with $\norm{\bx - \bm{a}_i}_2 \leq R^\star$ for all $i \in [m]$. That is,
\begin{align}\label{eq:mineb}
    R^\star \defeq \min_{\bx \in \R^n} \max_{\by \in \Delta^m} \by^\top \bm{A} \bx + \by^\top \bm{b} + \frac{1}{2} \norm{\bx}_2^2. 
\end{align}

\begin{definition}[Minimum enclosing ball (Min-EB)] In the minimum enclosing ball problem, we must solve the minimax problem (Definition~\ref{def:minimax-def}) for $f(\bx, \by) = \by^\top \bm{A} \bx + \by^\top \bm{b} + \frac{1}{2} \norm{\bx}_2^2.$
\end{definition}

\citet{carmon2019variance} showed that the Min-EB problem can be solved to accuracy $\epsilon/8$ by solving the following $\ell_2$-$\ell_1$ matrix game to accuracy $\epsilon/16$ (Lemma 11 of \cite{carmon2019variance}): 
\begin{align}\label{eq:mineb-reduction}
    \min_{\bx \in \mathbb{B}^n} \max_{\by \in \Delta^m} \by^\top \bm{A} \bx + \by^\top \bm{b} - \frac{\epsilon}{32 \log(m)} \sum_{i \in [m]} \by(i) \log(\by(i))+ \frac{1}{2} \norm{\bx}_2^2. 
\end{align}
As our MG Problem definition in the $\ell_2$-$\ell_1$ setup (Definition~\ref{def:mg-def}) captures both \eqref{eq:mineb-reduction} our methods can be used to obtain improved full batch versus sample query complexity trade-offs for solving the Min-EB problem. 

%% file: nouniform_smoothness.tex
\section{Application: Finite-sum minimization with non-uniform smoothness}\label{apx:nonuniform_smoothnes}

In this section, we discuss applications of pseudo-independence for improved full batch versus sample query trade-offs for finite sum minimization (FSM) where the component functions are of non-uniform smoothness. This is a generalization of Section~\ref{sec:finite-sum-minimization}. 

\begin{definition}[{\color{ForestGreen}{Generalized FSM (GFSM) problem}}]\label{def:finite-sum-minimization-nonuniform} In the GFSM problem, we are given $c > 1$ and  $\bm{x}_0 \in \R^d$ and must output $\hat{\bm{x}} \in \R^d$ such that $F(\hat{\bm{x}}) - \min_{\bm{x}} F(\bm{x}) \leq 1/c \cdot (F(\bm{\bm{x}}_0) - \min_{\bm{z}} F(\bm{x}))$ where $F:\R^d \rightarrow \R$ with $F:\R^d \rightarrow \R$ witj $F(x) \defeq \frac{1}{n}\sum_{i\in[n]} f_i(x)$, $F$ is $\mu$ strongly-convex, and each $f_i: \R^d \to \R$ is of known smoothnes $L_i$. 
\end{definition}

State-of-the-art query complexities for non-uniform smoothness finite-sum minimization can be achieved by using a primal-dual extra-gradient method \citep{jin2022sharper}. By using this primal-dual extra-gradient method as our sub-problem solver for solving regularized problems and using APP as the outer-solver \citep{frostig2015regularizing} as in the uniform-smoothness case (Section~\ref{sec:finite-sum-minimization}), we can obtain trade-offs between batch and sample queries that depend on the distribution of the smoothness parameters $L_i$ (rather than bounds that depend only on the worst-case smoothness $L$ as in Section~\ref{sec:finite-sum-minimization}). 

The gradient (batch) oracle and component (sample) oracle are the same as for FSM (Definitions~\ref{def:gradient-oracle} and ~\ref{def:component-oracle}. The FSM sub-problems and $\fsub$ are defined exactly as in the uniform smoothness case (Definition~\ref{def:fsm-subproblem}.) The only change relative to Section~\ref{sec:finite-sum-minimization} is the choice of sub-solver. For the GFSM problem, we use the primal-dual finite-sum minimization algorithm of \cite{jin2022sharper}. 

\begin{theorem}[PDFSM, Theorem 2 of \cite{jin2022sharper}, restated - {\color{ForestGreen} GFSM sub-problem solver}]\label{thm:primal-dual} Let $F$ be as in Definition~\ref{def:finite-sum-minimization-nonuniform}. There is a randomized algorithm $\cA^{\mathrm{PDFSM} | \lambda, c, \delta}_{\xi, \chi}(\bm{x})$ such that the following holds. 
\begin{itemize}[leftmargin=*]
    \item The algorithm takes in the random seeds $\xi, \chi$ distributed as follows. The first random seed $\xi \sim \{i_1, ..., i_T\}$ where each $\mathbb{P}[i_t = j] = \sqrt{L_j}/(\sum_{k \in [n]} \sqrt{L_k})$ for some 
    \begin{align*}
        T = \tilde{O}\paren{\sum_{i \in [n]} \sqrt{\frac{L_i}{n \mu}}}. 
    \end{align*}
    The second random seed $\chi \sim \UniformDist{[0]}$ is a 0-bit random seed. 
    \item If $\bm{x}' = \cA^{\mathrm{PDFSM} | \lambda, c, \delta}_{\xi, \chi}({\ba})$, then wp.\ $1-\delta$ over the draw of the random seeds $\xi$ and $\chi$, ${\bx}'$ is a solution to the $({\ba}, \lambda, c)$-sub-problem.
    \item $\cA^{\mathrm{PDFSM}| \lambda, c, \delta}_{\xi, \chi}$ makes only $\tilde{O}(1)$ batch queries and makes sample queries \emph{only} on the indices contained in $xi$.
\end{itemize}
\end{theorem}

As in the nonuniform case, by combining Theorem~\ref{thm:primal-dual} with Theorem~\ref{thm:approx-prox}, we obtain a standard trade-off of $\tilde{O}(\sqrt{\lambda/\mu})$ batch (gradient oracle) queries and $\tilde{O}\paren{\sqrt{\lambda/\mu} \cdot \sum_{i \in [n]} \sqrt{\frac{L_i}{n \mu}}}$ sample (component oracle) queries for any $\lambda \geq \mu$.

However, as in Section~\ref{sec:finite-sum-minimization}, using our sample-reuse framework, we observe that we can \emph{reuse} randomness across all $\nOuter = \sqrt{\lambda/\mu}$ sub-problem solves. More concretely, using identical convexity argument as in the proof of Lemma~\ref{lemma:svrg-convex-hp}, we obtain the following analog of Theorem~\ref{thm:primal-dual}. 

\begin{lemma}[{\color{ForestGreen} GFSM sub-problem solver high-precision}]\label{lemma:primal-dual-hp} Let $F$ be as in Definition~\ref{def:finite-sum-minimization-nonuniform}. There is a randomized algorithm $\cA^{\mathrm{PDFSM-HP} | \lambda, c, \delta}_{\xi, \chi}(\bm{x})$ such that the following holds. 
\begin{itemize}[leftmargin=*]
    \item The algorithm takes in the random seeds $\xi, \chi$ distributed as follows. The first random seed $\xi \sim \{i_1, ..., i_T\}$ where each $\mathbb{P}[i_t = j] = \sqrt{L_j}/(\sum_{k \in [n]} \sqrt{L_k})$ for some 
    \begin{align*}
        T = \tilde{O}\paren{\sum_{i \in [n]} \sqrt{\frac{L_i}{n \mu}}}. 
    \end{align*}
    The second random seed $\chi \sim \UniformDist{[0]}$ is a 0-bit random seed. 
    \item If $\bm{x}' = \cA^{\mathrm{SVRG-HP} | \lambda, c, \delta}_{\xi, \chi}({\ba})$, then wp.\ $1-\delta$ over the draw of the random seeds $\xi$ and $\chi$,
    \begin{align*}
    \norm{{\bx}' - \fsubbar_{\lambda'}(\bm{y}_{\ba})}_\infty^2 \leq \frac{1}{c} \cdot \paren{F({\by}_{{\ba}}) - \min_{\tilde{{\bx}} \in \R^d}  F(\tilde{{\bx}}) + \frac{\lambda}{2} \norm{\tilde{{\bx}} - {\by}_{{\ba}}}_2^2}. 
    \end{align*}
    \item $\cA^{\mathrm{PDFSM-HP} | \lambda, c, \delta}_{\xi, \chi}$ makes only $\tilde{O}(1)$ batch queries and makes sample queries \emph{only} on the indices contained in $\xi$.
\end{itemize}
\end{lemma}

By combining Lemma~\ref{lemma:primal-dual-hp} with Lemma~\ref{lemma:fsm-robust} and applying Theorem~\ref{thm:AtoC}, we obtain a trade-off of $\tilde{O}(\sqrt{\lambda/\mu})$ batch (gradient oracle) queries and $\tilde{O}\paren{\sqrt{\lambda/\mu} \cdot \sum_{i \in [n]} \sqrt{\frac{L_i}{n \mu}}}$ sample (component oracle) queries for any $\lambda \geq \mu$. The pseudo-code is analogous to Algorithm~\ref{alg:fsm-reuse} except that the inner-loop sub-problem solver is replaced with PDFSM-HP in place of SVRG-HP. 

\begin{theorem}[{\color{ForestGreen} Non-uniform smoothness FSM trade-off improvement}]\label{thm:fsmmain-nonuniform} For $\lambda \geq \mu$, there is an algorithm that solves FSM with non-uniform smoothness (Definition~\ref{def:finite-sum-minimization-nonuniform}) using $\tilde{O}(\sqrt{\lambda/\mu})$ full batch queries and only $\tilde{O}(\sum_{i \in [n]} \sqrt{L_i/(n \lambda)})$ sample queries. 
\end{theorem}

%% file: topev.tex
\section{Application: Top eigenvector computation}\label{sec:topEV}

Here, we discuss the setting of top eigenvector (TopEV) computation. This is an interesting specialized setting in which our improvement for finite-sum optimization (Definition~\ref{def:finite-sum-minimization}) can be applied even if the components $f_i$ of the finite sum might not be convex. 

Throughout this section, we use $\bm{A} \in \R^{n \times d}$ to denote a matrix, and we use $\bm{a}_{1}, ..., \bm{a}_{n} \in \R^d$ to denote its rows. We use $\norm{\bm{A}}_F \defeq \sqrt{\sum_{i,j} \bm{A}_{i,j}^2}$ and $\norm{\bm{A}}_2$ to denote the spectral norm of $\bm{A}.$
We use $\sr(\bm{A}) \defeq \sum_{i} \frac{\lambda_i}{\lambda_1} = \norm{\bm{A}}_F^2/\norm{\bm{A}}_2^2$, where $\lambda_1 \geq \lambda_2 \geq \cdots \lambda_d$ are the eigenvalues of $\bm{\Sigma}$. Note that we always have $\sr(\bm{A}) \leq \text{rank}(\bm{A}).$ We define the relative eigen-gap  $\gap \defeq \frac{\lambda_1-\lambda_2}{\lambda_1}.$

\begin{definition}[{\color{ForestGreen} TopEV problem}]\label{def:topev-problem} In the TopEV problem, we are given a matrix $\bm{A} \in \R^{n \times d}$ and $\epsilon> 0$ and must compute a unit vector $\bm{x} \in \R^d$ such that $\bm{x}^\top\bm{\Sigma} \bm{x} \geq (1-\epsilon) \lambda_1$ where $\bm{\Sigma} \defeq \bm{A}^\top \bm{A} \in \R^{d \times d}$ and $\lambda_1$ is the largest eigenvalue of $\bm{\Sigma}$.
\end{definition}

In typical settings \citep{garber2016faster}, the goal is to solve this problem with probability inverse polynomial in $d$ (i.e., with high probability in $d$.) We define the batch and sample queries for solving the problem as follows. 

\begin{definition}[Matrix-vector oracle - {\color{ForestGreen} TopEV batch oracle}] When queried with $\bm{x} \in \R^d$, the \emph{matrix-vector} oracle returns $\bm{A} \bm{x}$.
\end{definition}

\begin{definition}[Row oracle - {\color{ForestGreen} TopEV sample oracle}] When queried, with $i \in [n]$, a \emph{row oracle} for $\bm{A}$ returns $\bm{a}_i \in \R^d$.
\end{definition}

\paragraph{Reducing top eigenvector computation to solving linear systems in $\lambda' \bm{I} - \bm{A}^\top \bm{A}$} \citet{garber2016faster} showed how to reduce top eigenvector computation to performing an approximate version of the classical inverse power method in the shifted matrix $\lambda \bm{I} - \bm{A}^\top \bm{A}$, where $\lambda$ is an appropriately chosen parameter. This is called the \emph{approximate shift-and-invert} power method and can be implemented given access to an oracle that approximately solves linear systems in $\lambda \bm{I} - \bm{A}^\top \bm{A}$. 

Furthermore, \citet{garber2016faster} also showed that selecting $(1+\gap/150)\lambda_1 \leq \lambda \leq (1+\gap/100)\lambda_1$ is sufficient to ensure that $\tilde{O}(1)$ iterations of approximate shift-and-invert power method is sufficient to solve the top eigenvector problem. Section 6 of \cite{garber2016faster} shows that given an algorithm for approximately solving linear systems in $\lambda' \bm{I} - \bm{A}^\top \bm{A}$ for $\lambda' > \lambda_1 + \gap/120$, there is a method for computing a valid shift parameter $\lambda$ (such that $(1+\gap/150)\lambda_1 \leq \lambda \leq (1+\gap/100)\lambda_1$) with runtime and query complexity overhead that is only lower order relative to the approximate shift-and-invert power method steps. 

Consequently, in the remainder of this section, we simply focus on the problem of solving linear systems of the form 
\begin{align*}
    (\lambda' \bm{I} - \bm{A}^\top \bm{A}) \bm{x} = \bm{b},  
\end{align*}
where $\lambda' > \lambda_1 + \Theta(\gap)$ and $\bm{b}$ is known, since this is the fundamental subroutine used in the algorithms of \cite{garber2016faster}. Concretely, we summarize their reductions as follows. 

\begin{theorem}[Theorems 5, 8, 15, and 30 of \cite{garber2016faster}]\label{thm:reduction}  Given an algorithm $\code{Solve}(\bm{x}_0 \bm{A}, \lambda', \bm{b})$ that computes $\bm{x}$ such that wp.\ $1-\poly(1/d, \gap)$, 
\begin{align*}
    \norm{\bm{x} - (\lambda' \bm{I}- \bm{A}^\top \bm{A})^{-1}\bm{b}}_2^2 \leq \poly(1/d, \gap) \norm{\bm{x}_{0} - (\lambda' \bm{I} - \bm{A}^\top \bm{A})^{-1}\bm{b}}_2^2, 
\end{align*}
for any $\lambda > \lambda_1 + \gap/120$,
there is an algorithm that solves the TopEV problem with only $\tilde{O}(1)$ calls to $\Solve$. 
\end{theorem}

Thus, in the remainder of this section, we consider the problem of solving linear systems of the form 
\begin{align}\label{eq:linear-system}
    (\lambda' \bm{I} - \bm{A}^\top \bm{A}) \bm{x} = \bm{b},  
\end{align}
for some $\bm{b}\in \R^n$ and $\lambda' > \lambda_1 + \gap/120$. This can equivalently be viewed as the following minimization problem: 
\begin{align}\label{eq:topevsum}
    \min_{\bm{x} \in \R^d} \frac{1}{2} \bm{x}^\top (\lambda' \bm{I} - \bm{A}^\top \bm{A}) \bm{x} - \bm{b}^\top \bm{x} = \min_{\bm{x}} \frac{1}{n} \sum_{i \in [n]} \frac{1}{2}\cdot \bm{x}^\top (w_i \bm{I} - n {\bm{a}_{i}}{\bm{a}_{i}}^\top) \bm{x} -  \bm{b}^\top \bm{x} ,
\end{align}
whenever $\sum_{i} w_i/n = \lambda'.$ This is reminiscent of the FSM problem (Definition~\ref{def:finite-sum-minimization}) with $f_i(\bm{x}) = \frac{1}{2} \cdot \bm{x}^\top (w_i \bm{I} - n {\bm{a}_{i}}{\bm{a}_{i}}^\top) \bm{x} -\bm{b}^\top \bm{x}$. However, we have the caveat that the matrices $(w_i\bm{I} - {\bm{a}_{i}}{\bm{a}_{i}}^\top)$ need not be positive semi-definite; and consequently, the $f_i$'s are not necessarily convex \emph{even though} $F$ \emph{is} $\lambda' - \lambda_1 = \Theta(\gap)$-\emph{strongly convex} (and $\norm{A}_F^2$-smooth). 

Interestingly, the outer-solver procedure APP~(Theorem~\ref{thm:approx-prox}) still applies in this setting; however, the guarantees of SVRG (Theorem~\ref{thm:svrg-convex}) unfortunately do not apply \emph{directly} when the $f_i$ are potentially non-convex. Nonetheless, \citet{garber2016faster} leveraged problem structure to show that SVRG can be applied for \eqref{eq:topevsum}. Consequently, we define the {\color{ForestGreen} TopEV sub-problem} similarly as in Definition~\ref{def:sub-problem} for $F(\bx) \defeq f_i(\bx)$, where 
\begin{align}\label{eq:F-def}
    F(\bx) \defeq f_i(\bx) \text{ where } f_i(x) \defeq  \frac{1}{2} \cdot \bm{x}^\top \paren{w_i\bm{I} - n {\bm{a}^{(i)}} {\bm{a}^{(i)}}^\top} \bm{x} - \bm{b}^\top \bm{x};  
    w_i \defeq n \cdot \frac{\lambda' \norm{\bm{a}_i}_2^2}{\norm{\bm{A}}_F^2}.
\end{align}
Note that from this perspective, the batch oracle call for TopEV (Definition~\ref{def:gradient-oracle}) exactly corresponds to a gradient oracle call for $F$, and a sample oracle for call for $F$ exactly corresponds to a component oracle call for $F$ (Definition~\ref{def:component-oracle}.)

\begin{definition}[{\color{ForestGreen} TopEV sub-problem}]\label{def:topev-subproblem} Let $F$ be as in \eqref{eq:F-def} and $\rho \geq \lambda' - \lambda_1$. In the $(\bm{u}, \rho, c)$-sub-problem for TopEV, we must solve the $(\bm{u}, \rho, c)$-FSM-sub-problem (Definition~\ref{def:fsm-subproblem}.)
\end{definition}

Note the similarity between the TopEV sub-problem and the FSM sub-problem from Definition~\ref{def:fsm-subproblem}. The only difference is that the $f_i$ need not be \emph{convex}. Nonetheless, \citet{garber2016faster} provide the following guarantee, which is analogous to Theorem~\ref{thm:svrg-convex} from the FSM setting in Section~\ref{sec:finite-sum-minimization}. 

\begin{theorem}[SVRG, Theorem 2.2 of \cite{frostig2015regularizing}, restated - {\color{ForestGreen} TopEV sub-problem solver}]\label{thm:svrg-nonconvex} Let $F$ be as defined in \eqref{eq:F-def} and $\rho \geq (\lambda' - \lambda_1)$. There is a randomized algorithm $\cA^{\mathrm{SVRG-TopEV} | \rho, c, \delta}_{\xi, \chi}$ such that the following holds. 
\begin{itemize}[leftmargin=*]
    \item The algorithm takes in the random seeds $\xi, \chi$ distributed as follows. The first random seed $\xi \sim \{i_1, ..., i_T\}$ where each $\mathbb{P}(i_j = k) =\norm{\bm{a}_k}_2^2/\norm{\bm{A}}_F^2$ for some $T = \tilde{O}(L/\lambda)$. The second random seed $\chi \sim \UniformDist{[0]}$ is a 0-bit random seed. 
    \item If $\bm{x}' = \cA^{\mathrm{SVRG-TopEV} | \rho, c, \delta}_{\xi, \chi}({\ba})$, then wp.\ $1-\delta$ over the draw of the random seeds $\xi$ and $\chi$, ${\bx}'$ is a solution to the $({\ba}, \rho, c)$-sub-problem.
    \item $\cA^{\mathrm{SVRG-TopEV}| \rho, c, \delta}_{\xi, \chi}$ makes only $\tilde{O}(1)$ batch queries and makes sample queries \emph{only} on the indices encoded in $\xi$.
\end{itemize}
\end{theorem}

By Theorem~\ref{thm:approx-prox}, we can instantiate APP using SVRG (Theorem~\ref{thm:svrg-nonconvex}) with the to solve problems of the form \eqref{eq:linear-system} as required by Theorem~\ref{thm:reduction}. The pseudocode is the same as Algorithm~\ref{alg:fsm}, replacing SVRG-HP with SVRG-TopEV. Setting $\delta$ to be inversely polynomial in $d$, this solves the TopEV problem with high probability in $d$ and obtains a trade-off of $\tilde{O}(\sqrt{\rho/(\lambda' - \lambda_1)})$ full batch (matrix-vector oracle) queries and 
\begin{align*}
    \tilde{O}\paren{\frac{\rho^2 + 12\lambda_1 \norm{\bm{A}}_F^2}{(\rho - \lambda_1 + \lambda)^2}} \cdot \sqrt{\frac{\rho}{\lambda'-\lambda_1}}
\end{align*}
sample (row oracle) queries for any $\rho \geq \lambda'-\lambda_1$. 

Now, since $F$ in this case remains smooth and strongly convex (even though the $f_i$'s may not be convex), using the \emph{exact} same arguments as in Section~\ref{sec:finite-sum-minimization} (Lemmas~\ref{lemma:fsm-robust} and Lemma~\ref{lemma:svrg-convex-hp}) we can use our sample reuse framework from Section~\ref{sec:pseudoindependence} (Theorem~\ref{thm:AtoC}) to obtain the following improved trade-off. (Again, th pseudocode is the same as Algorithm~\ref{alg:fsm-reuse}, replacing SVRG-HP with SVRG-TopEV.)

\begin{theorem}[{\color{ForestGreen} TopEV trade-off improvement}] For $\lambda \geq \mu$, there is an algorithm that solves FSM with high probability in $d$ using only $\tilde{O}(\sqrt{\rho/(\lambda' - \lambda_1)})$ full batch queries and only $\otilde\paren{\frac{\rho^2 + 12\lambda_1 \norm{\bm{A}}_F^2}{(\rho - \lambda_1 + \lambda)^2}}$ sample queries.
\end{theorem}

Finally, we remark that to obtain the trade-offs reported in Table~\ref{table:main}, we simply make the change of variables $\rho = \alpha \lambda_1$. With this change of variables, 
\begin{align*}
    \sqrt{\frac{\rho}{\lambda' - \lambda_1}} = \sqrt{\frac{\alpha}{\gap}}, \quad \text{ and } \quad
    \frac{\rho^2 + 12\lambda_1 \norm{\bm{A}}_F^2}{(\rho - \lambda_1+\lambda)^2} = \tilde{O}\paren{\frac{\alpha^2\lambda_1^2 + 12 \lambda_1 \norm{\bm{A}}_F^2}{\alpha^2\lambda_1^2}} = \tilde{O}\paren{\frac{\sr(\bm{A})}{\alpha^2}}.
\end{align*}

%% file: conclusion.tex
\section{Conclusion}\label{sec:conclusion} We introduced a sample reuse framework to reuse randomness across multiple subproblem solutions in variance-reduced optimization methods without sacrificing theoretical correctness guarantees. Our results enabled improved query complexity trade-offs for a broad range of optimization problems. One limitation of our current sample reuse framework and analysis of pseudo-independence is that it only allows us to reuse samples across \emph{high-accuracy} subroutines. Although this already enables improvements for several problems, as seen in Table~\ref{table:main}, in future work it may be interesting to study whether tighter analysis allows sample reuse even for sub-routines which do not solve to high-accuracy. This could have applications, for instance, to non-convex optimization problems. As mentioned in Section~\ref{sec:intro:finite-motivate}, another limitation is that our results don't directly yield a worst-case asymptotic-runtime improvement that we are aware of; however it sheds new light on the information needed to solve FSM and could yield faster algorithms depending on caching and memory layout.

%% file: stable_proof.tex
\section{Inducing pseudoindependence numerically stably}\label{apx:stable_proof}

The pseudo-independent algorithm constructed in the proof of Theorem~\ref{thm:conversion-thm} is simple; however, it may not be directly implementable in finite precision, as it requires infinite precision to directly implement the addition of uniform random noise. The proof of Theorem~\ref{thm:stable_proof} below provides an alternative construction that can be implemented in finite precision.

\begin{theorem}[Finite precision analog of Theorem~\ref{thm:conversion-thm}]\label{thm:stable_proof}
    Let $\epsilon, \delta \in (0, 1)$, $\eta > 0$, $\eta' \defeq \min(\eta/4, \eta\epsilon/4)$, and let $\cA_{\xi, \chi}$ be a randomized algorithm that is an $(\eta', \delta)$-approximation of a function $f: {\R^d} \to {\R^p}$. Then, there is a numerically stable algorithm $\cA'_{\xi, \chi'}$ such that ${\cA'_{\xi, \chi'}}$ is an $(\epsilon, \delta)$-pseudo-independent of $\xi$ and an $(\eta, \delta)$-approximation of $f$ with the \emph{same} runtime and query complexities as $\cA_{\xi, \chi}$ up to an additive $O(p)$ in runtime. 
\end{theorem}

\begin{proof} Consider $\cS \subset \R^p$ to be a $\beta$-covering of $\R^p$ in the $\ell_\infty$ norm. Define ${\mathsf{round}}: \R^p \to \cS$ to be the operator that maps $\bx \in \R^p$ to some $\bx' \in \R^p$ such that $0 \leq \bx(i) - \bx'(i) \leq \beta$ for all $i \in [d]$, Define $\mathsf{smooth}_{\bm{\nu}, t}$ to be the operator which uses a random seed $\bm{\nu} \sim \cD_{\bm{\nu}}$ to map $\bx\in \cS$ to a uniformly random $\bx' \in \{y \in \cS: -t \leq (\bx(i)-\bx'(i)) \leq t\}$. Here, $t$ is a parameter that will be specified later in the proof. 

Let $\cA_{\xi, \chi}$ be the randomized algorithm which takes input $\bx \in \R^d$, random seed $\xi \sim \cD_\xi$, and $\chi$ where $\chi = (\chi', \bm{\nu}) \sim \cD_{\chi}^{\bx}$ is the concatenation of an independently drawn seed $\chi' \sim \cD_{\chi'}^{\bx}$ and seed $\bm{\nu} \sim \cD_{\bm{\nu}}$. For any realization $s, c, n$ of $\xi, \chi', \bm{\nu}$, let $\cA_{\xi=s, \chi=(c, \bm{e})}(\bx) = \mathsf{smooth}_{\bm{\nu}=\bm{e}, t}({\mathsf{round}}(\cA_{\xi=s, \chi'=c}(\bx)))$. 

First,  we'll construct a smoothing for $\cA_{\xi, \chi}$. Let $\bar{\cA}_{\chi}$ be the randomized algorithm which takes input $\bx\in \R^d$ and a random seed $\chi \sim \cD_{\chi}^{\bx}$. For any realization $c, n$ of $\chi', \bm{\nu}$, let $\bar{\cA}_{\chi = (c, \bm{e})}(\bx) = \mathsf{smooth}_{\bm{\nu}=\bm{e}, t}({\mathsf{round}}(f(\bx)))$.
Now, by \eqref{eq:target-closeness} and the fact that we have a $\beta$-covering, 
\begin{align*}
    \probSub{\tv{p_{A_{\xi = s, \chi}(\bx)}}{p_{\bar{A}_{\chi}(\bx)}} \leq \paren{\frac{ \ceil{\eta/\beta} }{\floor{2t/\beta}}}^p}{s \sim \cD_{\chi}} \geq 1-\delta.
\end{align*}
Therefore, for $p \geq 1$ and $\eta < 2t$, we have 
\begin{align*}
    \probSub{\tv{p_{A_{\xi = s, \chi}(x)}}{p_{\bar{A}_{\chi}(x)}} \leq {\frac{ \eta }{t}}}{s \sim \cD_{\chi}} &= \probSub{\tv{p_{A_{\xi = s, \chi}(x)}}{p_{\bar{A}_{\chi}(x)}} \leq \paren{\frac{\eta }{t}}^p}{s \sim \cD_{\chi}} \\
    &\geq 1-\delta.
\end{align*}

Next, we need to show that $\Prob_{\xi \sim \cD_{\xi}}(\Vert \cA_{\xi, \chi}(x) - f(x) \Vert_\infty \geq \eta) \leq \delta.$ Note that by  \eqref{eq:target-closeness}, with probability $1-\delta$ over the draw of $\xi$, $\Vert \cA_{\xi, \chi}(x) - f(x) \Vert_\infty \leq \eta' + \beta + 2t \leq \epsilon$ whenever $\beta, t, \eta' \leq \eta/4$.

Consequently, when $\tau = \eta/4$, $t = \frac{\eta}{\epsilon}$, $\bar{\cA}_{\chi}$ is an $(\epsilon, \delta)$-smoothing for $\cA_{\xi, \chi}$. And when $\eta' \leq \eta\epsilon/4$, $\Prob_{\xi \sim \cD_{\xi}}(\Vert \cA_{\xi, \chi}(x) - f(x) \Vert_\infty \geq \eta) \leq \delta$ as well. This completes the proof of the first guarantee of $\cA_{\xi, \chi}$. 

For the second guarantee, note that $\cA_{\xi, \chi}$ has the same runtime and query complexities up to an additive $O(p)$ increase in the runtime due to the cost of performing the $p$-dimensional random perturbation induced by $\bm{\nu}$. 
\end{proof}

%% file: appendix_helpers.tex
\section{Pseudoindependence and repeated compositions}\label{sec:bound-difference}

In this section, our goal is to prove the following theorem  (see Definition~\ref{def:composition} for related notation.)

\pseudoindependencemain*

Before proving \Cref{lem:pseudoindependencemain}, we state the following fact about transformations of random variables. 

\begin{restatable}{fact}{transformation}\label{fact:transformation}
    Let $A$ and $B$ be random variables in $\cQ$ and let $\zeta$ be a deterministic function $\zeta:\cQ \to \cQ$. Then, $\smash{\tv{p_{\zeta(A)}}{p_{\zeta(B)}} \leq \tv{p_A}{p_B}}$. 
\end{restatable} 
\begin{proof} For any deterministic function $f$, let $f^{-1}(\cdot)$ denote the preimage of $f$. That is, for any $T \subset \cQ$, let $f^{-1}(T) \defeq \{\omega' \in \Omega : f(\omega') \in T\}$. Then, by the definition of the total variation distance, we have 
\begin{align*}
    \tv{p_{\zeta(A)}}{p_{\zeta(B)}} &= \sup_{S \subset \cQ} |\prob{\zeta(A) \in S} - \prob{\zeta(B) \in S}| \\
    &= \sup_{S \subset \cQ} |\prob{A \in \zeta^{-1}(S)} - \prob{B\in \zeta^{-1}(S)}| \\
    &\leq \sup_{T \subset \cQ} |\prob{A \in T} - \prob{B\in T}| = \tv{p_{A}}{p_B}. 
\end{align*}
\end{proof}

Now, to prove Theorem~\ref{lem:pseudoindependencemain}, we first bound the TV distance between $p_{\Phi^T_{\cA'}(\ba; \cD_{\xi}, \mathsf{D}_{\chi'})}$ and $p_{H^T_{\bar{A}}(\ba; \mathsf{D}_{\chi'})}$ (recall Definition~\ref{def:composition}.)

\begin{restatable}{lemma}{reuserandomnessone}\label{lemma:fresh-randomness}  Let $\cA'_{\xi, \chi'}$ be a randomized algorithm which takes an input $\ba \in \R^d$ and two random seeds $\xi \sim \cD_{\xi}, \chi' \sim \cD_{\chi'}^{\ba}$. Suppose $\cA'_{\xi, \chi'}$ is $(\epsilon, \delta)$-pseudo-independent and that $\bar{\cA}_{\chi'}$ is an $(\epsilon
, \delta)$-smoothing of $\cA'$ with respect to $\xi$. Let $s \sim \cD_{\xi}$. Then,
\begin{align*}
    \tv{p_{\Phi_{\cA'}^T(\ba; \cD_{\xi}, \mathsf{D}_{\chi'})}}{p_{H_{\bar{\cA}}^T(\ba; \mathsf{D}_{\chi'})}} \leq T(\delta + \epsilon). 
\end{align*}
\end{restatable}
\begin{proof} Induct on $T$. If $T = 1$, the statement essentially follows from the definition of pseudo-independence, Fact~\ref{lemma:tv_decomposition}, and a union bound, as we elaborate in the next few sentences. Concretely, we have that with probability $1-\delta$ over the draw of $s \sim \cD_{\xi}$, 
\begin{align*}
    \tv{p_{\cA'_{\xi=s, \chi'}(\ba)}}{p_{\bar{\cA}_{\chi'}(\ba)}} \leq \epsilon. 
\end{align*}
Since $\zeta$ is a deterministic function, using Fact~\ref{fact:transformation}, we have
\begin{align*}
    \tv{p_{\zeta\paren{\ba, \cA_{\xi=s, \chi}(\ba)}}}{p_{\zeta\paren{\ba, \cA_{\xi=s, \chi}(\ba)}}} \leq  \tv{p_{\cA_{\xi=s, \chi}(\ba)}}{p_{\cA_{\xi=s, \chi}(\ba)}} \leq \epsilon. 
\end{align*}
Consequently, by Fact~\ref{lemma:tv_decomposition} and union bound,
\begin{align*}
    \tv{p_{\Phi_{\cA'}^1(\ba; \cD_{\xi}, \mathsf{D}_{\chi'})}}{p_{H_{\bar{\cA}}^1(\ba; \mathsf{D}_{\chi'})}} \leq (\delta + \epsilon). 
\end{align*}
This completes the base case. 

For the inductive step, suppose the theorem holds up to $T-1$. That is, suppose that 
\begin{align*}
    \tv{p_{\Phi_{\cA'}^T(\ba; \cD_{\xi}, \mathsf{D}_{\chi'})}}{p_{H_{\bar{\cA}}^T(\ba; \mathsf{D}_{\chi'})}} \leq (T-1)(\delta + \epsilon). 
\end{align*}
Then, by Fact~\ref{lemma:tv_decomposition}, there exists an event $E_1$ and random variables $C, D, F$ such that conditioned on $E_1$, $\Phi_{\cA'}^{T-1}(\ba; \cD_{\xi}, \mathsf{D}_{\chi'}) \overset{\cD}{=} C \overset{\cD}{=} {H_{\bar{\cA}}^{T-1}(\ba; \mathsf{D}_{\chi'})}$ and $\prob{E_1} \geq 1-(T-1)(\delta + \epsilon).$ Consequently, for any event $E$, we have
\begin{align*}
    \abs{p_{\Phi^{T}_{\cA'}(\ba; \cD_{\xi}, \mathsf{D}_{\chi'})}(E) - p_{H^{T}_{\bar{\cA}}(\ba; \mathsf{D}_{\chi'})}(E)} &\leq \abs{p_{\Phi^{T}_{\cA'}(\ba; \cD_{\xi}, \mathsf{D}_{\chi'}) | E_1}(E) - p_{H^{T}_{\bar{\cA}}(\ba; \mathsf{D}_{\chi'}) | E_1}(E)} + \prob{\neg E_1} \\
    &= \abs{p_{\Phi^{T}_{\cA'}(\ba; \cD_{\xi}, \mathsf{D}_{\chi'}) | E_1}(E) - p_{H^{T}_{\bar{\cA}}(\ba; \mathsf{D}_{\chi'}) | E_1}(E)} + (T-1)(\delta + \epsilon) \\
\end{align*}
To bound the first term, we observe that by the definition of $\Phi$ and $H$, we have
\begin{align*}
&\abs{p_{\Phi^{T}_{\cA'}(\ba; \cD_{\xi}, \mathsf{D}_{\chi}) | E_1}(E) - p_{H^{T}_{\bar{\cA}}(\ba; \mathsf{D}_{\chi}) | E_1}(E)}  \\
&= \abs{p_{\zeta\paren{\Phi^{T-1}_{\cA'}(\ba; \cD_\xi, \mathsf{D}_{\chi'}), \cA'_{\xi, \chi'}\paren{\Phi^{T-1}_{\cA'}(\ba; \cD_{\xi}, \mathsf{D}_{\chi'})}} | E_1}(E) - p_{\zeta\paren{H^{T-1}_{\cA'}(\ba; \mathsf{D}_{\chi'}), \bar{\cA}_{\xi, \chi'}\paren{H^{T-1}_{\bar{\cA}}(\ba; \mathsf{D}_{\chi'})}} | E_1}(E)} \\
&\leq \sup_{\bm{c} \in \R^d} \abs{p_{\zeta\paren{\bm{c}, \cA'_{\xi, \chi'}(\bm{c})}} - p_{\zeta\paren{\bm{c}, \bar{\cA}_{\chi'}(\bm{c})}}},
\end{align*}
where in the last line, we used the fact that conditional on $E_1$, ${\Phi^{T-1}(x; \cD_{\xi}, \mathsf{D}_{\chi'})} \overset{\cD}{=} C \overset{\cD}{=} H^{T-1}(x; \mathsf{D}_{\chi'})$. By Fact~\ref{fact:transformation} and the definition of pseudo-independence (using an identical argument to the base case) we have
\begin{align*}
\sup_{\bm{c} \in \R^d} \abs{p_{\zeta\paren{\bm{c}, \cA'_{\xi, \chi'}(\bm{c})}} - p_{\zeta\paren{\bm{c}, \bar{\cA}_{\chi'}(\bm{c})}}} \leq (\delta + \epsilon).  
\end{align*}
Thus, by substitution, we conclude that for any event $E$, 
\begin{align*}
    \abs{p_{\Phi^{T}_{\cA'}(\ba; \cD_{\xi}, \mathsf{D}_{\chi}) | E_1}(E) - p_{H^{T}_{\bar{\cA}}(\ba; \mathsf{D}_{\chi}) | E_1}(E)} \leq (\delta + \epsilon), 
\end{align*}
and consequently, 
\begin{align*}
    \abs{p_{\Phi^{T}_{\cA'}(\ba; \cD_{\xi}, \mathsf{D}_{\chi'})}(E) - p_{H^{T}_{\bar{\cA}}(\ba; \mathsf{D}_{\chi'})}(E)} \leq T(\delta + \epsilon). 
\end{align*}
\end{proof}

The following lemma obtains the analogous bound as Lemma~\ref{lemma:fresh-randomness} when the same random seed $\xi$ is reused across iterations. 

\begin{restatable}{lemma}{reuserandomness}\label{lemma:reuse-randomness}  Let $\cA'_{\xi, \chi'}$ be a randomized algorithm which takes an input $x \in \R^d$ and two random seeds $\xi \sim \cD_{\xi}, \chi' \sim {\cD_{\chi'}}^{\ba}$. Let $\bar{\cA}_{\chi'}$ be a $(\epsilon
, \delta)$-smoothing of $\cA'$ with respect to $\xi$. Let $s \sim \cD_{\xi}$. Then,
\begin{align*}
    \tv{p_{\Phi^T_{\cA'}(\ba; s, \mathsf{D}_{\chi'})}}{p_{H_{\bar{\cA}}^T(\ba; \mathsf{D}_{\chi'})}} \leq T(\delta + \epsilon). 
\end{align*}
\end{restatable}
\begin{proof} By Fact~\ref{lemma:tv_decomposition} and Fact~\ref{fact:transformation}, with probability $1-T\delta$ over the draw of $s$, there exist random variables  $C^t$ and events $E_t$ for $t = \{1, ..., T\}$ so that $\prob{E_t} \geq 1-\epsilon$, and conditioned on $E_1, ..., E_t$, 
\begin{align*}
    \zeta(\ba, \cA'_{\xi=s, \chi'}(\ba)) &\overset{\cD}{=} C^1 \overset{\cD}{=} \zeta(\ba, \bar{\cA}_{\chi'}(\ba)),
\end{align*}
and
\begin{align*}
    \Phi_{{\cA'}}^t(\ba; s, \mathsf{D}_{\chi'}) &\overset{\cD}{=} \zeta(C^{t-1}, \cA_{\xi=s, \chi'}(C^{t-1})) \overset{\cD}{=} C^t, \\
    H_{\bar{\cA}}^t(\ba; \mathsf{D}_{\chi'}) &\overset{\cD}{=} \zeta(C^{t-1}, \bar{\cA
    }_{\chi'}(C^{t-1})) \overset{\cD}{=} C^t. 
\end{align*}
for every $t \in [T]$. So, let $E' \defeq \wedge_{t \in [T]} E_t$, and note that $\prob{E'} \geq 1-\epsilon T$. Consequently, 
\begin{align*}
    \tv{p_{\Phi^T(x; s, \mathsf{D}_{\chi'})}}{p_{H^T(x; \mathsf{D}_{\chi'})}} \leq T\delta + (1 - \prob{E'}) &\leq T(\delta + \epsilon). 
\end{align*}
\end{proof}

\pseudoindependencemain*
\begin{proof}
    The result follows directly by applying triangle inequality, Lemma~\ref{lemma:reuse-randomness}, and Lemma~\ref{lemma:fresh-randomness}. 
\end{proof}

%% file: main_file.bbl
\begin{thebibliography}{46}
\providecommand{\natexlab}[1]{#1}
\providecommand{\url}[1]{\texttt{#1}}
\expandafter\ifx\csname urlstyle\endcsname\relax
  \providecommand{\doi}[1]{doi: #1}\else
  \providecommand{\doi}{doi: \begingroup \urlstyle{rm}\Url}\fi

\bibitem[Agarwal and Bottou(2015)]{agarwal2015lower}
Alekh Agarwal and Leon Bottou.
\newblock A lower bound for the optimization of finite sums.
\newblock In \emph{\cICML{2015}}, 2015.

\bibitem[Allen-Zhu et~al.(2014)Allen-Zhu, Liao, and Yuan]{allen2014optimization}
Zeyuan Allen-Zhu, Zhenyu Liao, and Yang Yuan.
\newblock Optimization algorithms for faster computational geometry.
\newblock In \emph{\cICALP{2014}}, 2014.

\bibitem[Asi and Duchi(2020)]{asi2020near}
Hilal Asi and John~C Duchi.
\newblock Near instance-optimality in differential privacy.
\newblock In \emph{\cNIPS{2020}}, 2020.

\bibitem[Beimel et~al.(2022)Beimel, Kaplan, Mansour, Nissim, Saranurak, and Stemmer]{beimel2022dynamic}
Amos Beimel, Haim Kaplan, Yishay Mansour, Kobbi Nissim, Thatchaphol Saranurak, and Uri Stemmer.
\newblock Dynamic algorithms against an adaptive adversary: generic constructions and lower bounds.
\newblock In \emph{\cSTOC{2022}}, 2022.

\bibitem[Bertsekas(2011)]{bertsekas2011approximate}
Dimitri~P Bertsekas.
\newblock Approximate policy iteration: A survey and some new methods.
\newblock In \emph{Journal of Control Theory and Applications}, 2011.

\bibitem[Braverman et~al.(2023)Braverman, Krauthgamer, Krishnan, and Sapir]{braverman2023lower}
Vladimir Braverman, Robert Krauthgamer, Aditya Krishnan, and Shay Sapir.
\newblock Lower bounds for pseudo-deterministic counting in a stream.
\newblock \emph{arXiv preprint arXiv:2303.16287}, 2023.

\bibitem[Carmon et~al.(2019)Carmon, Jin, Sidford, and Tian]{carmon2019variance}
Yair Carmon, Yujia Jin, Aaron Sidford, and Kevin Tian.
\newblock Variance reduction for matrix games.
\newblock \emph{\cNIPS{2019}}, 2019.

\bibitem[Chen et~al.(2022)Chen, Kyng, Liu, Peng, Gutenberg, and Sachdeva]{chen2022maximum}
Li~Chen, Rasmus Kyng, Yang~P Liu, Richard Peng, Maximilian~Probst Gutenberg, and Sushant Sachdeva.
\newblock Maximum flow and minimum-cost flow in almost-linear time.
\newblock In \emph{\cFOCS{2022}}, 2022.

\bibitem[Cherapanamjeri and Nelson(2020)]{cherapanamjeri2020adaptive}
Yeshwanth Cherapanamjeri and Jelani Nelson.
\newblock On adaptive distance estimation.
\newblock In \emph{Advances in Neural Information Processing Systems}, 2020.

\bibitem[Cohen et~al.(2024)Cohen, Nelson, Sarl{\'o}s, Singhal, and Stemmer]{cohen2024one}
Edith Cohen, Jelani Nelson, Tam{\'a}s Sarl{\'o}s, Mihir Singhal, and Uri Stemmer.
\newblock One attack to rule them all: Tight quadratic bounds for adaptive queries on cardinality sketches.
\newblock In \emph{arXiv preprint arXiv:2411.06370}, 2024.

\bibitem[Cohen et~al.(2015)Cohen, Lee, Musco, Musco, Peng, and Sidford]{cohen2015uniform}
Michael~B Cohen, Yin~Tat Lee, Cameron Musco, Christopher Musco, Richard Peng, and Aaron Sidford.
\newblock Uniform sampling for matrix approximation.
\newblock In \emph{\cITCS{2015}}, 2015.

\bibitem[Cohen et~al.(2020{\natexlab{a}})Cohen, Lee, and Song]{cohen2020solving}
Michael~B. Cohen, Yin~Tat Lee, and Zhao Song.
\newblock Solving linear programs in the current matrix multiplication time.
\newblock In \emph{Journal of the ACM}, 2020{\natexlab{a}}.

\bibitem[Cohen et~al.(2020{\natexlab{b}})Cohen, Musco, and Pachocki]{cohen2020online}
Michael~B Cohen, Cameron Musco, and Jakub Pachocki.
\newblock Online row sampling.
\newblock In \emph{Theory of Computing}, 2020{\natexlab{b}}.

\bibitem[Dixon et~al.(2022)Dixon, Pavan, Woude, and Vinodchandran]{dixon2022pseudodeterminism}
Peter Dixon, Aduri Pavan, Jason~Vander Woude, and NV~Vinodchandran.
\newblock Pseudodeterminism: promises and lowerbounds.
\newblock In \emph{\cSTOC{2022}}, 2022.

\bibitem[Fischetti et~al.(2018)Fischetti, Mandatelli, and Salvagnin]{fischetti2018faster}
Matteo Fischetti, Iacopo Mandatelli, and Domenico Salvagnin.
\newblock Faster sgd training by minibatch persistency.
\newblock In \emph{arXiv preprint arXiv:1806.07353}, 2018.

\bibitem[Frostig et~al.(2015)Frostig, Ge, Kakade, and Sidford]{frostig2015regularizing}
Roy Frostig, Rong Ge, Sham Kakade, and Aaron Sidford.
\newblock Un-regularizing: approximate proximal point and faster stochastic algorithms for empirical risk minimization.
\newblock In \emph{\cICML{2015}}, 2015.

\bibitem[Garber et~al.(2016)Garber, Hazan, Jin, Musco, Netrapalli, Sidford, et~al.]{garber2016faster}
Dan Garber, Elad Hazan, Chi Jin, Cameron Musco, Praneeth Netrapalli, Aaron Sidford, et~al.
\newblock Faster eigenvector computation via shift-and-invert preconditioning.
\newblock In \emph{\cICML{2016}}, 2016.

\bibitem[Gat and Goldwasser(2011)]{Gat2011ProbabilisticSA}
Erann Gat and Shafi Goldwasser.
\newblock Probabilistic search algorithms with unique answers and their cryptographic applications.
\newblock In \emph{Electronic Colloquium on Computational Complexity: ECCC}, 2011.

\bibitem[Ghazi et~al.(2022)Ghazi, Kumar, Manurangsi, and Nelson]{ghazi2022differentially}
Badih Ghazi, Ravi Kumar, Pasin Manurangsi, and Jelani Nelson.
\newblock Differentially private all-pairs shortest path distances: Improved algorithms and lower bounds.
\newblock In \emph{\cSODA{2022}}, 2022.

\bibitem[Goldwasser and Grossman(2017)]{goldwasser2017bipartite}
Shafi Goldwasser and Ofer Grossman.
\newblock Bipartite perfect matching in pseudo-deterministic nc.
\newblock In \emph{\cICALP{2017}}, 2017.

\bibitem[Goldwasser et~al.(2017)Goldwasser, Grossman, and Holden]{goldwasser2017pseudo}
Shafi Goldwasser, Ofer Grossman, and Dhiraj Holden.
\newblock Pseudo-deterministic proofs.
\newblock In \emph{arXiv preprint arXiv:1706.04641}, 2017.

\bibitem[Grondman et~al.(2012)Grondman, Busoniu, Lopes, and Babuska]{grondman2012survey}
Ivo Grondman, Lucian Busoniu, Gabriel~AD Lopes, and Robert Babuska.
\newblock A survey of actor-critic reinforcement learning: Standard and natural policy gradients.
\newblock In \emph{IEEE Transactions on Systems, Man, and Cybernetics, part C (applications and reviews)}, 2012.

\bibitem[Grossman et~al.(2023)Grossman, Gupta, and Sellke]{grossman2023tight}
Ofer Grossman, Meghal Gupta, and Mark Sellke.
\newblock Tight space lower bound for pseudo-deterministic approximate counting.
\newblock In \emph{\cFOCS{2023}}, 2023.

\bibitem[Impagliazzo et~al.(2022)Impagliazzo, Lei, Pitassi, and Sorrell]{impagliazzo2022reproducibility}
Russell Impagliazzo, Rex Lei, Toniann Pitassi, and Jessica Sorrell.
\newblock Reproducibility in learning.
\newblock In \emph{Proceedings of the 54th annual ACM SIGACT symposium on theory of computing}, 2022.

\bibitem[Jiang et~al.(2021)Jiang, Song, Weinstein, and Zhang]{jiang2021faster}
Shunhua Jiang, Zhao Song, Omri Weinstein, and Hengjie Zhang.
\newblock A faster algorithm for solving general lps.
\newblock In \emph{\cSTOC{2021}}, 2021.

\bibitem[Jin and Sidford(2021)]{jin2021towards}
Yujia Jin and Aaron Sidford.
\newblock Towards tight bounds on the sample complexity of average-reward mdps.
\newblock In \emph{\cICML{2021}}, 2021.

\bibitem[Jin et~al.(2022)Jin, Sidford, and Tian]{jin2022sharper}
Yujia Jin, Aaron Sidford, and Kevin Tian.
\newblock Sharper rates for separable minimax and finite sum optimization via primal-dual extragradient methods.
\newblock In \emph{\cCOLT{2022}}, 2022.

\bibitem[Jin et~al.(2024)Jin, Karmarkar, Sidford, and Wang]{jin2024truncated}
Yujia Jin, Ishani Karmarkar, Aaron Sidford, and Jiayi Wang.
\newblock Truncated variance reduced value iteration.
\newblock In \emph{arXiv preprint arXiv:2405.12952}, 2024.

\bibitem[Johnson and Zhang(2013)]{johnson2013accelerating}
Rie Johnson and Tong Zhang.
\newblock Accelerating stochastic gradient descent using predictive variance reduction.
\newblock In \emph{\cNIPS{2013}}, 2013.

\bibitem[Kearns and Singh(1998)]{kearns1998finite}
Michael Kearns and Satinder Singh.
\newblock Finite-sample convergence rates for q-learning and indirect algorithms.
\newblock In \emph{\cNIPS{2020}}, 1998.

\bibitem[Lee and Sidford(2014)]{lee2014path}
Yin~Tat Lee and Aaron Sidford.
\newblock Path finding methods for linear programming: Solving linear programs in o (vrank) iterations and faster algorithms for maximum flow.
\newblock In \emph{\cFOCS{2014}}, 2014.

\bibitem[Lee and Sidford(2015)]{lee2015efficient}
Yin~Tat Lee and Aaron Sidford.
\newblock Efficient inverse maintenance and faster algorithms for linear programming.
\newblock In \emph{\cFOCS{2015}}, 2015.

\bibitem[Lin et~al.(2015)Lin, Mairal, and Harchaoui]{lin2015universal}
Hongzhou Lin, Julien Mairal, and Zaid Harchaoui.
\newblock A universal catalyst for first-order optimization.
\newblock In \emph{\cNIPS{2015}}, 2015.

\bibitem[Littman et~al.(1995)Littman, Dean, and Kaelbling]{littman2013complexity}
Michael~L Littman, Thomas~L Dean, and Leslie~Pack Kaelbling.
\newblock On the complexity of solving markov decision problems.
\newblock In \emph{\cUAI{1995}}, 1995.

\bibitem[Liu et~al.(2023)Liu, Zhao, Xu, Yue, and Fang]{liu2023accelerated}
Yuanshi Liu, Hanzhen Zhao, Yang Xu, Pengyun Yue, and Cong Fang.
\newblock Accelerated gradient algorithms with adaptive subspace search for instance-faster optimization.
\newblock In \emph{arXiv preprint arXiv:2312.03218}, 2023.

\bibitem[Nemirovski(2004)]{nemirovski2004prox}
Arkadi Nemirovski.
\newblock Prox-method with rate of convergence o (1/t) for variational inequalities with lipschitz continuous monotone operators and smooth convex-concave saddle point problems.
\newblock In \emph{SIAM Journal on Optimization}, 2004.

\bibitem[Roch(2024)]{roch_mdp_2024}
Sebastien Roch.
\newblock Modern discrete probability: An essential toolkit.
\newblock In \emph{Cambridge Series in Statistical and Probabilistic Mathematics}, 2024.

\bibitem[Sallinen et~al.(2016)Sallinen, Satish, Smelyanskiy, Sury, and R{\'e}]{sallinen2016high}
Scott Sallinen, Nadathur Satish, Mikhail Smelyanskiy, Samantika~S Sury, and Christopher R{\'e}.
\newblock High performance parallel stochastic gradient descent in shared memory.
\newblock In \emph{2016 IEEE International Parallel and Distributed Processing Symposium (IPDPS)}, 2016.

\bibitem[Sidford et~al.(2018{\natexlab{a}})Sidford, Wang, Wu, Yang, and Ye]{sidford2018near}
Aaron Sidford, Mengdi Wang, Xian Wu, Lin Yang, and Yinyu Ye.
\newblock Near-optimal time and sample complexities for solving markov decision processes with a generative model.
\newblock In \emph{\cNIPS{2017}}, 2018{\natexlab{a}}.

\bibitem[Sidford et~al.(2018{\natexlab{b}})Sidford, Wang, Wu, and Ye]{SWWY18}
Aaron Sidford, Mengdi Wang, Xian Wu, and Yinyu Ye.
\newblock Variance reduced value iteration and faster algorithms for solving markov decision processes.
\newblock In \emph{\cSODA{2018}}, 2018{\natexlab{b}}.

\bibitem[Sidford et~al.(2023)Sidford, Wang, Wu, and Ye]{sidford2023variance}
Aaron Sidford, Mengdi Wang, Xian Wu, and Yinyu Ye.
\newblock Variance reduced value iteration and faster algorithms for solving markov decision processes.
\newblock In \emph{Naval Research Logistics (NRL)}, 2023.

\bibitem[Tseng(1990)]{tseng1990solving}
Paul Tseng.
\newblock Solving h-horizon, stationary markov decision problems in time proportional to log (h).
\newblock In \emph{Operations Research Letters}, 1990.

\bibitem[van~den Brand et~al.(2021)van~den Brand, Lee, Liu, Saranurak, Sidford, Song, and Wang]{brand2021}
Jan van~den Brand, Yin~Tat Lee, Yang~P. Liu, Thatchaphol Saranurak, Aaron Sidford, Zhao Song, and Di~Wang.
\newblock Minimum cost flows, mdps, and l1-regression in nearly linear time for dense instances.
\newblock In \emph{\cSTOC{2021}}, 2021.

\bibitem[van~den Brand et~al.(2022)van~den Brand, Gao, Jambulapati, Lee, Liu, Peng, and Sidford]{brandmaxflow}
Jan van~den Brand, Yu~Gao, Arun Jambulapati, Yin~Tat Lee, Yang~P. Liu, Richard Peng, and Aaron Sidford.
\newblock Faster maxflow via improved dynamic spectral vertex sparsifiers.
\newblock In \emph{\cSTOC{2022}}, 2022.

\bibitem[Woodworth and Srebro(2016)]{woodworth2016tight}
Blake~E Woodworth and Nati Srebro.
\newblock Tight complexity bounds for optimizing composite objectives.
\newblock In \emph{\cNIPS{2016}}, 2016.

\bibitem[Woodworth et~al.(2020)Woodworth, Patel, and Srebro]{woodworth2020minibatch}
Blake~E Woodworth, Kumar~Kshitij Patel, and Nati Srebro.
\newblock Minibatch vs local sgd for heterogeneous distributed learning.
\newblock In \emph{Advances in Neural Information Processing Systems}, 2020.

\end{thebibliography}
